\documentclass[11pt]{article}

\def\colorful{0}

\usepackage{amsthm,amsfonts,amsmath,amssymb,epsfig,color,float,graphicx,verbatim}
\usepackage{multirow}

\newif\ifhyper\IfFileExists{hyperref.sty}{\hypertrue}{\hyperfalse}
\hypertrue
\ifhyper\usepackage{hyperref}\fi

\usepackage{enumitem}
\usepackage[capitalize]{cleveref} 

\makeatletter
\renewcommand{\section}{\@startsection{section}{1}{0pt}{-12pt}{5pt}{\large\bf}}
\renewcommand{\subsection}{\@startsection{subsection}{2}{0pt}{-12pt}{-5pt}{\normalsize\bf}}
\renewcommand{\subsubsection}{\@startsection{subsubsection}{3}{0pt}{-12pt}{-5pt}{\normalsize\bf}}
\makeatother

\usepackage{framed}
\usepackage{nicefrac}

\def\nnewcolor{1}
\ifnum\nnewcolor=1

\fi
\ifnum\nnewcolor=0

\fi

\ifnum\colorful=1

\else

\fi

\newtheorem{theorem}{Theorem}[section]

\newtheorem{lemma}[theorem]{Lemma}
\newtheorem{proposition}[theorem]{Proposition}
\newtheorem{corollary}[theorem]{Corollary}
\newtheorem{claim}[theorem]{Claim}
\newtheorem{fact}[theorem]{Fact}

\newtheorem{observation}[theorem]{Observation}
\theoremstyle{definition}
\newtheorem{definition}[theorem]{Definition}

\newcommand{\R}{\mathbb{R}}
\newcommand{\C}{\mathbb{C}}
\newcommand{\Z}{\mathbb{Z}}
\newcommand{\E}{\mathbb{E}}

\newcommand{\poly}{\mathrm{poly}}
\newcommand{\polylog}{\mathrm{polylog}}
\newcommand{\sinc}{\mathrm{Sinc}}

\newcommand{\p}{\mathbf{P}}
\newcommand{\q}{\mathbf{Q}}
\newcommand{\h}{\mathbf{H}}

\newcommand{\dtv}{d_{\mathrm TV}}
\newcommand{\dk}{d_{\mathrm K}}

\newcommand{\wh}[1]{{\widehat{#1}}}

\newcommand{\var}{\mathrm{Var}}

\newcommand{\ignore}[1]{}

\newcommand{\eps}{\epsilon}

\newcommand{\ba}{\mathbf{a}}

\newcommand{\bb}{\mathbf{b}}

\newcommand{\Scal}{\mathcal{S}}

\newcommand{\normal}{\mathrm{N}}
\newcommand{\round}[1]{\lfloor #1 \rceil}
\newcommand{\discnorm}[2]{Z({#1},{#2})}

\newcommand{\Var}{\mathop{\textnormal{Var}}\nolimits}

\renewcommand{\eqref}[1]{Eq.~(\ref{#1})}

\newcommand{\eqdef}{\stackrel{{\mathrm {\footnotesize def}}}{=}}

\newcommand{\littlesum}{\mathop{\textstyle \sum}}
\newcommand{\littleprod}{\mathop{\textstyle \prod}}

\newenvironment{algorithm}[1][\  ] %
{ \rm
\begin{tabbing}
....\=.....\=.....\=.....\=.....\=  \+ \kill
} %
{\end{tabbing} }

\floatstyle{ruled}
\newfloat{Algorithm}{H}{alg}
\newfloat{Subroutine}{H}{sub}

{
\begin{minipage}{1.0\linewidth} \begin{algorithm} %
} { \end{algorithm} \end{minipage} }

\title{Optimal Learning via the Fourier Transform\\ 
for Sums of Independent Integer Random Variables}

\author{
Ilias Diakonikolas\thanks{Supported by EPSRC grant EP/L021749/1 and a Marie Curie Career Integration grant.}\\
University of Edinburgh\\
{\tt ilias.d@ed.ac.uk}.\\
\and
Daniel M. Kane\thanks{Some of this work was performed while visiting the University of Edinburgh.}\\
University of California, San Diego\\
{\tt dakane@cs.ucsd.edu}.\\
\and
Alistair Stewart\thanks{Supported by EPSRC grant EP/L021749/1.}\\
University of Edinburgh\\
{\tt stewart.al@gmail.com}.
}

\begin{document}

\maketitle

\thispagestyle{empty}

\begin{abstract}
We study the structure and learnability of sums of independent integer random variables (SIIRVs).
For $k \in \Z_{+}$, a {\em$k$-SIIRV of order $n \in \Z_{+}$} is the probability distribution of the sum of $n$
mutually independent random variables each supported on $\{0, 1, \dots, k-1\}$.
We denote by ${\cal S}_{n,k}$ the set of all $k$-SIIRVs of order $n$.

How many samples are required to learn an arbitrary distribution in ${\cal S}_{n,k}$?
In this paper, we tightly characterize the sample and computational complexity of this problem.
More precisely, we design a computationally efficient algorithm that uses $\widetilde{O}(k/\eps^2)$ samples, 
and learns an arbitrary $k$-SIIRV within error $\eps,$ in total variation distance. Moreover, we show that
the {\em optimal} sample complexity of this learning problem is 
$\Theta((k/\eps^2)\sqrt{\log(1/\eps)}),$ i.e., we prove an upper bound and a matching 
information-theoretic lower bound.
Our algorithm proceeds by learning the Fourier transform of the target $k$-SIIRV in its effective support. 
Its correctness relies on the {\em approximate sparsity} of the Fourier transform of $k$-SIIRVs -- 
a structural property that we establish, roughly stating that the Fourier transform of $k$-SIIRVs
has small magnitude outside a small set.

Along the way we prove several new structural results about $k$-SIIRVs.
As one of our main structural contributions, we give an efficient algorithm to construct a 
sparse {\em proper} $\eps$-cover for ${\cal S}_{n,k},$ in total variation distance.
We also obtain a novel geometric characterization of the space of $k$-SIIRVs. Our
characterization allows us to prove a tight lower bound on the size of $\eps$-covers for ${\cal S}_{n,k}$
-- establishing that our cover upper bound is optimal -- and is the key ingredient in our tight sample complexity lower bound.

Our approach of exploiting the sparsity of the Fourier transform in 
distribution learning is general, and has recently found additional applications. 
In a subsequent work~\cite{DKS15c}, we use a generalization of this idea (in higher dimensions)
to obtain the first efficient learning algorithm for Poisson multinomial distributions.
In~\cite{DKS15b}, we build on this approach to obtain the fastest known proper learning algorithm 
for Poisson binomial distributions ($2$-SIIRVs).
\end{abstract}

\thispagestyle{empty}
\setcounter{page}{0}

\newpage

\section{Introduction}  \label{sec:intro}

\subsection{Motivation and Background} \label{ssec:background}
We study sums of independent integer random variables:

\smallskip

\noindent {\bf Definition.}  
For $k \in \Z_{+}$, a \emph{$k$-IRV} is any random variable supported on
$\{0, 1, \dots, k-1\}$. A {\em$k$-SIIRV of order $n$} is any random variable
$X = \sum_{i=1}^n X_i$ where the $X_i$'s are independent $k$-IRVs.
We will denote by ${\cal S}_{n,k}$ the set of probability distributions of all
$k$-SIIRVs of order $n$.

\smallskip

For convenience, throughout this paper, we will often blur the distinction between a random variable and its distribution.
In particular, we will use the term $k$-SIIRV for the random variable or its corresponding distribution, and the distinction
will be clear from the context.

\smallskip

Sums of independent integer random variables (SIIRVs) comprise a rich class of distributions that
arise in many settings. The special case of $k=2$, ${\cal S}_{n,2}$,
was first considered by Poisson \cite{Poisson:37} as a non-trivial extension of the Binomial distribution,
and is known as Poisson binomial distribution (PBD). In application domains, SIIRVs have many uses in research areas
such as survey sampling, case-control studies, and survival analysis, see e.g., \cite{ChenLiu:97} for a survey of the many practical uses of these distributions.
We remark that these distributions are of fundamental interest and have been extensively studied in probability and statistics.
For example, tail bounds on SIIRVs form an important special case of Chernoff/Hoeffding bounds~\cite{Chernoff:52,Hoeffding:63,DP09}.
Moreover, there is a long line of research on approximate limit theorems for
SIIRVs, dating back
several decades (see e.g.,~\cite{Presman:83,Kruopis:86,BHJ:92}), and~\cite{CL10,CGS11} for some recent results.


\paragraph{Structure and Learning of $k$-SIIRVs.}

The main motivation of this work was the problem of learning an unknown $k$-SIIRV given access to independent samples.
Understanding this problem is intimately related to obtaining a refined structural understanding of the space of $k$-SIIRVs.
The connection between structure and distribution learning is the main thrust of this paper.

Distribution learning or {\em density estimation} is the following task~\cite{DG85, KMR+:94short, DL:01}:
Given independent samples from an unknown distribution $\p$ in a family $\mathcal{D},$ and an error tolerance $\eps>0,$
output a hypothesis $\h$ such that with high probability the total variation distance $\dtv(\h, \p)$ is at most~$\eps.$
The sample and computational complexity of this unsupervised learning problem
depends on the {\em structure} of the underlying family $\mathcal{D}.$
The main goals here are: (i) to characterize the {\em sample complexity} of the learning problem, i.e., to obtain
matching information-theoretic upper and lower bounds, and (ii) to design a {\em computationally efficient} learning algorithm --  i.e., 
an algorithm whose running time is polynomial in the sample (input) size -- that uses an information-theoretically optimal sample size.

While density estimation has been studied for several decades,
the number of samples required to learn is not yet well understood,
even for surprisingly simple and natural classes of univariate discrete distributions.
More specifically, there is no known complexity measure of a distribution family $\mathcal{D}$
that {\em characterizes} the sample complexity of learning an unknown distribution from $\mathcal{D}.$
In contrast, the VC dimension of a concept class plays such a role in the PAC model of
learning Boolean functions (see, e.g,~\cite{BEH+:89, KearnsVazirani:94}).

We remark that the classical information-theoretic quantity
of the {\em metric entropy}~\cite{VWellner96, DL:01, Tsybakov08},
i.e., the logarithm of the size of the smallest $\eps$-cover\footnote{Formally, a subset ${\cal D}_{\eps} \subseteq {\cal D}$ in a metric space
$({\cal D}, d)$ is said to be an \emph{$\eps$-cover of ${\cal D}$} with respect to the metric $d: \mathcal{X}^2 \to \R_+,$
if for every $\mathbf{x} \in {\cal D}$ there exists some $\mathbf{y} \in {\cal D}_{\eps}$ such that $d(\mathbf{x}, \mathbf{y}) \leq \eps.$ 
In this paper, we focus on the total variation distance between distributions.} of the distribution class, provides an {\em upper bound}
on the sample complexity of learning. 
Alas, this upper bound is suboptimal in general -- both quantitatively and
qualitatively -- and in particular for the class of $k$-SIIRVs, as we show in this paper.

{
Obtaining a computationally efficient learning algorithm with optimal (or near-optimal) sample complexity
is an important goal. In many learning settings, achieving this goal turns out to be quite challenging.
More specifically, in many scenarios, both supervised and unsupervised, the only computationally efficient learning algorithms known use a (provably) suboptimal
sample size. Intuitively, increasing the sample size (e.g., by a polynomial factor)  can make the algorithmic task substantially easier. 
Characterizing the tradeoff between sample complexity and computational complexity is of 
fundamental importance in learning theory. In this work, we essentially characterize this tradeoff for the unsupervised problem of learning 
SIIRVs.}



\subsection{Our Results} \label{ssec:results}
The main technical contribution of this paper is the use of Fourier analytic and geometric tools
to obtain a refined structural understanding of the space of $k$-SIIRVs. As a byproduct of our techniques, we
characterize the sample complexity of learning $k$-SIIRVs (up to constant factors),
and moreover we obtain a computationally efficient learning algorithm with near-optimal sample complexity.
Our results answer the main open questions of~\cite{DDS12stoc, DDOST13focs}.

Along the way we prove several new structural results of independent interest about $k$-SIIRVs,
including: the approximate sparsity of their Fourier transform; tight upper and lower bounds
on $\eps$-covers (in total variation distance and Kolmogorov distance);
and a novel geometric characterization of the space of $k$-SIIRVs, that is crucial for our sample complexity lower bound.
Below, we state our results in detail and elaborate on their context and the connections between them.

\paragraph{Learning SIIRVs via the Fourier Transform.}
As our first result, we give a sample near-optimal and computationally efficient learning algorithm for $k$-SIIRVs:

\begin{theorem} [Nearly Optimal Learning of $k$-SIIRVs]  \label{thm:learning-informal}
There is a learning algorithm for $k$-SIIRVs
with the following performance guarantee: Let $\p$ be any $k$-SIIRV of order $n.$
The algorithm uses $\widetilde{O}(k/\eps^2)$ samples from~$\p,$
runs in time\footnote{We work in the standard ``word RAM'' model in which basic arithmetic
operations on $O(\log n)$-bit integers are assumed to take constant time.} 
$\widetilde{O}(k^3/\eps^2),$ and with probability at least $2/3$ outputs a (succinct description of a) hypothesis
$\h$ such that $\dtv (\h, \p) \leq \eps.$
\end{theorem}

{
Our algorithm outputs a succinct description of the hypothesis $\h,$ via its Discrete Fourier Transform (DFT) $\wh{\h},$
which is supported on a set of small cardinality. The DFT immediately gives a fast evaluation oracle for $\h.$ We also show
how to use the DFT, in a black-box manner, to obtain an efficient approximate sampler for the target distribution $\p.$
Our efficient learning algorithm is described in Section~\ref{ssec:learn-efficient}.
In Section~\ref{ssec:sampler} we give the efficient construction of our sampler.
}

We remark that the sample complexity of our algorithm is optimal up to logarithmic factors.
Indeed, even learning a single $k$-IRV to variation distance $\eps$ requires $\Omega(k/\eps^2)$ samples.
For the case of $k=2,$~\cite{DDS12stoc}
gave a learning algorithm that uses $\widetilde{O}(1/\eps^2)$ samples, but runs in quasi-polynomial time, namely $(1/\eps)^{\polylog(1/\eps)}.$
More recently,~\cite{DDOST13focs} studied the case of general $k,$ and
gave an algorithm that uses  $\poly(k/\eps)$ samples and time.
Notably, the degree of this polynomial is quite high: the sample complexity of the~\cite{DDOST13focs} algorithm is $\Omega(k^9/\eps^6).$
Theorem~\ref{thm:learning-informal} gives a nearly-tight upper bound on the sample complexity of this
learning problem, and does so with a computationally efficient algorithm.

Given our $\widetilde{O}(k/\eps^2)$ sample upper bound, it would be tempting to conjecture that $\Theta (k/\eps^2)$
is in fact the optimal sample complexity of learning $k$-SIIRVs. 
If true, this would imply that learning a $k$-SIIRV is as easy as learning a $k$-IRV.
Surprisingly, we show that this is not the case:

\begin{theorem} [Optimal Sample Complexity]  \label{thm:sample-complexity}
For any $k \in \Z_+,$ $\eps \leq 1/ \poly(k),$ there is an algorithm that learns $k$-SIIRVs within variation distance $\eps$
using $O((k/\eps^2) \sqrt{\log(1/\eps)})$ samples. Moreover, any algorithm for this problem information-theoretically requires
$\Omega((k/\eps^2) \sqrt{\log(1/\eps)})$ samples.
\end{theorem}

Theorem~\ref{thm:sample-complexity} precisely characterizes the sample complexity of learning $k$-SIIRVs (up to constant factors)
by giving an upper bound and a matching information-theoretic sample lower bound.
The sharp sample complexity bound of $\Theta((k/\eps^2) \sqrt{\log(1/\eps)})$ is surprising, 
and cannot be obtained using standard information-theoretic tools (e.g., metric entropy).
We elaborate on this issue in Section~\ref{sec:techniques}.

{
We remark that the upper bound of Theorem~\ref{thm:sample-complexity} does not specify the running 
time of the corresponding algorithm. This is because the simplest such algorithm actually runs in time
exponential in $k.$ For the important special case of $k=2,$ we obtain a sample--optimal learning algorithm that runs in sample--linear time:

\begin{theorem} [Optimal Learning of PBDs ($2$-SIIRVs)]  \label{thm:pbd-learn-opt}
For any $\eps>0,$ there is an algorithm that learns PBDs within variation distance $\eps$
using $O((1/\eps^2) \sqrt{\log(1/\eps)})$ samples and running in time $O((1/\eps^2) \sqrt{\log(1/\eps)}).$
\end{theorem}

The upper bound of Theorem~\ref{thm:sample-complexity} and Theorem~\ref{thm:pbd-learn-opt} are established
in Section~\ref{sec:opt-sample}. Our tight sample complexity lower bound is proved in Section~\ref{sec:sample-lb}.
}

\paragraph{Using the Fourier Transform for Distribution Learning.}
Our learning upper bounds are obtained via an approach which is novel in this context.
Specifically, we show that the Fourier transform of $k$-SIIRVs is {\em approximately sparse}, 
and exploit this property to learn the distribution {\em via learning its Fourier transform in its effective support}.
The sparsity of the Fourier transform explains why this family of distributions is learnable with sample complexity independent of $n,$
and moreover it yields the sharp sample-complexity bound.
The algorithmic idea of exploiting Fourier sparsity for distribution learning is general (see Section~\ref{ssec:fourier-algo-gen}), 
and was subsequently used by the authors in other related settings~\cite{DKS15c, DKS15b}.

\paragraph{Structure of $k$-SIIRVs.} Our core structural result is the following simple property
of the Fourier transform of $k$-SIIRVs:
\begin{center}
{\em Any $k$-SIIRV with ``large'' variance has a Fourier transform with ``small'' effective support.}
\end{center}

One can obtain different versions of the above informal statement depending on the setting and the desired application.
See Lemma~\ref{FourierSupportLem} for a formal statement in the context of the DFT. 
The Fourier sparsity of $k$-SIIRVs forms the basis for our upper bounds in this paper. 
As previously mentioned, this structural property motivates and enables our learning algorithm. 
Moreover, it is useful in order to obtain sparse $\eps$-covers for $\mathcal{S}_{n, k},$ the space of $k$-SIIRVs,
under the total variation distance. 

More specifically, using the approximate sparsity of the Fourier transform of SIIRVs combined with analytic arguments, 
we obtain a computationally efficient algorithm to construct a {\em proper} $\eps$-cover for $\mathcal{S}_{n, k},$ of near-minimum size.
In particular, we show:

\begin{theorem}[Optimal Covers for $k$-SIIRVs] \label{thm:cover-informal}
For $\epsilon \leq 1/k$, there exists a proper $\eps$-cover ${\cal S}_{n, k, \eps} \subseteq {\cal S}_{n, k}$ of ${\cal S}_{n, k}$
under the total variation distance of size $|{\cal S}_{n, k, \epsilon}| \le n \cdot (1/\eps)^{O( k\log(1/\eps))}$
that can be constructed in polynomial time.
\end{theorem}

The best previous upper bound on the cover size of $2$-SIIRVs is
$n^2 + n \cdot (1/\eps)^{O(\log^2(1/\eps))}$~\cite{DaskalakisP09, DP:cover}.
For $k>2$,~\cite{DDOST13focs} gives a {\em non-proper} cover of size $n \cdot 2^{\poly(k/\eps)}.$

Our proper cover upper bound construction provides a smaller search space
for essentially any optimization problem over $k$-SIIRVs.
Specifically, Theorem~\ref{thm:cover-informal} has the following implication in computational game theory:
Via a connection established in~\cite{DaskalakisP07, DaskalakisP09},
the proper cover construction of Theorem~\ref{thm:cover-informal} (for $k=2$)
yields an improved $\poly(n) \cdot (1/\eps)^{O(\log(1/\eps))}$ time
algorithm for computing  $\eps$-Nash equilibria in anonymous games with $2$ strategies per player.
Our matching lower bound on the cover size implies that the ``cover-based approach''
cannot lead to an FPTAS for this problem. We note that computing an (exact) Nash equilibrium
in an anonymous game with a constant number of strategies was recently shown to be intractable~\cite{CDO15}.
Our cover upper bound is proved in Section~\ref{sec:cover-ub}.

{
We also prove a matching lower bound on the cover size, showing that our above construction is essentially optimal:

\begin{theorem}[Cover Size Lower Bound for $k$-SIIRVs] \label{thm:cover-lb-informal}
For $\eps \le 1/\poly(k),$ and $n  = \Omega(\log(1/\eps)),$
any $\eps$-cover for ${\cal S}_{n, k}$ has size at least $n \cdot (1/\eps)^{\Omega(k\log(1/\eps))}.$
\end{theorem}

}

Before our work, no non-trivial lower bound on the cover size was known.
We view the inherent quasi-polynomial dependence on $1/\eps$ of the cover size established here
as a rather surprising fact.
Our cover size lower bound proof relies on a new geometric characterization
of the space of $k$-SIIRVs that we believe is of independent interest, and may find other applications. 
Our tight lower bound on the sample complexity of learning $k$-SIIRVs relies critically on this characterization.
Our cover size lower bound is proved in Section~\ref{sec:lb}.


\subsection{Preliminaries} \label{sec:prelim}
We record a few definitions that will be used throughout this paper.

\vspace{-0.3cm}

\paragraph{Distributions and Metrics.} 
For $m \in \Z_+$, we denote $[m] \eqdef \{0,1,\dots,m\}.$
A function $\p : A \to \R$, over a finite set $A$, is called a {\em distribution}
if $\p(a) \ge 0$ for all $a \in A$, and $\sum_{a \in A} \p (a)=1.$
The function $\p$ is called a {\em pseudo-distribution} if $\sum_{a \in A} \p (a)=1.$
For a pseudo-distribution $\p$ over $[m],$ $m \in \Z_+,$ we write $\p(i)$ to denote
the value $\Pr_{X \sim \p}[X=i]$ of the probability density function (pdf) at point $i$, and $\p(\leq i)$ to denote
the value $\Pr_{X \sim \p}[X \leq i]$ of the cumulative density function (cdf) at point $i.$
For $S \subseteq [n]$, we write $\p(S)$ to denote $\sum_{i \in S}\p(i).$

The {\em total variation distance} between two (pseudo-)distributions
$\p$ and $\q$ supported on a finite set $A$ is
$\dtv\left(\p, \q \right) \eqdef \max_{S \subseteq A} \left |\p(S)-\q(S) \right|= (1/2) \cdot \| \p -\q  \|_1.$
Similarly, if $X$ and $Y$ are random variables, their total
variation distance $\dtv(X,Y)$ is defined as the total variation
distance between their distributions.
Another useful notion of distance between distributions/random variables is the {\em Kolmogorov
distance}, defined as $\dk\left(\p, \q \right) \eqdef \sup_{x \in \R} \left| \p(\leq x)-\q(\leq x) \right|.$
Note that for any pair of distributions $\p$ and $\q$
supported on a finite subset of $\R$ we have that $\dk\left(\p,\q\right) \le \dtv\left(\p,\q \right).$

\vspace{-0.3cm}

\paragraph{Distribution Learning.} Since we are interested in the computational complexity
of distribution learning, our algorithms will need to use a {\em succinct description} of their hypotheses.
A simple succinct representation of a discrete distribution is via an evaluation
oracle for the probability mass function. 
For $\eps > 0$, an \emph{$\eps$-evaluation oracle} for  a distribution $\p$ over $[m]$ is
a polynomial size circuit $C$ with  $O(\log m)$ input bits such that for each input $z$, 
the output of the circuit $C(z)$ equals the binary representation of the probability $\p'(z),$
for some pseudo-distribution $\p'$ which has $\dtv(\p', \p) \leq \eps.$
Another general way to succinctly specify a distribution is to give the code of an
efficient algorithm that takes ``pure'' randomness and transforms it into a sample from the distribution.
This is the standard notion of a sampler.
An \emph{$\eps$-sampler} for $\p$ is a circuit $C$
with $O(\log m+\log(1/\eps))$ input bits $z$ and $O(\log m)$ output bits $y$
which is such that when $z \sim U_m,$  then  $y \sim \p',$ for some distribution $\p'$ which has $\dtv(\p',\p) \leq \eps.$

We emphasize that our learning algorithms output {\em both an $\eps$-sampler
and an $\eps$-evaluation oracle} for the target distribution.

\vspace{-0.3cm}

\paragraph{Covers.}
Let ${\cal F}$ be a family of probability distributions. Given $\delta > 0$, a subset ${\cal G} \subseteq
{\cal F}$ is said to be a proper \emph{$\delta$-cover of ${\cal F}$} with respect to the metric $d(\cdot, \cdot)$
if for every distribution $\p \in {\cal F}$ there exists some $\q \in {\cal G}$ such that $d(\p, \q) \leq \delta.$
If ${\cal G}$ is not a subset of ${\cal F},$ then the cover is called non-proper.
The $\delta$-covering number for  $({\cal F}, d)$ is the minimum cardinality of a $\delta$-cover.
The $\delta$-packing number for  $({\cal F}, d)$ is the maximum number of points (distributions) in $\cal{F}$ at pairwise distance at least $\delta$
from each other.

\subsection{Our Approach and Techniques} \label{sec:techniques}

The unifying idea of this work is an analysis of the structure of the Fourier Transform (FT) of $k$-SIIRVs.
The FT is a natural tool to consider in this context. 
Recall that the FT of a sum of independent random variables is the product
of the FT's of the individual variables.
Moreover, if two random variables have similar FT's,
they also have similar distributions.
These two basic facts are the starting point of our analysis.
We now provide an overview of the ideas underlying our results, 
and give a comparison to  previous techniques.


\paragraph{Discussion \& Previous Approaches for Learning SIIRVs.}
Let $\mathcal{D}$ be a family of distributions over a domain of size $N.$
How many samples are required to learn an arbitrary $\p \in \mathcal{D}$ within
variation distance $\eps$? Without any restrictions on $\mathcal{D},$
it is a folklore fact that the sample complexity learning is $\Theta(N/\eps^2).$
The optimal learning algorithm in this case is the obvious one:
output the empirical distribution. By exploiting the structure of the family
$\mathcal{D},$ one may obtain better results.

A very natural type of structure to consider is some sort of  ``shape constraint'' on the probability density function, 
such as log-concavity or unimodality. There is a long line of work in statistics on this topic (see, e.g., the books~\cite{BBBB:72, GJ:14}),
and more recently in TCS~\cite{DDS12soda, CDSS14, CDSS14b, ADLS15}.
Alas, it turns out that $k$-SIIRVs do not satisfy any of the shape constraints considered in the literature (see~\cite{DDOST13focs} for a discussion).

A different type of structure, based on the notion of metric entropy~\cite{Yatracos85, Birg86, DL:01}, yields the following implication:
If a distribution class $\mathcal{D}$ has an $\eps/2$-cover of size $M,$ then it is learnable with $O(\log M / \eps^2)$
samples.\footnote{We remark that the running time of this method is $\Omega(M/\eps^2),$ which is not necessarily polynomial in the sample size.}
In a celebrated paper in information theory~\cite{Yang99}, Yang and Barron show that, for broad families of (continuous) distributions, 
the metric entropy {\em characterizes} the sample complexity of learning. 
For $k$-SIIRVs, however, this is not the case: Via Theorem~\ref{thm:cover-informal}, 
the metric entropy method implies a sample upper bound of $O((1/\eps^2) \cdot \log n + (k/\eps^2) \cdot \log^2(1/\eps)).$
Note that, since our cover size upper bound is tight, this sample bound is the limit of the metric entropy method for $k$-SIIRVs.
Thus, this method gives a suboptimal sample upper bound 
for our learning problem, both qualitatively (dependence on $n$), and quantitatively (dependence on $\eps$).

Previous work on learning $k$-SIIRVs~\cite{DDS12stoc, DDOST13focs} relies on a certain  ``regularity'' lemma about the structure
of these distributions: Any $k$-SIIRV is either $\eps$-close in total variation distance to being  $L = \Theta(k^9/\eps^4)$- ``sparse'',
i.e., it is supported on a set of at most $L$ consecutive integers, or $\eps$-close to being ``Gaussian like''.
In the former case, the distribution can be learned using $O(L/\eps^2)$ samples,
and in the latter case one can exploit the Gaussian structure to learn with a small number of samples as well.
Unfortunately, the sparse case is a bottleneck for this approach,
as any algorithm to learn a istribution over support $L$ requires $\Omega(L/\eps^2)$ samples.
Hence, one needs to exploit the structure of $k$-SIIRVs beyond the aforementioned.

\paragraph{Our Learning Approach.}
In this paper, we depart from the aforementioned approaches.
We identify a simple condition  -- the approximate sparsity of the Fourier transform -- as the
``right'' property that determines the sample complexity of our learning problem.
The Fourier sparsity explains why the sample complexity of learning $k$-SIIRVs is independent of $n,$
and allows us to obtain the sharp sample bound as a function of both $k$ and $\eps.$
We show that this is a more general phenomenon (see Theorem~\ref{thm:learn-generic} in Section~\ref{ssec:fourier-algo-gen}): 
any univariate distribution that has an $s$-sparse Fourier transform, 
in a certain well-defined technical sense, is learnable with $\widetilde{O}(s/\eps^2)$ samples.

Our computationally efficient learning algorithm proceeds as follows: 
It starts by drawing an initial set of samples to determine the effective support of the target distribution and its Fourier transform.
This is achieved by estimating the mean and variance of our SIIRV. 
{We remark that, for computational purposes, our algorithm uses the Discrete Fourier Transform (DFT). 
For the appropriate definition of the DFT, we show (Lemma~\ref{FourierSupportLem}) there exists an {\em explicit} set $S$ 
of cardinality $|S|  = O(k^2 \log(k/\eps))$ that contains all the ``heavy'' Fourier coefficients\footnote{We moreover show that 
there exists a set of cardinality $O(k \log(k/\eps))$ that contains all the ``heavy'' Fourier coefficients, alas this smaller set is not explicitly known a priori.}.}
Our algorithm then draws an additional set of samples to estimate the DFT of the target distribution 
at the points of the effective support $S,$ and sets the DFT to $0$ everywhere else.  
By exploiting the sparsity in the Fourier domain, we show that the inverse of the empirical DFT
achieves total variation distance $\eps/2$ after $\widetilde{O}(k/\eps^2)$ samples. 
Note that an explicit description of an accurate hypothesis for our learning problem  can have an effective support of size $\Omega(k\sqrt{n}).$ 
While we can easily obtain such a description (by explicitly computing the inverse DFT), 
this would not lead to a computationally efficient algorithm. 
We instead output a succinct description of our hypothesis (in time that is independent of $n$).
In particular, our algorithm outputs the empirical DFT at the points of its effective support.
Our learning algorithm is given in Section~\ref{ssec:learn-efficient}.

{
We emphasize that the implicit description of the hypothesis $\h$, via its DFT $\wh{\h},$ 
is sufficient to obtain both an approximate evaluation oracle and an approximate sampler for the target $k$-SIIRV $\p.$ 
Obtaining an approximate evaluation oracle is straightforward:  Since $\wh{\h}$ is supported on the set $S,$
we can compute $\h(i)$ in time $O(|S|).$ To obtain an efficient sampler, we proceed in two steps: We first 
show how to efficiently compute the CDF of $\h,$ using oracle access to the the DFT $\wh{\h}.$ To do this, we express
the value of the CDF at any point via a closed form expression involving the values of $\wh{\h}.$ 
Given oracle access to the CDF, we use a simple binary search procedure to sample from a distribution $\q$
satisfying $\dtv(\q, \h) \le \eps/2.$ Our sampler is given in Section~\ref{ssec:sampler}.
}

Finally, we note that our above-described Fourier-learning algorithm achieves a near-optimal sample complexity (up to logarithmic factors).
The basic idea to obtain the {\em optimal} sample complexity is to smoothly mollify the DFT instead of truncating it. 
This removes some artifacts caused by a sharp truncation and yields a hypothesis 
whose error from the true distribution decays rapidly as we move away from the mean.
Our sample-optimal upper bound is established in Section~\ref{sec:opt-sample}.


\paragraph{Cover Upper Bound.}

We start by commenting on previous approaches for proving cover upper bounds in this context.
The main technique for the $2$-SIIRV cover upper bound of~\cite{DaskalakisP09}
is the following lemma (that is deduced in~\cite{DaskalakisP09} using a result from~\cite{Roos:00}):
If two $2$-SIIRVs agree on their first $\Omega(\log(1/\eps))$ moments, then their total variation
distance is at most $\eps$. First, we show that this moment-matching lemma is quantitatively tight: 
we give an example of two $2$-SIIRVs over $k+1$ variables that agree on the first $k$ moments 
and have variation distance $2^{-\Omega(k)}$ (Proposition~\ref{prop:moments}). 

We emphasize however that such a moment-matching technique cannot be generalized to $k$-SIIRVs, even for $k=3.$
Intuitively, this is because knowledge about moments fails to account for potential periodic structure of the probability mass function
that comes into play for $k>2$. For example, $\Omega(n)$ moments do not suffice to distinguish between the cases that
a $3$-SIIRV of order $n$ is supported on the even versus the odd integers.
More specifically, in Proposition~\ref{prop:moments-3} (Appendix~\ref{app:moments-lb}), we give an explicit example of two
$3$-SIIRVs of order $n/2$ that agree exactly on the first $n-1$ moments and have disjoint supports.

In conclusion, moment-based approaches fail to detect periodic structure. 
On the other hand, this type of structure is easily detectable by considering the Fourier transform.
Our cover upper bound hinges on showing that the Fourier transform of a $k$-SIIRV is necessarily of low complexity, i.e.,
it can be succinctly described up to small error. In particular, since the Fourier transform is smooth,
we show (Lemma~\ref{lem:approx}), roughly, that its logarithm can be well approximated by a low degree Taylor polynomial on intervals of length $O(1/k).$
(Our actual statement is somewhat more complicated as it needs to account for roots of the Fourier transform
close to the unit circle.)
Therefore, providing approximations to the low-degree Taylor coefficients of the logarithm of the Fourier transform
provides a concise approximate description of the distribution.



\vspace{-0.2cm}

\paragraph{Cover Lower Bound \& Sample Lower Bound.}

Our lower bounds take a geometric view of the problem.
At a high-level, we consider the function that maps the set of $n (k-1)$ parameters defining a $k$-SIIRV to the corresponding
probability mass function.  We show that there exists a region of the space of distributions
where this function is locally invertible. For $k=2$, we in fact show that the distribution of any $2$-SIIRV with distinct parameters lies in the interior of this region.
This structural understanding allows us to use certain appropriately defined expectations to extract the effect of individual parameters on the distribution.
In addition, for $n = \Theta (\log(1/\eps)),$ we show
that near a particular $k$-SIIRV not only is the map from parameters to distribution locally a bijection, 
but that this map is actually surjective onto a ball of reasonable size. 
In other words, near this particular distribution, the $\Omega(k \log(1/\eps))$ parameters of the output distribution are effectively independent,
which intuitively implies the $(1/\eps)^{\Omega(k \log(1/\eps))}$ lower bound
on the cover size. 

To prove our sample lower bound, at a high-level,
we combine the aforementioned geometric understanding with Assouad's lemma~\cite{Assouad:83}. 
We note that one might naively expect that such a situation would lead to a lower bound of 
$\Omega(k\log(1/\eps)/\eps^2)$, but since the distributions under consideration have additional structure,
it turns out that the best lower bound that can be obtained is $\Omega(k\sqrt{\log(1/\eps)}/\eps^2).$

\subsection{Related  Work} \label{ssec:related}
Density estimation is a classical topic in statistics and machine learning
with a rich history and extensive literature (see e.g.,~\cite{BBBB:72, DG85, Silverman:86,Scott:92,DL:01}).
The reader is referred to~\cite{Izen91} for a survey of statistical techniques in this context.
In recent years, a large body of work in TCS  has been studying these questions from
a computational perspective; see e.g.,
\cite{KMR+:94short,FreundMansour:99short,AroraKannan:01, CGG:02, VempalaWang:02,FOS:05focsshort, BelkinSinha:10, KMV:10,
MoitraValiant:10,DDS12soda,DDS12stoc, DDOST13focs, CDSS14, CDSS14b, ADLS15}.

Covering numbers (and their logarithms, known as {\em metric entropy} numbers) were first defined by A. N. Kolmogorov in the 1950's
and have since played a central role in a number of areas, including
approximation theory, geometric functional analysis (see, e.g.,~\cite{Dudley:74, Makovoz86,  BleiGaoLi07}
and the books~\cite{Kolm59, Lor66, CarlStephani90, ET96}), geometric approximation algorithms~\cite{HP11},
information theory, statistics, and machine learning  (see, e.g.,~\cite{Yatracos85, Birg86, Hasm90, Haussler97, Yang99, GuntSen13}
and the books~\cite{VWellner96, DL:01, Tsybakov08}).

\vspace{-0.3cm}

\paragraph{Concurrent Work.}
Concurrent work by Daskalakis {\em et al.}~\cite{DKT15}, using different techniques,
gives upper bounds on the learning sample complexity of Poisson Multinomial Distributions (PMDs).
While upper bounds on the sample complexity of PMDs yield similar upper bounds for $k$-SIIRVs,
the implied upper bounds for $k$-SIIRVs are quantitatively significantly weaker than ours.
Moreover, the~\cite{DKT15} learning algorithm has 
running time exponential in $k$  and super-polynomial in $1/\eps.$

\vspace{-0.3cm}

\paragraph{Subsequent Work.} In a followup work~\cite{DKS15c}, the authors have generalized the techniques
of this paper to the multidimensional case, namely to the family of Poisson Multinomial Distributions (PMDs), i.e., sums
of independent random vectors supported over the standard basis in $\R^k.$
We note that the results of the current paper are not subsumed by the results of~\cite{DKS15c}.
In particular, ~\cite{DKS15c} gives an efficient learning algorithm for PMDs that uses $\log^{O(k)}(1/\eps)/\eps^2$
samples, and proves that the optimal cover size for PMDs depends doubly exponentially on $k.$

\subsection{Organization} \label{ssec:org}
In Section~\ref{sec:learn} we describe and analyze our learning algorithms for $k$-SIIRVs.
Section~\ref{sec:cover-ub} contains our cover upper bound construction. 
Our cover lower bound is given in Section~\ref{sec:lb}, and our sample lower bound in Section~\ref{sec:sample-lb}.

\section{Learning SIIRVs}  \label{sec:learn}
In this section, we describe our algorithms for learning $k$-SIIRVs.
The structure of this section is as follows: In Section~\ref{ssec:learn-efficient}, 
we give our sample near-optimal and computationally efficient learning algorithm.
As mentioned in the introduction, our algorithm outputs a succinct description of its hypothesis
$\h$, via its DFT.
In Section~\ref{ssec:fourier-algo-gen}, we provide a simple general algorithm
that learns any one-dimensional discrete distribution with a sparse Fourier support.
In Section~\ref{ssec:sampler}, we show how to efficiently obtain an $\eps$-sampler for our unknown $k$-SIIRV, 
using the DFT representation of $\h$ as a black-box. Finally, in Section~\ref{sec:opt-sample}
we present our more sophisticated Fourier-based learning algorithm with optimal sample complexity.

\subsection{A Computationally Efficient Sample Near-Optimal Algorithm} \label{ssec:learn-efficient}

\ignore{
Some notes on the upper bound:
\begin{itemize}
\item If we use the cover theorem as a black-box, we get a $\log n$ factor in the sample complexity.
\item If we additionally exploit the {\em structure} of the cover, we can get away without a dependence on $n$.
The idea for this is as follows: Given samples from the distribution, we learn it as if it was in Case (2) of Theorem 1.2
of my FOCS'13 paper, i.e., if it was a sum of a discrete Gaussian and a $k$-sparse random variable.
Then, we use the cover theorem for the $\poly(k/\eps)$-sparse case, and we do a hypothesis testing
between the two.
\end{itemize}
}

The main result of this subsection is Theorem~\ref{thm:learning-informal}, which we state
below in more detail for the sake of completeness.

\begin{theorem}\label{thm:FT-alg}
There is an algorithm {\tt Learn-SIIRV} that for any $\p\in \Scal_{n,k}$ and $\eps>0$,
takes $O(k\log^{2}(k/\eps)/\eps^2)$ samples from $\p$, runs in
time $\widetilde{O}(k^3/\eps^2)$ and returns a (succinct description of a) hypothesis $\h$
so that with probability at least $2/3$ we have that $\dtv(\p, \h) < \eps.$
\end{theorem}

For computational purposes, our learning algorithm in this section 
uses the Discrete Fourier Transform, which we now define.

\begin{definition}
For $x \in \R$ we will denote $e(x) \eqdef  \exp(-2 \pi i x)$.
The {\em Discrete Fourier Transform (DFT) modulo $M$} of a function
$F:[n] \rightarrow \C$ is  the function $\widehat{F}:[M-1]\rightarrow \C$ defined as
$\widehat{F}(\xi)=\sum_{j=0}^n e(\xi j/M) F(j) \;,$
for integers $\xi \in [M-1]$. The DFT modulo $M$ of a distribution $\p$, $\widehat{\p}$ is the DFT modulo $M$ of its probability mass function.
The {\em inverse DFT modulo $M$} onto the range $[m,m+M-1]$ of
$\widehat{F}: [M-1] \rightarrow \C$,
is the function $F: [m, m+M-1] \cap \Z \rightarrow \C$ defined by
$F(j)= \frac{1}{M} \sum_{\xi=0}^{M-1} e(-\xi j/M) \widehat{F}(\xi) \;,$
for $j \in [m,m+M-1] \cap \mathbb{Z}$.
The $L_2$ norm of the DFT is defined as $\|\widehat{F}\|_2 = \sqrt{\frac{1}{M} \sum_{\xi=0}^{M-1} |\widehat{F}(\xi)|^2} \;.$
\end{definition}

We start by giving an intuitive explanation of our approach.
The Fourier transform $\widehat{\q}$ of the empirical distribution $\q$ provides an approximation to the Fourier transform $\widehat{\p}$ of $\p$.
In particular, if we take $N$ samples from $\p$, we expect that the empirical Fourier transform $\widehat{\q}$ has error $O(N^{-1/2})$ at each point.
This implies that the expected $L_2$ error $\|\widehat{\q} -\widehat{\p} \|_2$ is $O(N^{-1/2})$, and thus by applying the inverse Fourier transform,
would yield a distribution with $L_2$ error of $O(N^{-1/2})$ from $\p$. This guarantee may sound good, but unfortunately, the distribution $\p$
has effective support of size approximately $s \sqrt{\log(1/\eps)},$ where $s=\sqrt{\var_{X \sim \p}[X]}$, and thus the resulting distribution
will likely have $L_1$ error of $O(N^{-1/2} s^{1/2} \log^{1/4}(1/\eps))$ from $\p$. This bound is prohibitively large, especially when the standard deviation of $\p$ is large.

This obstacle can be circumvented by relying on a new structural result that we believe may be of independent interest.
{\em We show that for any $k$-SIIRV with large variance, its Fourier Transform will have small effective support.}
In particular, for any $k$-SIIRV with standard deviation $s$ and $\eps > 0$ we consider its Discrete Fourier transform modulo $M$,
and show the set of points in $[M-1]$ whose Fourier transform is bigger than $\eps$
in magnitude has size at most $O(M ks^{-1}\sqrt{\log(1/\eps)})$. By choosing $M$ to be approximately $s \sqrt{\log(1/\eps)}$, i.e.,
of the same order as the effective support of $\p$, we conclude that the effective support of $\widehat{\p}$ (modulo $M$) is $O(k \log(1/\eps))$.

If the effective support for  $\widehat{\p}$  was explicitly known, we could truncate our empirical Discrete Fourier transform $\widehat{\q}$
(modulo $M$) outside this set and reduce the $L_2$ error $\|\widehat{\q} -\widehat{\p} \|_2$ to $N^{-1/2} k^{1/2}s^{-1/2}\log^{1/4}(1/\eps)$.
This in turn would correspond to an $L_1$ error of $O(N^{-1/2}k^{1/2}\sqrt{\log(1/\eps)})$.
Unfortunately, we do not know exactly where the support of the Fourier transform is,
so we will need to approximate it by calculating the empirical DFT where the support
might be, and then simply truncating this empirical DFT whenever it is sufficiently small.
Fortunately, we do have some idea of where the support is and it is not hard to show that we can truncate at all of the appropriate points with high probability.

\bigskip

\fbox{\parbox{6.4in}{
{\bf Algorithm} {\tt Learn-SIIRV}\\
Input: sample access to a $k$-SIIRV $\p$ and $\eps>0$.\\

\vspace{-0.2cm}

 Let $C$ be a sufficiently large universal constant.

\vspace{-0.2cm}

\begin{enumerate}

\item
Draw $O(1)$ samples from $\p$
and with confidence probability $19/20$ compute: (a) $\widetilde{\sigma}^2$, a factor $2$ approximation to $\var_{X \sim \p}[X]+1$,
and (b) $\widetilde{\mu}$, an approximation to $\E_{X \sim \p}[X]$ to within one standard deviation.

\vspace{-0.2cm}

\item Take $N=C^3 k/\eps^2 \ln^{2}(k/\eps)$ samples from $\p$ to get an empirical distribution $\q$.

\vspace{-0.2cm}

\item If $\widetilde{\sigma}  \le 4 k \ln(4/\eps)$, then output $\q$. Otherwise, proceed to next step.

\vspace{-0.2cm}

\item Set $M \eqdef 1+2\lceil 6\widetilde{\sigma}\sqrt{\ln(4/\eps)}) \rceil$. Let
\vspace{-0.2cm}
$$S \eqdef \left\{ \xi \in [M-1] \mid \exists a, b \in \Z, 0 \leq a \le b < k \textrm{ such that }
|\xi/M - a/b| \le  O(\log(k/\eps)/M)  \right\} \;.$$
For each $\xi \in S$, compute the DFT modulo $M$ of $\q$ at $\xi$, $\widehat{\q}(\xi)$.


\item Compute $\widehat{\h}$ which is defined  as $\widehat{\h}(\xi) = \widehat{\q}(\xi)$ if $\xi \in S$ and $|\widehat{\q}(\xi)| \geq R:=2 C^{-1} \eps/ \sqrt{k\ln(k/\epsilon)}$, and $\widehat{\h}(\xi) = 0$ otherwise.

\item Output $\widehat{\h}$ which is a succinct representation of $\h$,
the inverse DFT of $\widehat{\h}$ modulo $M$ onto the range  $[\lfloor \widetilde{\mu} \rfloor-(M-1)/2,\lfloor \widetilde{\mu} \rfloor+(M-1)/2]$.

\end{enumerate}
}}

\bigskip

The bulk of our analysis will depend on showing that the Fourier transform of $\p$ has appropriately small effective support.
To do this we need the following lemma:


\begin{lemma}\label{FourierSupportLem}
Let $\p \in \mathcal{S}_{n, k}$ with $\sqrt{\Var_{X \sim \p}[X]} = s$,
$1/2>\delta>0$, and $M \in \Z_+$ with $M>s$.
Let $\widehat{\p}$ be the discrete Fourier transform of $\p$ modulo $M$. Then, we have
\begin{itemize}
\item[(i)] Let $\mathcal{L} = \mathcal{L}(\delta, M,s) \eqdef \left\{ \xi \in [M-1] \mid \exists a, b \in \Z, 0 \leq a \le b < k \textrm{ such that }
|\xi/M - a/b| <  \frac{\sqrt{\ln (1/\delta)}}{2s}  \right\} \;.$ Then, $|\widehat{\p}(\xi)| \leq \delta$ for all $\xi \in [M-1] \setminus \mathcal{L}.$
That is, $|\widehat{\p}(\xi)| > \delta$ for
at most $|\mathcal{L}| \leq M k^2 s^{-1} \sqrt{\log(1/\delta)}$ values of $\xi$ .
\item[(ii)] At most $4Mks^{-1}\sqrt{\log(1/\delta)}$ many integers $0 \leq \xi \leq M-1$ have  $|\widehat{\p}(\xi)| > \delta \;.$
\end{itemize}
\end{lemma}

Before we proceed with the proof of the lemma some comments are in order.
Statement (i) of the lemma exhibits an explicit set $\mathcal{L}$ of cardinality $O(M k^2 s^{-1} \sqrt{\log(1/\delta)})$ that contains all the points $\xi \in [M-1]$
such that $|\widehat{\p}(\xi)| > \delta.$ Note that the set $\mathcal{L}$ can be efficiently computed  from $M$, $\delta$, $s$,
and does not otherwise depend on the particular $k$-SIIRV $\p$.
Statement (ii) of the lemma shows that the effective support $\mathcal{L}' = \mathcal{L}' (\delta) = \{\xi \in [M-1] \mid  |\widehat{\p}(\xi)| > \delta \}$
is in fact significantly smaller than $\mathcal{L}$, namely  $|\mathcal{L}'| = O(Mks^{-1}\sqrt{\log(1/\delta)})$.
This part of the lemma is non-constructive in the sense that it does not provide an explicit description for $\mathcal{L}'$
(beyond the fact that $\mathcal{L}' \subseteq \mathcal{L}$). The upper bound on the size of the effective support is the basis for the analysis of our algorithm.

\begin{proof}[Proof of Lemma~\ref{FourierSupportLem}]
Since  $\p \in \mathcal{S}_{n, k}$, for $X \sim \p$, we have $X=\sum_{i=1}^n X_i$
where each $X_i \sim \p_i$ for a $k$-IRV $\p_i$.
Let $Y_i=X_i-X'_i$ be the difference of two independent copies of $X_i$.
Let $p_{ij} = \Pr \left[|Y_i|=j \right] .$
Note that $Y_i$ is a symmetric random variable. Consider its DFT modulo $M$ which we will write as $\widehat{Y_i}$. We have the following sequence of (in)equalities:
\begin{eqnarray*}
|\widehat{\p_i}(\xi)|^2 = \widehat{\p_i}(\xi) \widehat{\p_i}(-\xi)  &=& \widehat{Y_i}(\xi)\\
					& =&  \sum_{j=0}^{k-1} p_{ij}\cos\left(\frac{2\pi\xi j}{M} \right)
					 = 1- \sum_{j=1}^{k-1} p_{ij}\left(1-  \cos\left(\frac{2\pi\xi j}{M} \right)  \right) \\
					& \leq&  1 - 8 \sum_{j=1}^{k-1} p_{ij}[\xi j/M]^2 \leq  \exp \left(- 8 \sum_{j=1}^{k-1} p_{ij}[\xi j/M]^2 \right) \;,
\end{eqnarray*}
where $[x]$, $x\in \R$, denotes the distance between $x$ and its nearest integer.
For the last two inequalities, we used that $\cos 2 \pi x \leq 1 - 8x^2$ when $|x| \leq 1/2$, and $e^{-x} \geq 1-x$ when $x \geq 0$.

Therefore, we have that
$|\widehat{\p}(\xi)|^2  = \prod_{i=1}^n |\widehat{\p_i}(\xi)|^2 \leq \exp(- 8 \sum_{i=1}^n\sum_{j=1}^{k-1} p_{ij}[\xi j/M]^2)$.
Taking square roots, we obtain
\begin{equation} \label{ineq:this}
|\widehat{\p}(\xi)| \leq \exp\Big(- 4 \sum_{i=1}^n\sum_{j=1}^{k-1} p_{ij}[\xi j/M]^2 \Big).
\end{equation}
Note that we can relate the variance of $\p$ to the $p_{ij}$'s as follows:
\begin{equation} \label{eq:varj}
s^2 = \var[X]  = \sum_{i=1}^n \var[X_i] = \frac{1}{2}\sum_{i=1}^n\E[Y_i^2]
= \frac{1}{2}\sum_{i=1}^n\sum_{j=1}^{k-1} p_{ij}j^2 \; .
\end{equation}
Using (\ref{ineq:this}), we get
$$ |\widehat{\p}(\xi)| \leq \exp\left(-8 s^2 \left( \min_j \big( \frac{[\xi j/M]}{j} \big) ^2 \right) \right) .$$
To complete the proof of (i), we will need a simple counting argument given in the following claim:
\begin{claim}\label{clm:lets-count}
For $a \in \R_+$ $j \in \Z_+$, there are at most $2Maj+j$ integers $0 \leq \xi \leq M-1$
with the following property:
there exists $c \in \Z$ with $0 \leq c \leq j$ such that $|\xi/M - c/j| < a$.
Therefore, there are at most $2Ma+j$ integers $0 \leq \xi \leq M-1$ with $[\xi j/M] < a$.
\end{claim}
\begin{proof}
For each $c$ satisfying $1 \leq c \leq j-1$
there are either $\lfloor 2Ma \rfloor$ or $\lfloor 2Ma \rfloor + 1$  integers $0 \leq \xi \leq M-1$ with $|\frac{\xi}{M} -\frac{c}{j}| < a$.
For $c=0$ and $c=j$ there are either $\lfloor Ma \rfloor$ or $\lfloor Ma \rfloor + 1$ integers with $|\frac{\xi}{M} -\frac{c}{j}| < a$.
Finally, note that  $|\frac{\xi}{M} -\frac{c}{j}| < a$ for some $1 \leq c \leq j-1$ if and only if $[j\xi/M] < aj$.
\end{proof}

An application of the above claim for $a = (1/2s)\sqrt{\ln(1/\delta)}$ implies that there are at most
$$\sum_{j=1}^{k-1}  2 M  j s^{-1} \sqrt{\ln(1/\delta)}/2 +j \leq M  k^2 s^{-1} \sqrt{\ln(1/\delta)}+k^2 \leq 2 M k^{2} s^{-1} \sqrt{\ln(1/\delta)}$$
integers $0 \leq \xi \leq M-1$  with $\min_j \big( \frac{[\xi j/M]}{j} \big) ^2 <  \ln (1/\delta)/(4s^2)$.
For all other integers we have $|\widehat{\p}(\xi)| \leq \delta \;,$
which completes the proof of (i).

To prove (ii) we proceed by the probabilistic method as follows:
Consider evaluating the RHS of (\ref{ineq:this}) with $\xi$ being an integer random variable uniformly distributed in $[M-1]$.
For $1 \leq j \leq k-1$, let $N_j$ be the indicator random variable for the event that $[\xi j/M] < k s^{-1} \sqrt{\ln(1/\delta)}/2$.
Observe that by Claim~\ref{clm:lets-count} it follows that $\E[N_j] \leq 2 k s^{-1} \sqrt{\ln(1/\delta)}$.

Note that $[\xi j/M] \geq \sqrt{1-N_j} \cdot k s^{-1} \sqrt{\ln(1/\delta)}/2$. Plugging this into (\ref{ineq:this}) gives
$$
|\widehat{\p}(\xi)|  \leq \exp\left( -\frac{k^2}{s^2} \ln(1/\delta) \sum_{i=1}^n\sum_{j=1}^{k-1} p_{ij}(1-N_j)\right).
$$
Since  $s^2 = \frac{1}{2} \sum_{i=1}^n\sum_{j=1}^{k-1} p_{ij}j^2 \leq \frac{k^2}{2}\sum_{i=1}^n\sum_{j=1}^{k-1} p_{ij}$,
it follows that
$\theta:= \sum_{i=1}^n\sum_{j=1}^{k-1} p_{ij} \geq 2s^2/k^2.$
Therefore,
$$
\E\left[\sum_{i=1}^n\sum_{j=1}^{k-1} p_{ij}N_j \right] \leq \theta \cdot 2 k s^{-1} \sqrt{\ln(1/\delta)}.
$$
By Markov's inequality, except with probability $4k s^{-1} \sqrt{\ln(1/\delta)}$, we have that
$\sum_{i=1}^n\sum_{j=1}^{k-1} p_{ij}N_j \leq \frac{\theta}{2}.$
In this event, we have $\sum_{i=1}^n\sum_{j=1}^{k-1} p_{ij}(1-N_j) \geq \frac{\theta}{2}$ and hence
\begin{eqnarray*}
|\widehat{\p}(\xi)| \leq \exp\left( -\frac{k^2}{s^2} \ln(1/\delta) \sum_{i=1}^n\sum_{j=1}^{k-1} p_{ij}(1-N_j)\right)
 \leq \exp\left( -\frac{k^2}{s^2} \ln(1/\delta) \frac{\theta}{2}\right) \leq \delta.
\end{eqnarray*}
Since $\xi$ is uniformly distributed on $[M-1]$, it follows that $|\widehat{\p}(\xi)| > \delta$
for at most $4Mks^{-1} \sqrt{\ln(1/\delta)}$ integers $\xi$ in $[M-1]$.
This completes the proof of (ii).
\end{proof}

\noindent We are now ready to prove Theorem~\ref{thm:FT-alg}.

\begin{proof}[Proof of Theorem~\ref{thm:FT-alg}]
Note that it is straightforward to verify the sample complexity bound.
The running time of the algorithm is dominated by computing the DFT $\widehat{\q}$.
Since the support of $\q$ is at  most $N$,
for each $\xi \in S$, we sum at most $N$ terms to calculate $\widehat{\q}(\xi)$. Therefore,
the overall running time is $O(N \cdot |S|) = O(k\log^2(k/\eps)/\eps^2 \cdot k^2 \log(k/\eps))=O(k^3 \log^3(k/\eps) / \eps^2)$ as claimed.

To show correctness, we will prove that the expected squared $L_2$ norm between $\widehat{\h}$ and $\widehat{\p}$ is small, i.e., that
$\|\widehat{\h}-\widehat{\p}\|_2^2 = (1/M) \cdot \sum_{\xi=0}^{M-1} |\widehat{\h}(\xi)-\widehat{\p}(\xi)|^2$
has small expected value.

It is easy to see that, after drawing a constant number of samples,
the quantities  $\widetilde{\mu}$ and $\widetilde{\sigma}$ can be estimated
to satisfy the required conditions with probability at least $19/20$.
(This follows for example by Lemma 6 of \cite{DDS12stoc} with $\eps=1/2.$)
We will henceforth condition on this event.

If $\widetilde{\sigma} \leq 4 k \ln(4/\eps)$, then $s \leq 2 k \ln(4/\eps)+1$,
and Bernstein's inequality implies that $X \sim \p$ is within $O(k \log(1/\eps))$ of the mean
with probability $1-\eps/2$. In this case, $O(k \log (1/\eps)/\eps^2) \leq N$ samples
are sufficient to give that $\dtv(\p,\q) \leq \eps$ with probability $2/3$.
(This follows from the fact that any distribution over support of size $L$
can be learned with $O(L/\eps^2)$ samples to total variation distance $\eps$.)
We henceforth assume that we have $|\mu-\widetilde{\mu}| \leq s$,
$s \geq \widetilde{\sigma}/2 \geq 2 k \ln(4/\eps)$ and $\widetilde{\sigma} \leq 2s$.

Since $M = 1+2\lceil 6\widetilde{\sigma}\sqrt{\ln(4/\eps)}) \rceil$, a random variable $X \sim \p$ lies in
$[\lfloor \widetilde{\mu} \rfloor-(M-1)/2, \lfloor \widetilde{\mu} \rfloor+(M-1)/2]$ with probability at least $1- \frac{\eps}{2}$.
Indeed, an application of
Bernstein's inequality for $X$ yields that
$$\Pr(X > \mu+t) \leq \exp\left(-\frac{t^2}{2s^2 + \frac{2}{3}kt}\right) \;,$$
where $\mu$ is the mean of $\p$, for any $t>0$.
For $t=2s \sqrt{ \ln(4/\eps)}$, we have $t^2=(\ln(4/\eps))4s^2$ and
$2s^2 + \frac{2}{3}kt = 2s^2+\frac{4}{3}ks \sqrt{\ln(4/\eps)} \leq \frac{8}{3}s^2 \leq 4s^2$.
Thus, $\Pr(X > \mu+t) \leq \eps/4$. Similarly, it holds $\Pr(X < \mu-t) \leq \eps/4.$
Now note that $ \lfloor\widetilde{\mu} \rfloor+(M-1)/2 \geq (\mu-s)+\lceil 3s\sqrt{\ln(4/\eps)}) \rceil \geq \mu + t$
and $\lfloor \widetilde{\mu} \rfloor - (M-1)/2 \leq \mu-t$.
Hence, $X$ is in $[\lfloor \widetilde{\mu} \rfloor-(M-1)/2, \lfloor \widetilde{\mu} \rfloor+(M-1)/2]$ with probability at least $1- \eps/2$ as desired.

Fix $T=R/2=C^{-1} \eps/(\sqrt{k \ln(k/\eps))}$. We analyze separately the contribution to the squared $L_2$ norm coming
from $\xi$ with $|\widehat{\p}(\xi)|>T$ and with $|\widehat{\p}(\xi)|\leq T$.
Let us denote $\mathcal{L'}(T) = \{ \xi \in [M-1] \mid |\widehat{\p}(\xi)|>T  \}$.
First consider
$$
(1/M) \cdot \sum_{\xi \in \overline{\mathcal{L'}}(T)} |\widehat{\h}(\xi)-\widehat{\p}(\xi)|^2.
$$
We first claim that with high probability $\widehat{\h}(\xi)=0$ for all $\xi \in \overline{\mathcal{L'}}(T)$.
This happens automatically when $\xi \not \in S$, where the $S$ is defined in the algorithm description.
Note that $|S| = O(k^2 \log(k/\eps))$.
For $\xi \in S \setminus \mathcal{L'}(T)$, we note that $\widehat{\q}(\xi)$ is an average of $N$
i.i.d. numbers each of absolute value $1$ and mean $\widehat{\p}(\xi)$ (which has absolute value less than $1$).
Note that if $|\widehat{\q}(\xi) -\widehat{\p}(\xi)| \geq R-T$,
then either the real or the imaginary part is at least $(R-T)/\sqrt{2}$.
By a Chernoff bound, the probability that for a given
$\xi \in S \setminus \mathcal{L'}(T)$,  $\Re (\widehat{\q}(\xi) -\widehat{\p}(\xi)) \geq (R-T)/\sqrt{2}$
is at most $2\exp(-N(R-T)^2/4)$.
The same is true of the imaginary part so by a union bound the probability
that $|\widehat{\q}(\xi) -\widehat{\p}(\xi)| \geq R-T$ is at most $4\exp(-N(R-T)^2/4)$.
Again by a union bound we get that the probability that any
$\xi \in S \setminus \mathcal{L'}(T)$ has $|\widehat{\q}(\xi) -\widehat{\p}(\xi)| \geq R-T$
is at most $O(k^2 \log(k/\eps) \exp(-N(R-T)^2/4))=O(k^2 \log(k/\eps) \exp(-C \ln(k/\eps))) = O(\eps^{C-1})$.
Hence, except with probability $O(\eps^{C-1})$, for all $\xi$ in $S\setminus \mathcal{L'}(T)$
we have $|\widehat{\q}(\xi) -\widehat{\p}(\xi)| < R-T$ and so $|\widehat{\q}(\xi)| \leq R$.
In fact, the total expected contribution to the squared $L_2$ norm coming from cases
when $\widehat{\h}(\xi)$ is not identically $0$ on all such $\xi$ is also $O(\eps^{C-1})$.
Therefore, up to negligible error, the squared $L_2$ error coming from this range is at most
$$
\sum_{r\geq 0} (T 2^{-r})^2 \left(\frac{\#\{\xi:|\widehat{\p}(\xi)|>T 2^{-r-1} \}}{M}\right).
$$
Applying Lemma \ref{FourierSupportLem} (ii) with $\delta := T 2^{-r-1}$ for each $r \geq 0$, this is at most
\begin{align*}
\sum_{r\geq 0} (T 2^{-r})^2 \left(\frac{\#\{\xi:|\widehat{\p}(\xi)|>T 2^{-r-1} \}}{M}\right) & \leq \sum_{r\geq 0} T^2 4^{-r} 4k s^{-1}\sqrt{\log(2^r/T)}\\
& \leq 8 T^2 k s^{-1} \sqrt{\log(1/T)}.
\end{align*}
We now consider the remaining contribution
$$
(1/M) \cdot \sum_{\xi \in \mathcal{L'}(T)} |\widehat{\h}(\xi)-\widehat{\p}(\xi)|^2.
$$
By Lemma \ref{FourierSupportLem} (i) applied with $\delta :=T,$
it follows that $\mathcal{L'}(T) \subseteq \mathcal{L}(T,M,s).$
{Since $\sqrt{\ln(1/T)}/2s = O(\log(k/\eps)/M)$, we can choose the constant in the definition of $S$ so that $ \mathcal{L}(T,M,s) \subseteq S$}.
So, for $\xi \in \mathcal{L'}(T)$, we do compute $\widehat{\q}(\xi)$ and then either $\widehat{\h}(\xi)=\widehat{\q}(\xi)$ or $|\widehat{\q}(\xi)|<R$ and $\widehat{\h}(\xi)=0$. In either case, we have that $|\widehat{\h}(\xi)-\widehat{\q}(\xi)|<R$.   Recall that the expected size of $|\widehat{\q}(\xi)-\widehat{\p}(\xi)|^2$ is $1/N$ for any $\xi \in [M-1]$. So, for $\xi \in \mathcal{L'}(T)$, the expected squared error at $\xi$ satisfies $\E[|\widehat{\h}(\xi) - \widehat{\p}(\xi)|^2] \leq 2(R^2+N^{-1})$.

By Lemma \ref{FourierSupportLem} (ii) applied with $\delta :=T,$
we have $|\mathcal{L'}(T)| \leq 4 k s^{-1} \sqrt{\ln(1/T)}$.
So, the expected size of the $L_2^2$ error on $\mathcal{L'}(T)$ has
$$
\E[(1/M) \cdot \sum_{\xi \in \mathcal{L'}(T)} |\widehat{\h}(\xi)-\widehat{\p}(\xi)|^2] \leq 4(R^2 + N^{-1})(2 k s^{-1} \sqrt{\ln(1/T)}) \;.
$$
Combining the above results, we find that the expected $L_2^2$  error between $\widehat{\h}$ and $\widehat{\p}$ is at most
$$
4(R^2 + N^{-1}+T^2)(2k s^{-1} \sqrt{\log(1/T)}) = O(C^{-1}s^{-1}\epsilon^2/\sqrt{\log(k/\epsilon)}).
$$
Therefore, if $C$ is sufficiently large, Markov's inequality yields that,
with probability $\frac{2}{3}$, we have $\|\widehat{\h}-\widehat{\p}\|_2^2 < \epsilon^2/M.$

At this point, we would like to use Plancherel's theorem followed by Cauchy-Schwartz to complete the proof.
Formally, since $\p$ may be supported outside $[\lfloor \widetilde{\mu} \rfloor -(M-1)/2, \lfloor \widetilde{\mu} \rfloor +(M-1)/2]$,
we cannot use Plancherel's theorem directly to show that  $\|\h-\p\|_2 =\|\widehat{\h}-\widehat{\p}\|_2$.
Instead, consider the function $\p' : [\lfloor\widetilde{\mu}\rfloor-(M-1)/2, \lfloor \widetilde{\mu} \rfloor+(M-1)/2] \cap \Z \rightarrow [0,1]$ defined as
$\p'(i) = \sum_{j \equiv i \pmod M} \p(j)$ for $\lfloor \widetilde{\mu} \rfloor-(M-1)/2 \leq i \leq \lfloor \widetilde{\mu} \rfloor+(M-1)/2$.
 Note that $\widehat{\p'}=\widehat{\p}$ by the definition of the DFT modulo $M$,
 since $e(\xi j/M) =e(\xi i/M)$ when $j \equiv i \pmod  M$ for all $\xi \in [M-1]$ and $i,j \in [n]$.
Thus, $\|\widehat{\h}-\widehat{\p'}\|_2^2 < \eps^2/M$ and Plancherel's theorem gives $\|\h-\p'\|_2 = \|\widehat{\h}-\widehat{\p'}\|_2 < \eps/\sqrt{M}$.
Since $\p'$ has support at most $M$,  an application of Cauchy-Schwartz gives $\|\h-\p'\|_1 \leq \|\h-\p'\|_2\sqrt{M} < \eps.$

Since $X \sim \p$ is in $[\lfloor \widetilde{\mu} \rfloor -(M-1)/2, \lfloor \widetilde{\mu} \rfloor +(M-1)/2]$ with probability at least $1-\eps/2$, we have $\|\p-\p'\|_1 \leq \eps$ and so $\|\p-\h\|_1 \leq \|\p-\p'\|_1 + \|\h-\p'\|_1 \leq 2\eps$.
Since $\widehat{\h}(0)=\widehat{\q}(0)=1$, it follows that $\sum_{i=0}^n \h(i) = 1$. Also, by symmetry, all the $\h(i)$'s are real.
This completes the proof of Theorem~\ref{thm:FT-alg}.
\end{proof}

\subsection{A General Fourier Learning Algorithm} \label{ssec:fourier-algo-gen}

The algorithmic approach of the previous subsection is not specialized to $k$-SIIRVs, but is applicable more generally.
In essence, the approach really only depended upon two facts:
\begin{itemize}
\item $\p$ is effectively supported on a small set $T.$
\item $\wh{\p}$ is effectively supported on a small set $S.$
\end{itemize}
It turns out that by using similar ideas, we can learn {\em any}
probability distribution with these properties. The following simple theorem
provides a generalization for integer-valued random variables. However, the
approach can also be generalized to higher dimensions and to continuous distributions.

\begin{theorem} \label{thm:learn-generic}
Let $\p$ be an integer-valued random variable and $\eps>0$.
Let $T\subset \Z$ and $S\subset \R/\Z$ be known subsets so that:
$$
\sum_{n\in\Z \setminus T} \p(n) \leq \eps/3, \textrm{ and } \int_{\xi\in(\R/\Z) \setminus S}|\wh{\p}(\xi)|^2d\xi < \epsilon^2/(9{|T|}).
$$
Then, there exists an algorithm which learns $\p$ to total variational distance $\eps$
using $N=O(|T|\mu(S)/\eps^2)$ samples, {where $\mu(S)$ is the Lebesgue measure of $S.$}
\end{theorem}

The generic algorithm is as follows:

\bigskip

\fbox{\parbox{6.4in}{
{\bf Algorithm} {\tt Learn-Sparse-FT}\\
Input: sample access to a distribution $\p$ over $[n]$ and $\eps>0.$\\

\vspace{-0.2cm}

 Let $C$ be a sufficiently large universal constant.

\vspace{-0.2cm}

\begin{enumerate}

\item Take $N=C |T|\mu(S)/\eps^2$ samples from $\p$ to get an empirical distribution $\q$.

\item Compute $\widehat{\h}$ which is defined  as $\widehat{\h}(\xi) = \widehat{\q}(\xi),$ if $\xi \in S,$ and $\widehat{\h}(\xi) = 0$ otherwise.

\item Output $\h,$ where $\h$ is the inverse Fourier transform on $\wh{\h}$ restricted to $T.$
In particular $\h(i)=\int_{\xi\in S} e(-n\xi)\wh{\h}(\xi)d\xi$ for $i\in T$ and $0$ for $i\not\in T.$

\end{enumerate}
}}

\bigskip

Note that this is exactly the form of the algorithm for learning $k$-SIIRVs,
except that the latter algorithm must also learn $T$ (which is done by computing an approximate mean and variance)
and $S$ (which is obtained through a thresholding procedure). Also, note that we use the continuous Fourier transform
here rather than a discrete Fourier transform. This is mostly for conceptual convenience.
In practice, the continuous Fourier transform can be replaced by a sufficiently fine discrete Fourier transform,
yielding an algorithm in which the integrals can be replaced by finite sums.

The analysis of the algorithm is not difficult.
We begin by bounding that the expected $L_2$ difference between $\wh{\p}$ and $\wh{\h}.$
In particular, we note that
\begin{align*}
\int_{\xi \in \R/\Z} |\wh{\p}(\xi)-\wh{\h}(\xi)|^2 & =  \int_{\xi \in S} |\wh{\p}(\xi)-\wh{\h}(\xi)|^2 + \int_{\xi \in (\R/\Z) \setminus S} |\wh{\p}(\xi)-\wh{\h}(\xi)|^2 \\ 
& \leq \int_{\xi \in S} |\wh{\p}(\xi)-\wh{\h}(\xi)|^2 + \epsilon^2/(9|T|).
\end{align*}
Now, for any given value of $\xi$, we note that $\wh{\p}(\xi)-\wh{\h}(\xi)$ has mean $0$ and variance at most $1/N$. Therefore, we have that
$$
\E\left[\int_{\xi \in S} |\wh{\p}(\xi)-\wh{\h}(\xi)|^2 \right] \leq \mu(S)/N = \eps^2/(C|T|).
$$
For $C$ large enough, by the Markov inequality, this is at most $\eps^2/(9|T|)$ with probability at least $2/3.$ If this holds, then
$$
\int_{\xi \in \R/\Z} |\wh{\p}(\xi)-\wh{\h}(\xi)|^2 \leq \eps^2/(4|T|).
$$
By Plancherel's Theorem, this would imply that the squared $L^2$ distance between $\p$
and the inverse Fourier transform of $\h$ is at most $\epsilon^2/(4|T|).$ Along with Cauchy-Schwartz, this implies that
$$
\sum_{n\in T} |\p(n)-\h(n)| \leq \sqrt{(\eps^2/(4|T|))|T|} = \eps/2.
$$
On the other hand,
$$
\sum_{n\in\Z\backslash T} |\p(n)-\h(n)| = \sum_{n\in\Z\backslash T} \p(n) \leq \eps/3.
$$
Therefore, $\dtv(\p,\h)<\eps.$

\subsection{An Efficient Sampler for our Hypothesis} \label{ssec:sampler}

The learning algorithm of Section~\ref{ssec:learn-efficient} outputs a succinct description
of the hypothesis pseudo-distribution $\h$, via its DFT. This immediately provides us with an efficient
evaluation oracle for $\h,$ i.e., an $\eps$-evaluation oracle for our target SIIRV $\p.$
The running time of this oracle is linear in the size of $S,$  the effective support of the DFT.

Note that we can explicitly output the hypothesis $\h$
by computing the inverse DFT at all the points
of the support of $\h.$ However, in contrast to the effective support of
$\wh{\h},$ the support of $\h$ can be large, and this explicit description would not lead to a computationally efficient algorithm.
In this subsection, we show how to efficiently obtain an $\eps$-sampler for our unknown $k$-SIIRV $\p,$ using the
DFT representation of $\h$ as a black-box.
In particular, starting with the DFT of an accurate hypothesis $\h,$ represented via its DFT,
we show how to efficiently obtain an $\eps$-sampler for the unknown target distribution.
We remark that the efficient procedure of this section is not restricted to $k$-SIIRVs,
but is more general, applying to all  univariate discrete distributions
for which an efficient oracle for the DFT is available.

In particular, we prove the following theorem:

\begin{theorem} \label{thm:sampler}
Let $M \in \Z_+$, and $a, b \in \Z$ with $b-a=M-1.$
Let $\h: [a,b] \to \R$ be a pseudo-distribution succinctly represented via its DFT (modulo $M$), $\wh{\h}$, which is supported on a set $S,$
i.e., $\h(x)= (1/M) \cdot \sum_{\xi \in S} e(- \xi \cdot x) \widehat{\h}(\xi),$
for $x \in [a,b],$ with $0 \in S$ and $\widehat{\h}(0)=1.$
Suppose that there exists a distribution $\p$
with $\dtv(\h, \p) \leq \eps/3.$
Then, there exists an $\eps$-sampler for $\p,$ i.e., a sampler for a distribution $\q$ such that
$\dtv(\p,\q) \leq \eps,$ running in time $O(\log(M) \log(M/\epsilon) \cdot |S|).$
\end{theorem}

Combining the above with Theorem \ref{thm:FT-alg}, we get:

\begin{corollary} \label{cor:pmd-sampler}
For all $n, k \in \Z_+$ and $\eps>0,$
there is an algorithm with the following
performance guarantee: Let $\p\in \Scal_{n,k}$ be an unknown $k$-SIIRV.
The algorithm uses $O(k\log^{2}(k/\eps)/\eps^2)$
samples from $\p,$  runs in time $\widetilde{O}(k^3/\eps^2) \cdot {\log n},$
and with probability at least $9/10$ outputs an $\eps$-sampler for $\p.$
This $\eps$-sampler produces a single sample in time $O(k \log^2 (kn) \log^2 (k/\eps)).$
\end{corollary}
\begin{proof}
For the output of algorithm {\tt Learn-SIIRV}, $M= O((1+\sigma)\sqrt{\log(1/\eps)})=O(kn)$ and 
$|S| \leq |\mathcal{L'}(T)| \leq 2 M k s^{-1} \sqrt{\ln(1/T)}=O(k\log(k/\eps))$.
\end{proof}

{
Note that we can effectively reduce the $k$-SIIRV learning problem to the case of $n = \poly(k/\eps).$ 
We can use this fact as a simple bootstrapping step to eliminate the logarithmic dimension on $n$ 
in the runtime of the above described sampler. The details are deferred to Appendix~\ref{app:sampler}.}

\medskip

This section is devoted to the proof of Theorem~\ref{thm:sampler}.
We start by providing some high-level intuition.
Roughly speaking, we obtain the desired sampler by 
the Cumulative Distribution Function (CDF) corresponding to $\h.$
We use the DFT to obtain a closed form expression for the CDF of $\h,$
and then we query the CDF using an appropriate binary search procedure to sample from the distribution.
One subtle point is that $\h(x)$ is a pseudo-distribution, i.e. it is not necessarily non-negative at all points.
Our analysis shows that this does not pose any problems with correctness.

Our first lemma shows that it is sufficient to have an efficient oracle for the CDF:

\begin{lemma} \label{lem:bs}
Given a pseudo-distribution $\h$ supported on $[a,b]  \cap \Z$, $a,b \in \Z,$
with CDF $c_{\h}(x)=\sum_{i: a \leq i \leq x}  \h(i)$
(which satisfies $c_{\h}(b)=1$), and oracle access to a function $c(x)$ 
so that $|c(x)-c_{\h}(x)| < \epsilon/(10(b-a+1))$ for all $x,$ we have the following: 
If there is a distribution $\p$ with $\dtv(\h,\p) \leq \eps/3,$
there is a sampler for a distribution $\q$ with $\dtv(\p,\q) \leq \eps,$
using  $O(\log (b+1-a) + \log (1/\eps))$ uniform random bits as input,
running in time $O((D+1)(\log (b+1-a)) + \log (1/\eps)),$
where $D$ is the running time of evaluating the CDF $c(x).$
\end{lemma}

\begin{proof}
We begin our analysis by producing an algorithm that works when we are able to exactly compute $c_{\h}(x).$

We can compute an inverse to the CDF $d_{\h}:[0,1] \rightarrow [a,b] \cap \Z$, at $y \in [0,1]$,
using binary search, as follows:
\begin{enumerate}
\item We have an interval $[a',b']$, initially $[a-1,b]$, with $c_{\h}(a') \leq y \leq c_{\h}(b')$ and $c_{\h}(a') < c_{\h}(b').$
\item If $b'-a'=1$, output $d_{\h}(y)=b'.$
\item Otherwise, find the midpoint $c'=\lfloor(a'+b')/2 \rfloor.$
\item \label{step:last} If $c_{\h}(a') < c_{\h}(c')$ and $y \leq c_{\h}(c')$,
repeat with $[a',c']$; else repeat with $[c',b]$.
\end{enumerate}
The function $d_{\h}$ can be thought of as some kind of inverse to the CDF $c_{\h}:[a-1,b] \cap \Z \rightarrow [0,1]$ in the following sense:
\begin{claim} \label{clm:invariant}
The function $d_{\h}$ satisfies: For any $y \in [0, 1]$,  it holds
$c_{\h}(d_{\h}(y)-1) \leq y \leq c_{\h}(d_{\h}(y))$ and $c_{\h}(d_{\h}(y)-1) < c_{\h}(d_{\h}(y)).$
\end{claim}
\begin{proof}
Note that if we don't have $c_{\h}(a') < c_{\h}(c')$ and $y \leq c_{\h}(c')$, then $c_{\h}(c') < y \leq c_{\h}(b')$.
So, Step \ref{step:last} gives an interval $[a',b']$ which satisfies $c_{\h}(a') \leq y \leq c_{\h}(b')$ and $c_{\h}(a') < c_{\h}(b')$.
The initial interval $[a-1,b]$ satisfies these conditions since $c_{\h}(a-1) =0$ and $c_{\h}(b)=1$.
By induction, all $[a',b']$ in the execution of the above algorithm have $c_{\h}(a') \leq y \leq c_{\h}(b')$ and $c_{\h}(a') < c_{\h}(b')$.
Since this is impossible if $a'=b'$, and Step \ref{step:last} always recurses on a shorter interval, we eventually have $b'-a'=1$.
Then, the conditions $c_{\h}(a') \leq y \leq c_{\h}(b')$ and $c_{\h}(a') < c_{\h}(b')$ give the claim.
\end{proof}

 Computing $d_{\h}(y)$ requires $O(\log (b-a+1))$ evaluations of $c_{\h}$,
 and $O(\log (b-a+1))$  comparisons of $y.$ For the rest of this proof, we will use $n = b-a+1$ to denote the support size.

Consider the random variable $d_{\h}(Y)$, for $Y$ uniformly distributed in $[0,1]$, whose distribution we will call $\q'$.
When $d_{\h}(Y)=x$, we have $c_{\h}(x-1) \leq Y \leq c_{\h}(x)$, and so when $\q'(x) > 0$,
we have $\q'(x) \leq \Pr\left[c_{\h}(x-1) \leq Y \leq c_{\h}(x)\right] = c_{\h}(x)-c_{\h}(x-1) = \h(x)$.
So, when $\h(x) > 0$, we have $\h(x) \geq \q'(x)$.
But when $\h(x) \leq 0$, we have $\q'(x)=0$, since then $c_{\h}(x) < c_{\h}(x-1)$ and no $y$ has $c_{\h}(x-1) \leq y \leq c_{\h}(x)$.
So, 
we have $\dtv(\q',\h) = \sum_{x:\h(x) < 0} -\h(x) \leq \dtv(\h,\p) \leq \eps/3$.

We now show how to effectively sample from $\q'$.
The issue is how to simulate a sample from the uniform distribution on $[0,1]$ with uniform random bits.
We do this by flipping coins for the bits of $Y$ lazily. We note that we will only need to know more than $m$ bits of $Y$ if $Y$ is within $2^{-m}$ of one of the values of $c_{\h}(x)$ for some $x$. By a union bound, this happens with probability at most $n2^{-m}$ over the choice of $Y$. Therefore, for $m > \log_2(10n/\epsilon)$, the probability that this will happen is at most $\epsilon/10$ and can be ignored.

Therefore, the random variable $d_{\h}(Y')$, for $Y'$ uniformly distributed on the multiples of $2^{-r}$ in $[0,1)$ for $r = O(\log n + \log (1/\eps))$,
has distribution $\q'$ that satisfies $\dtv(\q,\q') \leq \eps/10$. Therefore, $\dtv(\p,\q') \leq \dtv(\p,\h) + \dtv(\h,\q) + \dtv(\q,\q') \leq 9\eps/10.$
This is an $\eps$-sampler that uses $O(\log n + \log (1/\eps))$ coin flips,
$O(\log n)$ calls to $c_{\h}(x)$, and has the desired running time.

We now need to show how this can be simulated without access to $c_{\h}$ and instead only having access to its approximation $c(x)$. 
The modification required is rather straightforward. 
Essentially, we can run the same algorithm using $c(x)$ in place of $c_{\h}(x).$ 
Observe that all comparisons with $Y$ will produce the same result, 
unless the chosen $Y$ is between $c(x)$ and $c_{\h}(x)$ for some value of $x.$ 
We note that because of our bounds on their difference, 
the probability of this occurring for any given value of $x$ is at most $\epsilon/(10 n).$ 
By a union bound, the probability of it occurring for any $x$ is at most $\epsilon/10.$ 
Thus, with probability at least $1-\epsilon/10$ our algorithm returns the same result 
that it would have had it had access to $c_{\h}(x)$ instead of $c(x).$ This implies that 
the variable sampled by this algorithm has variation distance at most $\epsilon/10$ 
from what would have been sampled by our other algorithm. Therefore, 
this algorithm samples a $\q$ with $\dtv(\p,\q)\leq \epsilon.$
\end{proof}

We next show that we can efficiently compute an appropriate CDF, using the DFT.

\begin{proposition} \label{prop:CDF}
For $\h$ as in Theorem \ref{thm:sampler}, there is an algorithm to compute the CDF
$c_{\h}:[a,b] \cap \Z \rightarrow [0,1]$
with $c_{\h}(x)=\sum_{i: a \leq i \leq x}  \h(i)$ to any precision $\delta>0$,
where $b-a=M-1$, $M \in \Z_+$.
The algorithm runs in time $O(|S|\log(1/\delta)).$

\end{proposition}
\begin{proof}
Recall that the PMF of $\h$ at $x \in S$ is given by the inverse DFT:
\begin{equation}  \label{eq:inverse-DFT-H}
\h(x) = \frac{1}{M} \sum_{\xi \in S} e(- \xi x/M) \widehat{\h}(\xi) \;.
\end{equation}
The CDF is given by:
\begin{align*}
c_{\h}(x)  = \frac{1}{M} \sum_{i: a \leq i \leq x} \sum_{\xi \in T} e(-\xi x/M) \widehat{\h}(\xi)
		 = \frac{1}{M} \sum_{\xi \in T} \widehat{\h}(\xi) \sum_{i: a \leq i \leq x} e(-\xi x/M) \;.
\end{align*}
When $\xi \not= 0$, the term $\sum_{i: a \leq i \leq x} e(-\xi x/M)$ is a geometric series.
By standard results on its sum, we have:
$$\sum_{i: a \leq i \leq x} e(-\xi x) = \frac{e(-\xi a/M) - e(-\xi (x+1)/M)}{1-e(-\xi/M)} \;.$$
When $\xi = 0$, $e(-\xi)=1$, and we get $\sum_{a \leq i \leq x} e(-\xi x/M) = i + 1 - a.$
In this case, we also have $\widehat{\h}(\xi)=1.$
Putting this together we have:
\begin{equation} \label{eq:CDF-1d}
c_{\h}(x) = \frac{1}{M} \left( i + 1 - a + \sum_{\xi \in S \setminus \{0\}} \wh{\h}(\xi) \frac{e(-\xi a/M) - e(-\xi (x+1)/M)}{1-e(-\xi/M)} \right) \;.
\end{equation}
Hence, we obtain a closed form expression for the CDF that can be approximated to desired precision in time $O(|S|\log(1/\delta)).$
\end{proof}

Now we can prove the main theorem of this subsection.
\begin{proof}[Proof of Theorem \ref{thm:sampler}]
By Proposition~\ref{prop:CDF},  we can efficiently calculate the CDF of $\h$.
So, we can apply Lemma~\ref{lem:bs} to this CDF.
This gives us an $\eps$-sampler for $\h.$
To find the time it takes to compute each sample, we need to substitute
$D=O\left(|S|\log(M/\eps)\right)$ from the running time of the CDF into the bound
in Lemma \ref{lem:bs}, yielding $O\left(\log M  \cdot \log(M/\eps)\right) \cdot |S|$ time. This completes the proof.
\end{proof}

\subsection{Sample--Optimal Learning Algorithm} \label{sec:opt-sample}

In this subsection, we show how to improve the sample complexity of our
learning algorithm for $k$-SIIRVs given in Section~\ref{ssec:learn-efficient}, 
and obtain an algorithm with optimal sample complexity (up to constant factors).
The basic idea behind the improvement is as follows:
In our previous analysis, we made critical use of the fact that essentially
all of the mass of the distribution in question lies in an explicit interval of length $O(s\sqrt{\log(1/\eps)}),$
where $s$ is the standard deviation. By using our Fourier learning approach, 
we were able to learn a distribution that approximated our target on this support.
In order to improve this algorithm, we observe that although it is necessary to move
 $\Omega(\sqrt{\log(1/\eps)})$ standard deviations from the mean before
 the cumulative density function (CDF) drops below $\eps,$
 the CDF has already begun to drop off exponentially after only a single standard deviation from the mean.

Unfortunately, applying a sharp threshold to our Fourier transform (as in Section~\ref{ssec:learn-efficient})
can lead to effects that fall off relatively slowly with distance.
Note that such a sharp thresholding in the Fourier domain
is equivalent to convolution with a $\sinc$ function,
which has tails proportional to $1/|x|.$
In order to correct this issue, we will instead perform our thresholding
by multiplying by a function with smooth cutoffs.
This smooth thresholding step corresponds to convolving with a function 
of width approximately $s$ with Gaussian tails.
We remark that this step has the critical effect of causing our expected errors to be much smaller
at points further from the mean, since most of our samples (within a few standard deviations of the mean)
will have little effect on our output for these points.
A careful analysis of the expected error at each point will yield our final bound.

\medskip

We will warm up in Section \ref{ssec:opt-sample-pbd}, where we describe our algorithm in the case of $2$-SIIRVs.
This will exhibit the important new ideas of this technique.
Then, in Section \ref{ssec:opt-sample-k-siirv}, we extend these results to $k$-SIIRVs,
which brings with it several technical complications, mostly arising from the fact that
we do not know a priori a good effective support for the Fourier transform.

\subsubsection{Sample Optimal Learning Algorithm for $2$-SIIRVs}\label{ssec:opt-sample-pbd}

In this subsection, we will prove the following theorem:
\begin{theorem} \label{PBDThm}
There exists an algorithm that given $N=O(\sqrt{\log(1/\eps)}/\eps^2)$ independent samples to a $2$-SIIRV $X,$
runs in time $O(N)$ and with probability at least $2/3$ outputs
a hypothesis distribution $Y$ that is within $\eps$ of $X$ in total variational distance.
\end{theorem}

Our new algorithm {\tt Learn-$2$-SIIRV-Optimal-Sample} is described in pseudocode below.
We first provide an equivalent alternative interpretation of our algorithm in terms of truncating the Fourier transform.
As in our algorithm {\tt Learn-SIIRV} of Section~\ref{ssec:learn-efficient}, 
we start by obtain approximations $\widetilde{\sigma}^2$ and $\widetilde{\mu}$ for the variance and mean.
Similarly, we output  the empirical distribution if $\widetilde{\sigma} \leq \Theta(\sqrt{\ln(1/\eps)}).$
This allows us to assume that $\widetilde{\sigma} = \Omega(\sqrt{\ln(1/\eps)}).$
(Note that this bound is not as strong as that in  {\tt Learn-SIIRV} because in the current setting we aim to use fewer samples.)

Our new learning algorithm proceeds by computing the empirical Fourier transform of $X$ and truncating it in a judiciously chosen way.
Let $\wh{G}(\xi)$, $\xi \in \R,$ be a Gaussian of standard deviation $1/\widetilde{\sigma}$ taken modulo $1.$ More specifically, let
$$
\wh{G}(\xi) = \sum_{n\in \Z} \frac{1}{\sqrt{2\pi/\widetilde{\sigma}^2}} \cdot e^{-\widetilde{\sigma}^2(n+\xi)^2/2} \;.
$$
Let $I(\xi)$ be the indicator function of the interval $[-C\widetilde{\sigma}^{-1}\sqrt{\log(1/\eps)},C\widetilde{\sigma}^{-1}\sqrt{\log(1/\eps)}],$
for $C$ a sufficiently large constant. Let $\wh{F}$ be the convolution of $I$ and $\wh{G},$ i.e., $\wh{F} = I \ast  \wh{G}.$
We note that multiplication by $\wh{F}$ is an appropriate method of thresholding.
In particular, we start by showing that $\wh{F}$ approximates $I$ in the following way:
 \begin{claim} \label{clm:Fhat}
\begin{itemize}
\item[(i)] $\wh{F}(\xi)\in[0,1]$ for all $\xi$.
\item[(ii)] $\wh{F}(\xi) \geq 1-\eps^2$ for $|\xi| \leq (C-3)\widetilde{\sigma}^{-1}\sqrt{\log(1/\eps)}$.
\item[(iii)] $\wh{F}(\xi) \leq \eps^2$ for $\frac{1}{4} \geq |\xi| \geq (C+3)\widetilde{\sigma}^{-1}\sqrt{\log(1/\eps)}$.
\end{itemize}
\end{claim}
\begin{proof} 
Note that $\wh{F}$ is the convolution of $I$ and $\wh{G}.$ We can write:
\begin{equation} \label{eq:Fhat}
\wh{F}(\xi) = \int_{\xi-C\widetilde{\sigma}^{-1}\sqrt{\log(1/\eps)}}^{\xi+C\widetilde{\sigma}^{-1}\sqrt{\log(1/\eps)}} \wh{G}(\nu) d\nu 
\leq \int_{-\infty}^\infty \frac{1}{\sqrt{2\pi/\widetilde{\sigma}^2}} e^{-\widetilde{\sigma}^2(\xi)^2/2}d\xi = 1.
\end{equation}
Clearly, this convolution is positive at all points. Thus, we get (i).

When $|\xi| \leq (C-3)\widetilde{\sigma}^{-1}\sqrt{\log(1/\eps)},$
note that the integral in (\ref{eq:Fhat}) is over
$$\nu \in [\xi-C\widetilde{\sigma}^{-1}\sqrt{\log(1/\eps)},\xi+C\widetilde{\sigma}^{-1}\sqrt{\log(1/\eps)}] \supseteq [-3\widetilde{\sigma}^{-1}\sqrt{\log(1/\eps)}, 3\widetilde{\sigma}^{-1}\sqrt{\log(1/\eps)}].$$
By standard tail bounds, the Gaussian  $\frac{s}{\sqrt{2\pi}} e^{-\widetilde{\sigma}^2\nu^2/2}$
has all but $1-\eps^2$ of its mass in the interval
$[-3\widetilde{\sigma}^{-1}\sqrt{\log(1/\eps)}, 3\widetilde{\sigma}^{-1}\sqrt{\log(1/\eps)}],$
and so $\wh{F}(\xi) \geq 1-\eps^2.$ This gives us (ii).

When $\frac{1}{4} \geq |\xi| \geq (C+3)\widetilde{\sigma}^{-1}\sqrt{\log(1/\eps)},$
the integral in (\ref{eq:Fhat}) is over
$\nu \in [\xi-C\widetilde{\sigma}^{-1}\sqrt{\log(1/\eps)},\xi+C\widetilde{\sigma}^{-1}\sqrt{\log(1/\eps)}],$
which is disjoint from the interval $[-3\widetilde{\sigma}^{-1}\sqrt{\log(1/\eps)}, 3\widetilde{\sigma}^{-1}\sqrt{\log(1/\eps)}].$
By the same bound, the Gaussian has at most $\eps^2$ of its mass
outside $[-3\widetilde{\sigma}^{-1}\sqrt{\log(1/\eps)}, 3\widetilde{\sigma}^{-1}\sqrt{\log(1/\eps)}].$
So, we deduce (iii).
\end{proof}

At a high-level, our new algorithm involves the following steps:
\begin{enumerate}
\item Let $Z$ be the empirical distribution and $\wh{Z}$ be the (continuous) Fourier transform of $Z.$
\item Let $\wh{Y}(\xi)=\wh{Z}(\xi)\wh{F}(\xi).$
\item Let $Y$ be the truncation of the inverse Fourier transform of $\wh{Y}$
to $[\widetilde{\mu}-C\widetilde{\sigma}\sqrt{\log(1/\eps)},\widetilde{\mu}+C\widetilde{\sigma}\sqrt{\log(1/\eps)}],$ for $C$ a sufficiently large constant.
\end{enumerate}
Both to aid in the performance of this computation and the theoretical analysis,
we note another way to obtain the same answer.
As $Y$ is the truncation of the inverse Fourier transform of a pointwise product of $\wh{Z}$ and $\wh{F},$
we may instead write it as the truncation to the same interval
$[\widetilde{\mu}-C\widetilde{\sigma}\sqrt{\log(1/\eps)},\widetilde{\mu}+C\widetilde{\sigma}\sqrt{\log(1/\eps)}]$
of the convolution of $Z$ and $F:\Z \rightarrow \R$, the inverse Fourier transform of $\wh{F}.$
We show below (Claim~\ref{clm:fnonhat}) that
$$F(x) = e^{-x^2/(2\widetilde{\sigma}^2)} 2C\widetilde{\sigma}^{-1}\sqrt{\log(1/\eps)}\sinc(2\pi C\widetilde{\sigma}^{-1}\sqrt{\log(1/\eps)}x) \;,$$
where $\sinc(x) \eqdef (\sin x)/x.$
Also note that $F$ can be computed explicitly to within absolute error $\delta$ in time $\poly(\log(1/\delta)),$
and thus this convolution can be computed efficiently,
yielding an alternative algorithm for computing $Y.$

\medskip

\fbox{\parbox{6.4in}{
{\bf Algorithm} {\tt Learn-$2$-SIIRV-Optimal-Sample}\\
Input: sample access to a $2$-SIIRV $X$ and $\eps>0$\\
Output: A hypothesis pseudo-distribution $Y$ that is $\eps$-close to $X$ \\

\vspace{-0.2cm}

\begin{enumerate}
\item
Draw $O(1)$ samples from $\p$
and with confidence probability $19/20$ compute: (a) $\widetilde{\sigma}^2,$
a factor $2$ approximation to $\var_{X \sim \p}[X]+1,$
and (b) $\widetilde{\mu},$ an approximation to $\E_{X \sim \p}[X]$ to within one standard deviation.


\item If $\widetilde{\sigma} \geq \Omega(1/\eps),$
draw $O(1/\eps^2)$ samples and use them to {estimate the mean and variance of $X.$
Output a discrete Gaussian with this mean and variance.}


\item Take $N=\Theta(\sqrt{\log(1/\eps)})/\eps^2 $ samples from $\p$ to get an empirical distribution $Z.$

\item If $\widetilde{\sigma}  \le O( \sqrt{\ln(1/\eps)}),$ then output $Z.$ Otherwise, proceed to next step.

\item If $M,$ the difference between the largest and smallest sample is $\Omega(\widetilde{\sigma} \sqrt{\log(1/\eps)}),$ output fail.

\item Compute $F(x)$ to within $O(\eps^4)$ for integers $|x| \leq M+C\widetilde{\sigma}\sqrt{\log(1/\eps)}.$

\item Compute the convolution $Y$ of $Z$ and $F$ using the FFT (modulo $\geq2M+2C\widetilde{\sigma}\sqrt{\log(1/\eps)}).$

\end{enumerate}
}}

\bigskip

We start by analyzing the running time of the algorithm.
First note that the first two steps run in sample-linear time, i.e., $O(1/\eps^2).$
We now focus on the running time of the remaining steps.
Note that computing the empirical distribution $Z$ takes time $O(N).$
Computing the values of $F(x)$ in Step~6 up to an additive error $\poly(\eps)$ can be done in time $M \polylog(1/\eps),$
where $M = O(\widetilde{\sigma} \sqrt{\log(1/\eps)}) = \widetilde{O}(1/\eps).$
Computing the convolution is done using the FFT modulo a power of two that is $\Theta(M),$
and so can be done in time $O(M \log M).$
So, the overall running time is
$O(N+M \log M \poly(\log(1/\eps))) = O(N).$

We now proceed to show correctness.
In the proof of Theorem~\ref{thm:FT-alg}, we argued that $O(1)$ samples
suffice to get that with high probability $\widetilde{\sigma}$ and $\widetilde{\mu}$ satisfy the desired bounds.
We condition on this event.
We claim that when $\widetilde{\sigma}$ is small, namely $O(\sqrt{\ln(1/\eps)}),$
the empirical distribution suffices. This follows from the fact that the empirical estimate of a discrete 
distribution $\p$ has expected variation distance $ \le \eps$ from $\p$ after $O(\|\p\|_{1/2}/\eps^2)$ samples.
By an application of Bernstein's inequality (see Lemma~\ref{lem:half-norm}) it follows that a $2$-SIIRV
with standard deviation $\sigma$ has $1/2$-norm bounded from above by $O(\sigma+1).$
This proves our claim. 

We also note that if the standard deviation of $X$ is $\Omega(1/\eps),$ then $X$ is $\eps$-close
to a discretized Gaussian with the same mean and variance.  Indeed, for any $2$-SIIRV with mean $\mu$
and standard deviation $\sigma,$ we have $\dtv(X, G) \le O(1/\sigma),$ where $G \sim  Z(\mu, \sigma^2).$
(See, e.g., Theorem~7.1 of~\cite{CGS11}.) In this case, we claim that Step~2 of the algorithm outputs an $\eps$-accurate hypothesis.
Indeed, by Lemma~6 of~\cite{DDS15-journal} it follows that with $O(1/\eps^2)$ samples from a discrete distribution,
we can obtain (in sample-linear time) estimates $\wh{\mu}$ and $\wh{\sigma}$ such that with high constant probability
$|\wh{\mu} - \mu| \le \eps \sigma$ and $|\sigma^2 - \wh{\sigma}^2| \le \eps \sigma^2 \sqrt{4+1/\sigma^2}.$
Proposition~\ref{prop:normal-dist} completes the proof of our claim.

So, we henceforth assume that $\widetilde{\sigma}$ is $\Omega( \sqrt{\ln(1/\eps)})$ and $O(1/\eps).$
We now proceed with the main part of the analysis. We start with the following simple claim:
\begin{claim} \label{clm:fnonhat}
We have that
$$F(x) = e^{-x^2/(2\widetilde{\sigma}^2)} 2C\widetilde{\sigma}^{-1}\sqrt{\log(1/\eps)}\sinc(2\pi C\widetilde{\sigma}^{-1}\sqrt{\log(1/\eps)}x) \;,$$
for all $x \in \Z.$
Also, $|F(x)| = O(\widetilde{\sigma}^{-1})\sqrt{\log(1/\eps)}\exp(-\Omega((x/\widetilde{\sigma})^2)).$
\end{claim}
\begin{proof}
As $\wh{F}$ is a convolution of functions, 
$F(x)$ is the pointwise product of $G(x)$ the inverse Fourier transform of $\wh{G}(\xi)$ with $S(x),$ 
the inverse Fourier transform of $I(\xi).$
We define $G(x):=e^{-x^2/(2 \widetilde{\sigma}^2)}$ and 
$S(x):=2C\widetilde{\sigma}^{-1}\sqrt{\log(1/\eps)}\sinc(2\pi C\widetilde{\sigma}^{-1}\sqrt{\log(1/\eps)}x).$
Standard results for the Fourier transform of the Gaussian and $\sinc(x)$ give us the result.
Since $|\sinc(x)| \leq 1$ for all $x,$ we have that $|F(x)| = O(\widetilde{\sigma}^{-1})\sqrt{\log(1/\eps)}\exp(-\Omega((x/\widetilde{\sigma})^2)).$
\end{proof}
In order to show the correctness of our algorithm, we will need to introduce a new distribution, $Y'.$ 
We let $Y'$ be the truncated inverse Fourier transform of the pointwise product of $\wh{F}$ with $\wh{X}$ 
(note that $Y$ differs from $Y'$ by using $\wh{Z}$ instead of $\wh{X}$). 
We begin by showing that $\dtv(X,Y')$ is small. To do this, we let $\wh{Y'}=\wh{X}\wh{F}.$
\begin{claim} 
We have that
{$\dtv(X,Y')=O(\eps^2 \sqrt{\log(1/\eps)}).$}
\end{claim}
\begin{proof}
Note that
$$
\wh{X}(\xi)-\wh{Y'}(\xi) = \wh{X}(\xi)(1-\wh{F}(\xi)).
$$
If $[\xi] \leq (C-3)\widetilde{\sigma}^{-1}\sqrt{\log(1/\eps)},$
by Claim \ref{clm:Fhat} (ii), we have $1-\wh{F}(\xi) \leq \eps^2.$
Since $|\wh{X}(\xi)| = |\E[e(X\xi)]| \leq 1,$
in this case we have $|\wh{X}(\xi)-\wh{Y'}(\xi)| \leq \eps^2.$
Otherwise, if $[\xi] \geq (C-3)\widetilde{\sigma}^{-1}\sqrt{\log(1/\eps)},$
by Lemma \ref{FourierSupportLem} part (i) it follows that $|\wh{X}(\xi)| \leq \exp(-\Omega(\widetilde{\sigma}^2[\xi]^2)).$
Since $0 \leq 1-\wh{F}(\xi) \leq 1,$
in this case we have $ |\wh{X}(\xi)-\wh{Y'}(\xi)| \leq |\wh{X}(\xi)||1-\wh{F}(\xi)| \leq  \exp(-\Omega(\widetilde{\sigma}^2[\xi]^2)).$
Therefore,
\begin{align*}
|\wh{X}-\wh{Y'}|_1 = & \int_{-1/2}^{1/2} |\wh{X}(\xi)-\wh{Y'}(\xi)| d\xi \\
					= &  \int_{-1/2}^{-(C-3)\widetilde{\sigma}^{-1}\sqrt{\log(1/\eps)}} |\wh{X}(\xi)-\wh{Y'}(\xi)| d\xi + 
					\int_{-(C-3)\widetilde{\sigma}^{-1}\sqrt{\log(1/\eps)}}^{(C-3)\widetilde{\sigma}^{-1}\sqrt{\log(1/\eps)}} |\wh{X}(\xi)-\wh{Y'}(\xi)| d\xi \\
					+ & \int_{-(C-3)\widetilde{\sigma}^{-1}\sqrt{\log(1/\eps)}}^{1/2} |\wh{X}(\xi)-\wh{Y'}(\xi)| d\xi \\
					\leq & \eps^2 \cdot 2(C-3)\widetilde{\sigma}^{-1}\sqrt{\log(1/\eps)} + 
					2 \int_{(C-3)\widetilde{\sigma}^{-1}\sqrt{\log(1/\eps)}}^{1/2}  \exp(-\Omega(\widetilde{\sigma}^2[\xi]^2)) d\xi \\
					= & O(\eps^2\sqrt{\log(1/\eps)}/\widetilde{\sigma}) + 
					\sqrt{2 \pi}/s \cdot \Pr_{W \sim N(0,\widetilde{\sigma}^{-2})}\left[ |W| \geq  (C-3)\widetilde{\sigma}^{-1}\sqrt{\log(1/\eps)} \right]\\
					\leq & O(\eps^2\sqrt{\log(1/\eps)}/\widetilde{\sigma}) \;.
\end{align*}
Taking an inverse Fourier transform implies that $|X-Y'|_\infty=O(\eps^2/\widetilde{\sigma})$,
within the domain of truncation. Since this domain has size $O(\widetilde{\sigma} \sqrt{\log(1/\eps)}),$
we have that the $L_1$ error between $X$ and $Y'$
within this domain is $O(\sqrt{\log(1/\eps)} \eps^2).$
However, both $X$ and $Y'$ have at most $O(\eps^2)$ mass outside of this domain,
and therefore we have that $\dtv(X,Y')=O(\eps^2 \sqrt{\log(1/\eps)}).$
\end{proof}
It remains to bound from above $\dtv(Y,Y')$.
In particular, we will show that $\dtv(Y,Y')$ has expectation $O(\eps).$
Then, by decreasing $\eps$ by a constant factor
and applying the Markov and triangle inequalities,
we will have that $\dtv(X,Y)<\eps$ with probability at least $2/3.$
\begin{proposition}
We have that $\E\left[\dtv(Y,Y') \right] \leq O(\eps).$
\end{proposition}
\begin{proof}
Recall that $Y$ is a the convolution of $Z$ with $F.$
If we consider our samples to be random variables $X_{(1)}, \ldots, X_{(N)}$
each of which is an i.i.d. copy of $X,$
we can express $Y(p)$ for a given $p$ as a random variable:
$
Y(p) = \frac{1}{N} \sum_{i=1}^N F(p-X_{(i)}) \;,
$
for $a \leq p \leq b,$
where $a= \widetilde{\mu}-C\widetilde{\sigma}\sqrt{\log(1/\eps)}$ and $b=\widetilde{\mu}+C\widetilde{\sigma}\sqrt{\log(1/\eps)}.$
Note that the expectation of $Y(p)$ is
$$\frac{1}{N} \sum_{i=1}^N  \E_X[F(p-X)]= \E_X[F(p-X)]=Y'(p).$$
Therefore, we have that $\E[|Y(p)-Y'(p)|] =O(\sqrt{\var(Y(p))}).$
For the variance we have the following sequence of (in)equalities:
\begin{align*}
& \var[Y(p)] =  \var[F(p-X)]/N =  \E\left[\left(F(p-X) - \sum_{q=a}^b F(p-q) X(q)\right)^2\right]/N \\
= & \sum_{r=a}^b (X(r)/N) \cdot \left(F^2(p-r)   \vphantom{\sum_{q=a}^b} \right.\\
+ & \left.  \sum_{q=a}^b \Big( F^2(p-q) X(q)^2 - 2 F(p-r)F(p-q) X(q) + 2\littlesum_{q' \neq q} F(p-q)F(p-q')X(q)X(q') \Big) \right) \\
= & 1/N \cdot \sum_{q=a}^b F^2(p-q) (X(q)-X(q)^2) 
\leq  1/N \cdot \sum_{q=a}^b F^2(p-q) X(q) \;.
\end{align*}
We claim that this quantity will become small as $p$ moves away from $\mu.$
Intuitively, this should be the case because for $p$ far from $\mu,$
then for all $q$ either $|p-q|$ will be large or $|q-\mu|$ will be large.
In the former case, $F(p-q)$ is small, and in the latter $X(q)$ is.
In order to properly analyze this quantity, we will have to group up these errors for $p$ in blocks of size $\widetilde{\sigma}.$
In particular, we have that
\begin{align*}
& \sum_{p=\mu+t\widetilde{\sigma}}^{\mu+(t+1)\widetilde{\sigma}}  \E[|Y(p)-Y'(p)|] = \sum_{p=\mu+t\widetilde{\sigma}}^{\mu+(t+1)\widetilde{\sigma}} O(\sqrt{\var(Y(p))})\\
 & = O(1/\sqrt{N}) \sum_{p=\mu+t\widetilde{\sigma}}^{\mu+(t+1)\widetilde{\sigma}} \sqrt{\sum_{q=a}^b F^2(p-q) X(q)}\\
 & = O(\sqrt{\widetilde{\sigma}/N}) \sqrt{\sum_{p=\mu+t\widetilde{\sigma}}^{\mu+(t+1)\widetilde{\sigma}}  \sum_{q=a}^b F^2(p-q) X(q)} \tag*{(by Cauchy-Schwartz)}\\
& \leq O(\sqrt{\widetilde{\sigma}/N}) \sqrt{\sum_{r=\mu+t\widetilde{\sigma}-a}^{b-\mu-(t+1)\widetilde{\sigma}} F^2(r) \sum_{q=\mu+t\widetilde{\sigma}-r}^{\mu+(t+1)\widetilde{\sigma}-r} X(q)} \\
& \leq O(\sqrt{\widetilde{\sigma}/N}) \sqrt{\sum_{r=-\infty}^\infty F^2(r) \exp(-\Omega(|t\widetilde{\sigma}-r|/\widetilde{\sigma})^2)}  \tag*{(by Bernstein's inequality)}\\
& = O(\sqrt{\widetilde{\sigma}/N}) \sqrt{\sum_{r=-\infty}^\infty S^2(r) \exp(-\Omega(((t\widetilde{\sigma}-r)/\widetilde{\sigma})^2 + (r/\widetilde{\sigma})^2))} \\
& = O(\sqrt{\widetilde{\sigma}/N}) \sqrt{\sum_{r=-\infty}^\infty  S^2(r) \exp(-\Omega(t^2))}\\
& = O(\sqrt{\widetilde{\sigma}/N}) \exp(-\Omega(t^2)) \sqrt{\int_{\xi=0}^1  I^2(\xi) d\xi} \tag*{(by Plancherel's Theorem)}\\
& = O(\sqrt{\widetilde{\sigma}/N}) \exp(-\Omega(t^2)) \widetilde{\sigma}^{-1}\widetilde{\sigma}^{1/2}\log^{1/4}(1/\eps) = O(\log^{1/4}(1/\eps)/\sqrt{N})\exp(-\Omega(t^2)).
\end{align*}
Summing over $t$ gives that
$$
\E[\dtv(Y,Y')] = O(\log^{1/4}(1/\eps)/\sqrt{N}) = O(\eps) \;,
$$
for $N=\sqrt{\log(1/\eps)}/\eps^2.$
This completes the proof.
\end{proof}

\subsubsection{Sample Optimal Learning Algorithm for $k$-SIIRVs}\label{ssec:opt-sample-k-siirv}

\begin{theorem}\label{kSIIRVThm}
For $\eps \le 1/\poly(k)$, there exists an algorithm that given $O(k\sqrt{\log(1/\eps)}/\eps^2)$
independent samples from a $k$-SIIRV, $X$, with probability at least $2/3$
outputs a hypothesis distribution $Y$ that is within $\eps$ of $X$ in total variational distance.
\end{theorem}

The proof of this theorem is somewhat analogous to that of Theorem~\ref{PBDThm}.
However, it should be noted that the runtime of this algorithm is not given.
This is because the runtime of the simplest such algorithm is actually exponential in $k.$
The difficulty is that while in the $2$-SIIRV case we could determine the effective support
 of the Fourier transform just from the standard deviation,
 in the case of $k$-SIIRVs this is not the case.
 In essence, our algorithm will first guess this effective support (of which there are exponentially many possibilities),
 and then given this guess will run an appropriate algorithm.
 At the end, we will need to run a standard tournament procedure (e.g., ~\cite{DL:01})
 to determine which of these guesses lead to the closest approximation to $X.$
 Since the number of possibilities is $2^{O(k)}$ (see Claim~\ref{claim:w}),
 the sample complexity of this tournament is $O(k/\eps^2).$

As in Algorithm {\tt Learn-SIIRV}, we begin by estimating the mean and variance with $O(1)$ samples,
producing estimates $\widetilde{\mu}$ and $\widetilde{\sigma}^2$ with $(\var[X]+1)/2 \leq \widetilde{\sigma}^2 \leq 2(\var[X]+1)$
and $\E[X]-\widetilde{\sigma} \leq \widetilde{\mu} \leq \E[X]+\widetilde{\sigma}$.
Again, if $\widetilde{\sigma} =O(k\sqrt{\log(1/\eps)})$, we output
the empirical distribution after taking  $O(k\sqrt{\log(1/\eps)}/\eps^2)$ samples.
Our upper bound on the $1/2$-norm of $k$-SIIRVs (Lemma~\ref{lem:half-norm}) implies that this step
gives an $\eps$-accurate hypothesis.
This allows us to assume that $\widetilde{\sigma}= \Omega(k\sqrt{\log(1/\eps)})$. We will assume this throughout the remainder of our analysis.

Once again, under these assumptions, we can use Bernstein's inequality to prove concentration bounds for $X$:
\begin{lemma}\label{concentrationLemma}
Suppose that $\widetilde{\sigma} \geq Ck\sqrt{\log(1/\eps)}$, for $C$ sufficiently large,
then for all $t\geq 0$ we have that
$$
\Pr(|X-\widetilde{\mu}| > (2+t)\widetilde{\sigma}) \leq \exp(-\Omega(t^2))+\eps^2.
$$
\end{lemma}
\begin{proof}
We assumed that $|\mu-\widetilde{\mu}| \leq \widetilde{\sigma}$. So, 
if $|X-\widetilde{\mu}| \geq (2+t)\widetilde{\sigma})$, then $|X-\mu| \geq (1+t)\widetilde{\sigma}$.
Bernstein's inequality gives that
$$\Pr(X-\E[X] \geq (1+t)\widetilde{\sigma}) \leq \exp\left(\frac{-\frac{1}{2}(1+t)^2\widetilde{\sigma}^2}{\Var[X]+ \frac{1}{3}k}\right).$$
Since $\widetilde{\sigma}= \Omega(k\sqrt{\log(1/\eps)})=\Omega(k)$ and $\Var[X] = O(\widetilde{\sigma}^2)$, we have that $\Var[X]+ \frac{1}{3}k=O(\widetilde{\sigma}^2)$.
\end{proof}

In particular, this implies that with probability $1-2\eps^2$ that $|X-\widetilde{\mu}|=O(\widetilde{\sigma}\sqrt{\log(1/\eps)}).$
Next, we will recall concentration bounds on the Fourier transform of $X$. To do so, we first devise some notation.
Let
$
X=\sum_{i=1}^n X_i \;,
$
where $X_i$ are independent $k$-IRVs. We let $p_{i,j}$ be the probability that two independent copies of $X_i$
have absolute difference $j$. We let $v_j=\sum_{i=1}^n p_{i,j}$. In terms of this,
we restate Equations (\ref{ineq:this}) and (\ref{eq:varj}) from the proof of Lemma~\ref{FourierSupportLem} as
$$
|\wh{X}(\xi)| = \exp\left(-\Omega\left(\sum_{j=1}^{k-1} v_j [j\xi]^2\right) \right) \;,
$$
and
$$
\var(X) = \sum_{j=1}^{k-1} j^2 v_j \;.
$$
Finally, we note that we can find some particular good scale to consider. In particular, we note
\begin{lemma} \label{lem:power-two-j}
There exists an $m \in [k]$ so that
$
\sum_{j=m}^{2m} v_j = \Omega(\widetilde{\sigma}/k)^2.
$
\end{lemma}
\begin{proof}
We assume for sake of contradiction that this is not the case. We have that
$$
\sum_{j=m}^{2m} v_j < c (\widetilde{\sigma}/k)^2 \;,
$$
where $c$ is a sufficiently small constant. This implies that
$$
\sum_{j=m}^{2m} j^2 v_j < c (2\widetilde{\sigma}m/k)^2.
$$
Summing over $m$ powers of $2$ less than or equal to $k$, we find that
$$
\sum_{j=1}^k j^2 v_j < c \sum_{\ell=1}^{\lfloor \log_2(k) \rfloor} (2s/k)^2 4^\ell \leq c  4^{\log_2(k) +1}(4\widetilde{\sigma}^2/k^2) = 16c\widetilde{\sigma}^2.
$$
However, we know that
$$
\sum_{j=1}^k j^2 v_j = \var(X) = \Theta(\widetilde{\sigma}^2) \;.
$$
This yields a contradiction for $c$ sufficiently small.
\end{proof}
Our algorithm will begin by guessing a value for $m.$
We assume throughout the following that $m$ represents such an integer.
Furthermore, we assume that our algorithm has guessed values $w_m,w_{m+1},\ldots,w_{2m}$
so that $w_i \leq v_i$ for all $i$ and $\sum_{j=m}^{2m} w_j = \Omega(\widetilde{\sigma}/k)^2.$
\begin{claim} \label{claim:w}
Given $m$ and $\widetilde{\sigma}$, this can be done by considering only $2^{O(k)}$ possible vectors of $w$'s.
\end{claim}
\begin{proof}
By Lemma \ref{lem:power-two-j}, we have that $\sum_{j=m}^{2m} v_j = \Omega(\widetilde{\sigma}/k)^2.$
Suppose concretely that $C'$ is a constant such that we always have $\sum_{j=m}^{2m} v_j \geq C'(\widetilde{\sigma}/k)^2.$
Then, we claim that there is some set of non-negative integers $a_j,$
for $m \leq j \leq 2m$ such that $\sum_{j=m}^{2m} a_j=m$
and $v_j \geq  (a_j/(m+1)) \cdot C'(\widetilde{\sigma}/k)^2/2.$
In particular, take $a_j=\lfloor \frac{v_j 2(m+1)}{C'(\widetilde{\sigma}/k)^2}\rfloor$
then $|(C'(\widetilde{\sigma}/k)^2/2(m+1))\left(\sum_{j=m}^{2m} a_j \right) - v_j| \leq C'(\widetilde{\sigma}/k)^2/2,$
and so $\sum_{j=m}^{2m} a_j \geq \frac{|v_j - C'(\widetilde{\sigma}/k)^2/2|}{C'(\widetilde{\sigma}/k)^2/2(m+1)} \geq m+1.$

If we guess such integers, then we can set $w_j = (a_j/(m+1)) \cdot C'(\widetilde{\sigma}/k)^2/2$
and have $v_j \geq w_j$ and $\sum_{j=m}^{2m} w_j \geq C'(\widetilde{\sigma}/k)^2/2= \Omega(\widetilde{\sigma}/k)^2.$
There are  ${n+k \choose n}$ $k$-vectors $a$ of non-negative integers summing to $n.$
So in our case, there are ${2m \choose m} \leq 2^{2m} \leq 2^{2k}$ possible combinations of $w_j.$
\end{proof}

We then have that
$$
|\wh{X}(\xi)| \leq B(\xi) \eqdef \exp\left(-\Theta\left(\sum_{j=m}^{2m} w_j [j\xi]^2\right) \right).
$$
We have the following simple lemma about $B$:
\begin{lemma}
If $|\xi-\xi'| < 1/{(6m)}$, then
$$
B(\xi)B(\xi') = \exp(-\Omega(\widetilde{\sigma}^2(\xi-\xi')^2 m^2/k^2)).
$$
\end{lemma}
\begin{proof}
Firstly, we show that for each $m \leq j \leq 2m$,
either we have $[j\xi] \geq j|\xi-\xi'|/2$ or $[j\xi'] \geq j|\xi-\xi'|/2$. If $[j\xi] \leq j|\xi-\xi'|/2$,
then there is an integer $i$ such that $|j\xi-i| \leq |j\xi-j\xi'|/2$
and so $|j\xi'-i| \geq |j\xi-j\xi'| - |j\xi-i| \geq |j\xi-j\xi'|/2$.
But we also have $|j\xi'-i| \leq |j\xi-j\xi'| + |j\xi-i| \geq 3|j\xi-j\xi'|/2 \leq 3j/12m \leq \frac{1}{2}$.
So, $i$ is still one of the closest integers to $j\xi'$ and $[j\xi']=|j\xi'-i| \geq |j\xi-j\xi'|/2.$

Thus, we have:
\begin{align*}
B(\xi)B(\xi') = & \exp\left(-\Theta\left(\sum_{j=m}^{2m} w_j [j\xi]^2 + [j\xi'^2]\right) \right) \\
				\leq & \exp\left(-\Omega\left(\sum_{j=m}^{2m} w_j j^2(\xi-\xi')^2\right) \right) \\
				\leq & \exp\left(-\Omega\left((\sum_{j=m}^{2m} w_j) m^2(\xi-\xi')^2\right) \right) \\
				\leq & \exp(-\Omega(\widetilde{\sigma}^2(\xi-\xi')^2)m^2/k^2) \;,
\end{align*}
where the final line follows since we guessed $w$ so that $\sum_{j=m}^{2m} w_j = \Omega(\widetilde{\sigma}/k)^2.$
\end{proof}

This implies that within each interval of length $1/(6m),$
$B(\xi)$ is bounded by an appropriate Gaussian.
In particular, for $0\leq i < 6m$, let $I_i$ be the interval $[i/6m,(i+1)/6m]$,
and let $\xi_i$ be the element of $I_i$ at which $B$ is maximized.
Since $\sum_{j=m}^{2m} w_j [j\xi]^2$ is a piecewise quadratic,
we can easily calculate its minima $\xi_i$ on each $I_i$
given $w_j$ for $m \leq j \leq 2m.$
As a corollary of the above, we have:
\begin{corollary}
For $\xi\in I_i$, we have that
$$
|\wh{X}(\xi)| \leq \exp(-\Omega(\widetilde{\sigma}^2(\xi-\xi_i)^2)m^2/k^2).
$$
\end{corollary}
\begin{proof} We can write
$|\wh{X}(\xi)| \leq B(\xi) \leq \sqrt{B(\xi)B(\xi_i)} = \exp(-\Omega(\widetilde{\sigma}^2(\xi-\xi_i)^2)m^2/k^2) \;.$
\end{proof}

From this point onwards, our analysis is nearly identical to that from the previous subsection.
We will need the function $I(\xi)$ to be small not just near $0,$
but also near all of the $\xi_i$'s so that $\wh{F}$ will be close to $1$
on the effective support of $\wh{X}.$
This has the effect of making its inverse Fourier transform
a sum of $O(m)$ $\sinc$ functions rather than a single one.
This in turn will increase the size of the $F$ by a factor of $m,$
which is where the final additional factor of $k$ in our sample size comes from.

Our algorithm depends on taking the empirical Fourier transform of $X$
and truncating it in a judiciously chosen way.
Let $\wh{G}(\xi)$ be a Gaussian of standard deviation $1/\widetilde{\sigma}$ taken modulo $1.$ In particular,
$$
\wh{G}(\xi) = \sum_{n\in \Z} \frac{1}{\sqrt{2\pi/\widetilde{\sigma}^2}} e^{-\widetilde{\sigma}^2(n+\xi)^2/2}.
$$
Let $I(\xi)$ be the indicator function that is $1$ if and only if $\xi$ is within
$Ck\widetilde{\sigma}^{-1}\sqrt{\log(1/\eps)}/m$
of one of the $\xi_i$ modulo $1,$ for $C$ a sufficiently large constant.
Let $\wh{F}$ be the convolution of $I$ and $\wh{G}.$
As before, $\wh{F}$ approximates $I$ in that:
\begin{claim}
\begin{itemize}
\item[(i)] $\wh{F}(\xi)\in[0,1]$ for all $\xi$.
\item[(ii)] $\wh{F}(\xi) \geq 1-\eps^2/k$ for $\xi$ within $(Ck/m-3)\widetilde{\sigma}^{-1}\sqrt{\log(1/\eps)}$ of some $\xi_i$.
\item[(iii)] $\wh{F}(\xi) \leq \eps^2/k$ for $\xi$ not within $(Ck/m+3)\widetilde{\sigma}^{-1}\sqrt{\log(1/\eps)}$ of any $\xi_i$.
\end{itemize}
\end{claim}
\begin{proof}
Note that $\wh{F}$ is the convolution of $I$ and $\wh{G}.$
$I(x)$ is the indicator function of some set $T$.
Explicitly, we have:
\begin{align*}
\wh{F}(\xi) =  \int_{T} \wh{G}(\nu) d\nu 
			\leq & \int_0^1 \wh{G}(\nu)d\nu 
            =  \int_{-\infty}^\infty \frac{1}{\sqrt{2\pi/\widetilde{\sigma}^2}} e^{-\widetilde{\sigma}^2(\nu)^2/2}d\nu = 1.
\end{align*}
This gives (i).

For (ii), we note that since $I$ contains the interval
$[\xi_i-Ck\widetilde{\sigma}^{-1}\sqrt{\log(1/\eps)}/m, \xi_i+Ck\widetilde{\sigma}^{-1}\sqrt{\log(1/\eps)}/m],$
we have
$$\wh{F}(\xi) \geq \int_{\xi - \xi_i-Ck\widetilde{\sigma}^{-1}\sqrt{\log(1/\eps)}/m}^{\xi-\xi_i+Ck\widetilde{\sigma}^{-1}\sqrt{\log(1/\eps)}/m} \frac{s}{\sqrt{2\pi}} e^{-\widetilde{\sigma}^2([\nu])^2/2} d\nu .$$
Since $|\xi-\xi_i| \leq (Ck/m-3)\widetilde{\sigma}^{-1}\sqrt{\log(1/\eps)},$
this interval contains $[-3\widetilde{\sigma}^{-1}\sqrt{\log(1/\eps)},3\widetilde{\sigma}^{-1}\sqrt{\log(1/\eps)}],$
and so
$$\wh{F}(\xi) \geq \int_{3\widetilde{\sigma}^{-1}\sqrt{\log(1/\eps)}}^{3\widetilde{\sigma}^{-1}\sqrt{\log(1/\eps)}} \frac{s}{\sqrt{2\pi}} e^{-\widetilde{\sigma}^2\nu^2/2} d\nu   \geq 1 - O(\eps^3) \geq 1-\eps^2/k \;,$$
by standard bounds on the Gaussian.

For (iii), we note that $T$
is disjoint from the set $[\xi-3\widetilde{\sigma}^{-1}\sqrt{\log(1/\eps)}, \xi+3\widetilde{\sigma}^{-1}\sqrt{\log(1/\eps)}].$
We have
\begin{align*}
\wh{F}(\xi) = & \int_{\nu\in \R, \nu-\xi \pmod{\Z} \in T} \frac{1}{\sqrt{2\pi/\widetilde{\sigma}^2}} e^{-\widetilde{\sigma}^2(\nu)^2/2}d\nu \\ 
\leq & \int_{|\nu| \geq 3\widetilde{\sigma}^{-1}\sqrt{\log(1/\eps)}} \frac{1}{\sqrt{2\pi/\widetilde{\sigma}^2}} e^{-\widetilde{\sigma}^2(\nu)^2/2}d\nu = O(\eps^3) \leq \eps^2/k.
\end{align*}
\end{proof}
Our algorithm is now quite simple to state and works as follows:
\begin{enumerate}
\item Let $Z$ be the empirical distribution and $\wh{Z}$ be the Fourier transform of $Z.$
\item Let $\wh{Y}$ be the pointwise product of $\wh{Z}$ with $\wh{F}$.
\item Let $Y$ be the truncation of the inverse Fourier transform of $\wh{Y}$
to $\left[ \mu-C\widetilde{\sigma}\sqrt{\log(1/\eps)},\mu+C\widetilde{\sigma}\sqrt{\log(1/\eps)} \right],$ for $C$ a sufficiently large constant.
\end{enumerate}
Both to aid in the performance of this computation and in its theoretical analysis,
we note another way to obtain the same answer.
As $Y$ is the truncation of the inverse Fourier transform of a pointwise
product of $\wh{Z}$ and $\wh{F}$, we may instead write it as the truncation of the convolution of $Z$ and $F,$
the inverse Fourier transform of $\wh{F}.$ As $\wh{F}$ is a convolution of functions,
$F(x)$ is the pointwise product of $G(x)$
(a Gaussian of standard deviation $\Theta(\widetilde{\sigma})$, normalized to have size $1$ at the origin)
with $S(x),$ an explicit combination of $\sinc$ functions.
Note that $F$ can be computed explicitly, and thus this convolution can be computed in polynomial time.

In order to analyze the correctness,
we will need to introduce a new distribution, $Y'.$
We let $Y'$ be the truncated inverse Fourier transform
of the pointwise product of $\wh{F}$ with $\wh{X}$ (note that $Y$ differs by using $\wh{Z}$ instead of $\wh{X}$).
We begin by showing that $\dtv(X,Y')$ is small. To do this, we let $\wh{Y'}=\wh{X}\wh{F}.$
\begin{claim}
We have that
$$
|\wh{X}-\wh{Y'}|_1 = O(\eps^2{\sqrt{\log(1/\eps)}}/\widetilde{\sigma}).
$$
\end{claim}
\begin{proof}
We similarly use the fact that
$$
\wh{X}(\xi)-\wh{Y'}(\xi) = \wh{X}(\xi)(1-\wh{F}(\xi)).
$$
If $[\xi - \xi_i]$ is at most $(Ck/m-3)\widetilde{\sigma}^{-1}\sqrt{\log(1/\eps)}$ for some $i,$
the above expression has absolute value at most $\eps^2/k$ because $1-\wh{F}(\xi)$ does.
Otherwise, it has absolute value $\exp(-\Omega(\widetilde{\sigma}^2[\xi-\xi_{i_0}]^2m^2/k^2)),$
where $\xi_{i_0}$ is such that $i_0 \in \mathrm{argmin}_i [\xi - \xi_i].$ {Next, we combine these bounds and integrate.

We consider intervals $[a_i,b_i]$ with $b_i=a_{i+1}$ for $1 \leq i < 6m$ 
and $b_{6m}=a_1+1$ such that $\xi_i \in [a_i,b_i]$ and for any $x + \Z$ in $[a_i,b_i]+\Z$ is at least as close 
to $\xi_i+\Z$ than to any $\xi_j+\Z$ for $j \neq i.$  
\begin{align*}
|\wh{X}-\wh{Y'}|_1 = & \int_{a_1}^{b_{6m}} |\wh{X}(\xi)-\wh{Y'}(\xi)| d\xi \\
					= &  \sum_{i=1}^{6m} \left( \int_{a_i}^{\max \{\xi_i- (Ck/m-3)\widetilde{\sigma}^{-1}\sqrt{\log(1/\eps)}, a_i \}} |\wh{X}(\xi)-\wh{Y'}(\xi)| d\xi \right. \\
					+ & \int_{\max \{\xi_i- (Ck/m-3)\widetilde{\sigma}^{-1}\sqrt{\log(1/\eps)}, a_i \}}^{{\min \{\xi_i+ (Ck/m-3)\widetilde{\sigma}^{-1}\sqrt{\log(1/\eps)}, b_i \}}} |\wh{X}(\xi)-\wh{Y'}(\xi)| d\xi \\
					+ & \left. \int_{\min \{\xi_i+ (Ck/m-3)\widetilde{\sigma}^{-1}\sqrt{\log(1/\eps)}, b_i \}}^{b_i} |\wh{X}(\xi)-\wh{Y'}(\xi)| d\xi  \right) \\
					\leq  & O(1) \cdot \sum_{i=1}^{6m} \left( \Pr_{W \sim N(0,\widetilde{\sigma}^{-2}k^2/m^2)}\left[|W| \geq  (Ck/m-3)\widetilde{\sigma}^{-1}\sqrt{\log(1/\eps)}\right] \right. \\
					+ & \left. (\eps^2/k) \cdot 2 (Ck/m-3)\widetilde{\sigma}^{-1}\sqrt{\log(1/\eps)} 
					\vphantom{\Pr_{W \sim N(0,\widetilde{\sigma}^{-2}k^2/m^2)}} \right) \\ \leq & O(\eps^2\sqrt{\log(1/\eps)}/\widetilde{\sigma}) \;.
\end{align*}}
This completes the proof.
\end{proof}
Taking an inverse Fourier transform implies that $|X-Y'|_\infty=O(\eps^2\sqrt{\log(1/\eps)}/\widetilde{\sigma}),$
at least within the domain of truncation. Since this domain has size $O(\widetilde{\sigma} \sqrt{\log(1/\eps)}),$
we have that the $L_1$ error between $X$ and $Y'$ within this domain is $O(\sqrt{\log(1/\eps)} \eps^2).$
However, both $X$ and $Y'$ have at most $O(\eps^2)$ mass outside of this domain,
and therefore we have that
$$\dtv(X,Y')=O(\eps^2 \log(1/\eps)).$$
It remains to bound $\dtv(Y,Y').$
In particular, we will show that it has expectation $O(\eps).$
Then, by decreasing $\eps$ by a constant factor
and applying Markov's and triangle inequalities,
we will have that $\dtv(X,Y)<\eps,$ with probability at least $2/3$.
\begin{proposition} We have that
$\E[\dtv(Y,Y')] \leq O(\eps).$
\end{proposition}
\begin{proof}
Recall that $Y$ is a the convolution of $Z$ with $F.$
If we consider our samples to be random variables $X_{(1)}, \ldots, X_{(N)}$
each of which is an i.i.d. copy of $X,$
we can express $Y(p)$ for a given $p$ as a random variable:
$$
Y(p) = \frac{1}{N} \sum_{i=1}^N F(p-X_{(i)}) \;,
$$
for $a \leq p \leq b,$
where $a= \widetilde{\mu}-C\widetilde{\sigma}\sqrt{\log(1/\eps)}$
and $b=\widetilde{\mu}+C\widetilde{\sigma}\sqrt{\log(1/\eps)}.$
Note that the expectation of $Y(p)$ is
$$ \frac{1}{N} \sum_{i=1}^N  \E_X[F(p-X)]= \E_X[F(p-X)]=Y'(p).$$
Therefore, we have that $\E[|Y(p)-Y'(p)|] =O(\sqrt{\var(Y(p))}).$
We bound the variance as follows:
\begin{align*}
& \var[Y(p)] =  \var[F(p-X)]/N =  \E\left[(F(p-X) - \sum_{q=a}^b F(p-q) X(q))^2\right]/N \\
= & \sum_{r=a}^b (X(r)/N) \cdot \left(F^2(p-r)   \vphantom{\sum_{q=a}^b} \right.\\
+ & \left. \sum_{q=a}^b \Big(  F^2(p-q) X(q)^2 - 2 F(p-r)F(p-q) X(q) + 2\littlesum_{q' \neq q} F(p-q)F(p-q')X(q)X(q') \Big) \right) \\
= & 1/N \cdot \sum_{a}^b F^2(p-q) (X(q)-X(q)^2) \leq   (1/N) \cdot \sum_{q=a}^b F^2(p-q) X(q) \;.
\end{align*}
 We have that
\begin{align*}
\sum_{p\in[\mu+t\widetilde{\sigma},\mu+(t+1)\widetilde{\sigma}]} & \E[|Y(p)-Y'(p)|]\\
 & = O(1/\sqrt{N}) \sum_p \sqrt{\sum_q F^2(p-q) X(q)}\\
& = O(\sqrt{\widetilde{\sigma}/N}) \sqrt{\sum_r F^2(r) \sum_{q\in [\mu+t\widetilde{\sigma}-r,\mu+(t+1)\widetilde{\sigma}-r]} X(q)} \tag*{(by Cauchy-Schwarz)}\\
& = O(\sqrt{\widetilde{\sigma}/N}) \sqrt{\sum_r F^2(r) \exp(-\Omega(|t\widetilde{\sigma}-r|/\widetilde{\sigma})^2))} \tag*{(by Lemma \ref{concentrationLemma})}\\
& = O(\sqrt{\widetilde{\sigma}/N}) \sqrt{\sum_r S(r)^2 \exp(-\Omega(((t\widetilde{\sigma}-r)/\widetilde{\sigma})^2 + (r/\widetilde{\sigma})^2))}\\
& = O(\sqrt{\widetilde{\sigma}/N})  \sqrt{\sum_r S(r)^2 \exp(-\Omega(t^2))}\\
& = O(\sqrt{\widetilde{\sigma}/N}) \exp(-\Omega(t^2))\sqrt{\sum_r S(r)^2 }\\
& = O(\sqrt{\widetilde{\sigma}/N}) \exp(-\Omega(t^2))\sqrt{\int_{\xi=0}^{1} I(\xi)^2 } \tag*{(by Plancherel's Theorem)}\\
& = O(\sqrt{\widetilde{\sigma}/N}) \exp(-\Omega(t^2)) \sqrt{k\widetilde{\sigma}^{-1}\sqrt{\log(1/\eps)}} \\
& = O(k^{1/2}\log^{1/4}(1/\eps)/\sqrt{N})\exp(-\Omega(t^2)).
\end{align*}
Summing the above over $t$ gives that
$$
\E[\dtv(Y,Y')] = O(k^{1/2}\log^{1/4}(1/\eps)/\sqrt{N}) = O(\eps) \;,
$$
for $N=k\sqrt{\log(1/\eps)}/\eps^2.$
This completes the proof.
\end{proof}

\newcommand{\pr}{\text{Pr}}

\section{Cover Size Upper Bound and Efficient Construction} \label{sec:cover-ub}

We start by establishing an upper bound on the cover size and then proceed
to describe our efficient algorithm for the construction of a proper cover with near--minimum size.
To prove the desired upper bound on the size of the cover, we proceed as follows:
We start (Section~\ref{ssec:reduce}) by reducing the cover size problem
to the case that the order $n$ of the $k$-SIIRV is at most $\poly(k/\eps)$.
In the second and main step (Section~\ref{ssec:sparse-upper}), we prove the desired upper bound for  the polynomially sparse case.
Our efficient algorithm for the cover construction (Section~\ref{sec:constr-cover})
is based on dynamic programming and  follows a similar case analysis.

\subsection{Reduction to Sparse Case} \label{ssec:reduce}
Our starting point is the following theorem:
\begin{theorem} \label{thm:reg}[\cite{DDOST13focs}, Theorem I.2]
Let $\p \in \mathcal{S}_{n, k}$ be a $k$-SIIRV of order $n$. Then, for any $\eps > 0$,
$\p$ is either
\vspace{-0.2cm}
\begin{enumerate}
\item a distribution with  variance at most $\poly(k/\eps)$; or
\vspace{-0.2cm}
\item $\eps$-close to a distribution $\p'$ such that for a random variable $X \sim \p'$, we have $X=c Z+ Y$ for some $1\leq c \leq k-1$, where $Y$,
$Z$ are independent random variables such that:
(i) $Y$ is distributed as a $c$-IRV, and (ii) $Z$ is a discretized normal random variable with parameters
${\frac {\mu} c}, {\frac {\sigma^2} {c^2}}$ where $\mu = \E[X]$ and $\sigma^2 = \Var[X]$.
\end{enumerate}
\end{theorem}

The above theorem allows us to reduce the problem of constructing an $O(\eps)$-cover
for $\mathcal{S}_{n, k}$ to the problem of constructing an $\eps$-cover for $\mathcal{S}_{n', k}$,
where $n' = \poly(k/\eps)$. Indeed, given an arbitrary $k$-SIIRV $\p \in \mathcal{S}_{n, k}$ we proceed as follows:
If $\p$ belongs to Case 1 of the above theorem, then we show (Lemma~\ref{lem:simple}) that there exists
a translation of a $k$-SIIRV with $n' = \poly(k/\eps)$ variables that is $\eps$-close to $\p$.
We show in the following subsection (Proposition~\ref{prop:sparse-upper}) that $\mathcal{S}_{n', k}$
admits an $\eps$-cover of size $(1/\epsilon)^{O(k\log(1/\epsilon))}$.
Since there are $O(kn)$ possible translations,
this gives a $2\eps$-cover of size $n(1/\epsilon)^{O(k\log(1/\epsilon))}$ for $k$-SIIRVs in Case 1.


Moreover, it is not difficult to show that there exists an $\eps$-cover for distributions in Case 2 with at most $n \cdot (k/\eps)^{O(k)}$ points.
In particular, we claim that for distributions in sub-case 2(i) there exists an $\eps$-cover of size  $(1/\eps)^{O(k)}$,
and for distributions in sub-case 2(ii) there exists an $\eps$-cover of size $O(n).$ Assuming these claims, the sub-additivity of total variation
distance (Proposition~\ref{prop:dtv-subadditive}) implies that distributions in Case 2 have a $2\eps$-cover of size
$n \cdot (1/\eps)^{O(k)}$ as desired.

Note that the random variable $Y$ in Case 2(i) is distributed as a $k$-IRV, i.e., it has support $k$. It is well-known and easy to show that the set of all distributions
over a domain of size $k$ has an $\eps$-cover of size $(1/\eps)^{O(k)}.$ It remains to show that we can $\eps$-cover the set of discretized
normal distributions of Case2(ii) with $O(nk/\eps)$ points. To do this, we exploit the fact that the variance of such distributions is large.
Let $\sigma_{\min} = {\Omega (k^9/\eps^3)}$ be the minimum variance of a $k$-SIIRV $X$ in Case 2. Note that the discrete Gaussian
in Case 2 has a variance of $\Var[X]/c^2$. Hence, we want to $\eps$-cover the set of discrete Gaussians with standard deviation $\sigma$ in the interval
$[\sigma_{\min} , \sigma_{\max}]$, where $\sigma_{\max} = O(\sqrt{n} k)$, and mean value $\mu$ in the interval $[0, n(k-1)]$.
Consider the following discretization of the space $(\sigma^2, \mu)$: We first define a geometric grid on $\sigma^2$ with ratio $(1+\eps)$, i.e.,
$\sigma_i^2 = \sigma^2_{\min} (1+\eps)^i$, where where  $0 \le i \le i_{\max}$ and $i_{\max} = O((1/\eps) \cdot \log(n)).$
For every fixed $i$, we define an additive grid on the means, so that  $|\mu_{j+1} - \mu_j| \le \eps \cdot \sigma_i$.
A combination of Propositions~\ref{prop:dpdtv} and~\ref{prop:normal-dist}  implies that this grid defines an $\eps$-cover.
Note that the total size of the described grid on $(\sigma^2, \mu)$ is
$$\sum_{i=0}^{i_{\max}} \frac{n(k-1)}{ \eps \cdot \sigma_i} =  \sum_{i=0}^{i_{\max}} \frac{n(k-1)}{ \eps \cdot \sigma_{\min} (1+\eps)^{i/2}} = O(n),$$
where the last inequality follows from the lower bound on $\sigma_{\min}$ and the elementary inequality $\sum_i  (1+\eps)^{-i/2} = O(1/\eps).$

The following lemma completes our reduction to the $n = \poly(k/\eps)$ case:
\begin{lemma} \label{lem:simple}
Let $\p\in {\cal S}_{n,k}$ be a $k$-SIIRV with $\Var_{X \sim \p}[X]=V$. For any $0< \delta < 1/4$,
there exists $\q\in {\cal S}_{n,k}$ with $\dtv(\p,\q) = O(\delta V)$
such  that all but $O(k+V/\delta)$ of the $k$-IRV's defining $\q$ are constant.
\end{lemma}
The proof of Lemma~\ref{lem:simple} is deferred to Appendix~\ref{app:simple}.
Note that an application of the lemma for $\delta = \eps/V$ completes the proof.

\subsection{Cover Upper Bound for Sparse Support} \label{ssec:sparse-upper}

In this subsection we prove the desired upper bound on the cover size for the sparse case:
\begin{proposition}\label{prop:sparse-upper}
Fix arbitrary constants $c,C>0$. Consider $n,k,\epsilon$ satisfying $\eps \leq k^{-c}$ and
$n\leq (k/\epsilon)^C$. Then there exists an $\epsilon$-cover of ${\cal S}_{n,k}$ under $\dtv$ of size
$(1/\eps)^{O_{c,C}(k\log(1/\eps))}$.
\end{proposition}

Our proof proceeds by analyzing the Fourier transform of the probability density functions of $k$-SIIRVs.
We will need the following definitions.

\medskip

\noindent {\bf Basic Definitions.}
For $\xi \in \R$,  recall that we use the notation $e(\xi) \eqdef \exp(-2\pi i \xi).$
For a probability distribution $\p$ over $\Z$, its Fourier Transform is the function $\widehat{\p}: [0, 1) \to \mathbb{C}$ defined
by $\widehat{\p}(\xi) = \E_{y \sim \p} [\exp(-2\pi i y \xi)] =  \E_{y \sim \p} [e(y \xi)].$ Note that Parseval's identity states that for two pdf's $\p$ and $\q$
we have  $\|\p - \q \|_2 = \| \widehat{\p} - \widehat{\q}  \|_2.$ In our context, $\p$ and $\q$ are going to be supported on a discrete set $A$,
in which case we have  $\|\p - \q \|_2 = \left(\sum_{a \in A} (\p(a)-\q(a))^2\right)^{1/2}$. On the other hand, $\widehat{\p}$ and $\widehat{\q}$
are Lebesgue measurable and we have $ \| \widehat{\p} - \widehat{\q}  \|_2 = \left(\int_0^1|\widehat{\p}(\xi)-\widehat{\q}(\xi)|^2 d\xi \right)^{1/2}.$

An equivalent way to view the Fourier transform is as a function defined on the unit circle in the complex plane.
For our purposes, we will need to analyze the corresponding polynomial defined over the entire complex plane.
Namely, we will consider the probability generating function $\widetilde{\p}:  \mathbb{C} \to \mathbb{C}$ of $\p$ defined as
$\widetilde{\p}(z) = \E_{y \sim \p} [z^y].$ Note that when $|z|=1$, this function agrees with the Fourier transform,
i.e.,  $\widehat{\p}(\xi) = \widetilde{\p}(e(\xi))$.



At a high-level, our proof is conceptually simple: For a $k$-SIIRV $\p$,
we would like to show that  the logarithm of its Fourier transform
$\log \widehat{\p}(\xi)$ is determined up to an additive $\eps$ by its degree $O(\log(1/\eps))$ Taylor
polynomial. Assuming this holds, it is relatively straightforward to prove the desired upper bound on the cover size.
Unfortunately, such a statement cannot be true in general for the following reason:
the function $\widetilde{\p}(z)$ may have roots near (or on) the unit circle, in which case
the logarithm of the Fourier transform is either very big or infinite at certain points.
Intuitively, we would like to show that the magnitude of $\widetilde{\p}(z)$
close to a root is small. Unfortunately, this is not necessarily true.

We circumvent this problem as follows: We partition the unit circle into $O(k)$ arcs
each of length $O(1/k)$. We perform a case analysis based on the number of roots that are close
to an arc.  If there are at least $\Omega(\log(1/\eps))$ roots of $\widetilde{\p}(z)$ close to a particular arc,
then we show (Lemma~\ref{lem:close-roots-ksiirvs}(i)) that the magnitude of $\widetilde{\p}(z)$ within the arc is going to be negligibly small.
Otherwise, we consider the polynomial $q(z)$ obtained by  $\widetilde{\p}(z)$
after dividing by the corresponding roots, and show that $\log q(z)$ is determined up to an additive $\eps$
by its degree $O(\log(1/\eps))$ Taylor polynomial within the arc (see Lemma~\ref{lem:approx}).
Using the aforementioned structural understanding, to prove the cover upper bound,
we define a ``succinct'' description of the Fourier Transform
based on the logarithm of $q(z)$ and appropriate discretization of $O(\log(1/\eps))$ nearby roots.

Note that we take advantage of the fact that our distributions are supported over a domain of size $\ell = \poly(k/\eps),$
in order to relate their total variation distance
to the $L_{\infty}$ distance between their Fourier transforms. In particular, we have the following simple fact:
\begin{fact} \label{fact:trivial}
For any pair of pdfs $\p, \q$ over $[\ell]$, we have
$\|\p-\q\|_1 \le \sqrt{\ell+1}  \|\widehat{\p}-\widehat{\q}\|_{\infty}.$
\end{fact}
Indeed,  note that
$\|\p-\q\|_1 \le \sqrt{\ell+1} \|\p-\q\|_2  = \sqrt{\ell+1}  \|\widehat{\p}-\widehat{\q}\|_2 \le \sqrt{\ell+1}  \|\widehat{\p}-\widehat{\q}\|_{\infty},$
where the equality is Parseval's identity.

For the rest of this section we fix an arbitrary $\p \in{\cal S}_{n,k}$ and analyze the polynomial $\widetilde{\p}(x)$. 
We start with the following important lemma whose proof is deferred to Appendix~\ref{ap:close-roots}:

\begin{lemma} \label{lem:close-roots-ksiirvs}
Fix $x \in \mathbb{C}$ with $|x|=1$.
Suppose that $\rho_1,\ldots,\rho_m$ are roots of $\widetilde{\p}(x)$ (listed with appropriate multiplicity)
which have $|\rho_i-x| \leq \frac{1}{2k}$. Then, we have the following:
\begin{itemize}
\item[(i)] $|\widetilde{\p}(x)| \le 2^{-m} \;.$
\item[(ii)] For the polynomial $q(x) = \widetilde{\p}(x) / \prod_{i=1}^m (x - \rho_i)$,  we have that $|q(x)| \leq k^m$.
\end{itemize}
\end{lemma}

Our main lemma for this section shows that we can $\eps$-approximate the Taylor series of $q(x)$ by only considering the first $O(\log(1/\eps))$ terms:

\begin{lemma}\label{lem:approx}
Fix $w \in \mathbb{C}$ with $|w|=1$.
Suppose that $\rho_1,\ldots,\rho_m$ are all the roots of $\widetilde{\p}(x)$ (listed with appropriate multiplicity)
which have $|\rho_i-w| \leq \frac{1}{3k}$.  Let $q(x) = \frac{\widetilde{\p}(x)}{\prod_{i=1}^m (x-\rho_i)}$ and
let the Taylor series of $\ln(q(x))$ at $w$ be $\ln q(x) = \sum_{j=0}^\infty c_j (x-w)^j \; .$
Then, we have that $|c_j|\leq nk(3k)^j$, for all $j \ge 1$, and the real part of $c_0$ is at most $m \ln k$.

Fix $0< \eps \leq 1/(12mk)$ and an integer $\ell$ satisfying $\ell \geq \log(9nk)$. 
For $\rho'_j$ with $|\rho'_j-\rho_j| \leq \eps$ for $j \in \{1, \ldots , m\}$,
and $c'_j$ with $|c'_j-c_j| \leq \eps$ for $j \in \{1,\ldots, \ell \}$  we have:
For all $x \in \mathbb{C}$ with $|x|=1$ and $|x-w| \leq \frac{1}{6k}$
\begin{equation}\label{eqn:error}
\left|\widetilde{\p}(x)-\Big( \littleprod_{j=1}^m(x-\rho'_j) \Big) \exp\Big(\littlesum_{j=0}^\ell c'_j (x-w)^j \Big) \right| \leq O\left(\eps m k + nk2^{-\ell} \right).
\end{equation}
\end{lemma}

\begin{proof}
We start by noting that, by the triangle inequality,  Lemma~\ref{lem:close-roots-ksiirvs} applies to all points
$x \in \mathbb{C}$ with $|x|=1$ and $|x-w| \leq \frac{1}{6k}$.
Observe that $c_0=\ln [q(w)]$ and by Lemma~\ref{lem:close-roots-ksiirvs}(ii)
$|q(w)| \leq k^m$. This gives the claim on the real part of $c_0$.

Note that $\ln(q(x))$ can be expressed as a sum of the form
$$
\ln(q(x)) = c_0 + \sum_{h=1}^{R} \ln(1 - (x-w)/(r_h-w)) \;,
$$
where $c_0 =  \ln [q(w)]$, $r_j$ are the roots of $q(x)$, and $R \leq n(k-1)$ is the degree of $q(x)$.
By the definition of $q$, it follows that $|r_h-w| > \frac{1}{3k}$ for all $1 \leq h \leq R$.

Inserting the standard Taylor series $\ln(1+y)=\sum_{j=0}^\infty \frac{y^j}{j}$ gives
$$
\ln(q(x)) = c_0 + \sum_{h=1}^{R} \sum_{j=0}^\infty \frac{(-1)^j(x-w)^j}{j \cdot (r_h-w)^j} .
$$
Considering the $(x-w)^j$ term above gives $c_j = \frac{(-1)^j}{j}\sum_{j=1}^R (r_j-w)^{-j}$.
Therefore,
$$|c_j| \leq R(3k)^j  \leq nk(3k)^j \;.$$
This gives the desired bound on $|c_j|$, $j \ge 1$.

We now proceed to prove (\ref{eqn:error}).
We start by considering the difference
$$
\sum_{j=0}^\ell c'_j (x-w)^j - \ln(q(x)) \;,
$$
for $x$ in the appropriate range.
Since $|x-w| \leq \frac{1}{6k} \leq 1/2$ and $|c'_j-c_j| \leq \eps$,
we have
$$\left|\sum_{j=0}^\ell c'_j (x-w)^j - \sum_{j=0}^\ell c_j (x-w)^j \right| \leq \eps \cdot \sum_{j=0}^{\ell} 2^{-j}  \leq 2 \eps \;.$$
So, we need to consider the error introduced by truncating the Taylor series after the first $\ell$ terms.
We have
\begin{eqnarray*}
\left| \sum_{j=0}^\ell c_j (x-w)^j - \ln(q(x)) \right| & = & \left| \sum_{j>\ell} c_j(x-w)^j \right| \\
& \leq & \sum_{j>\ell} nk(3k)^j(6k)^{-j} \\
& = & nk 2^{-\ell}
\end{eqnarray*}
Therefore, by the triangle inequality,
$$
\left|\sum_{j=0}^\ell c'_j (x-w)^j - \log(q(x))\right| \leq  2\epsilon + nk2^{-\ell} .
$$
Thus, the multiplicative error in this approximation, i.e.,
$$\frac{1}{q(x)} \exp \left(\sum_{j=0}^\ell c'_j (x-w)^j \right) = \frac{1}{\widetilde{\p}(x)} \left( \prod_{j=1}^m(x-\rho_j)\right) \exp \left(\sum_{j=0}^\ell c'_j (x-w)^j \right)$$
 is {$\exp(E),$ where $|E| \leq 2 \eps + nk2^{-\ell}.$}
Since $|\widetilde{\p}(x)| \leq 1$ and by our assumptions on $\ell$, $2 \eps + nk2^{-\ell} \leq 1$, we have that
$$
\left|\widetilde{\p}(x)- \left( \prod_{j=1}^m(x-\rho_j) \right) \exp\left(\sum_{j=0}^\ell c'_j (x-w)^j \right) \right| \leq e\cdot ( 2 \eps + nk2^{-\ell}).
$$
We next replace each $\rho_j$ by the corresponding $\rho'_j$ one at a time.
By a simple induction, we will show that for all $1 \le h \le m$
\begin{equation} \label{eq:root-induct}
\left|\widetilde{\p}(x)- \left( \prod_{j=1}^h(x-\rho'_j) \right) \left( \prod_{j=h+1}^m(x-\rho_j) \right)  \exp\left(\sum_{j=0}^\ell c'_j (x-w)^j \right) \right| \leq e\cdot ( 2 \eps + nk2^{-\ell}) + 4hk\eps.
\end{equation}
We have just shown this for $h=0$.
So, we assume (\ref{eq:root-induct}) for $0 \leq h \leq m-1$ and seek to prove it for $h+1$.
For simplicity, we rewrite (\ref{eq:root-induct}) as
$$\left|\widetilde{\p}(x)- (x-\rho_h)f_h(x) \right| \leq e\cdot ( 2 \eps + nk2^{-\ell}) + 4hk\eps \;,$$
where $f_h(x)= \left( \prod_{j=1}^{h-1}(x-\rho'_j) \right) \left( \prod_{j=h+1}^m(x-\rho_j) \right)  \exp\left(\sum_{j=0}^\ell c'_j (x-w)^j \right)$.

Note that the RHS of (\ref{eq:root-induct}) satisfies
$$e\cdot ( 2 \eps + nk2^{-\ell}) + 4hk\eps \leq e\cdot ( 2 \eps + nk2^{-\ell}) + 4mk\eps \leq 1\;,$$
by our assumptions on $\eps$ and $\ell$.
Since $|\widetilde{\p}(x)| \leq 2^{-m} \leq 1$, we have $|(x-\rho_h)f_h(x)| \leq 2$ or $|f_h(x)| \leq \frac{2}{|x-\rho_{h}|} \leq 4k$.
Now if we replace $(x-\rho_h)f_h(x)$ with $(x-\rho_h')f_h(x)$, we introduce an error of
$|(x-\rho_h)f_h(x) - (x-\rho_h')f_h(x)| = |\rho'_h-\rho_h||f_h(x)| \leq \eps \cdot 4k$.
Hence,
$$\left|\widetilde{\p}(x)- (x-\rho'_h)f_h(x) \right| \leq e\cdot ( \ell \epsilon + nk2^{-\ell}) + 4(h+1)k\eps$$
But this is just (\ref{eq:root-induct}) for $h+1$, completing the induction.

Taking $h=m$ in (\ref{eq:root-induct}) gives:
$$
\left|\widetilde{\p}(x)- \left( \prod_{j=1}^m(x-\rho'_j) \right) \exp\left(\sum_{j=0}^\ell c'_j (x-w)^j \right) \right| \leq e\cdot ( 2 \eps + nk2^{-\ell}) +4mk\eps
$$
as required.
\end{proof}

We are now prepared to prove Proposition~\ref{prop:sparse-upper}.

\begin{proof}[Proof of Proposition~\ref{prop:sparse-upper}.]
By replacing $\epsilon$ by a power of itself, we may assume that $\epsilon \leq k^{-1}$ and that $n \leq \epsilon^{-1}$. We may additionally assume that $\epsilon$ is sufficiently small.

It suffices to find a subset $T$ of ${\cal S}_{n,k}$ of appropriate size so that for any $\p\in{\cal S}_{n,k}$ there is some $\q\in T$ so that $|\widetilde{\p}(z)-\widetilde{\q}(z)|\leq \epsilon^2$ for all $|z|=1$, as Fact~\ref{fact:trivial} would then imply that $\dtv(\p,\q)\leq \eps$.

We begin by defining some parameters. Let $m$ be an integer larger than $3\log(1/\epsilon)$. Let $\ell$ be an integer larger than $\log(nk/\epsilon^3)$ and $\delta>0$ a real number smaller than $\epsilon^3/(mk+\ell)$. Additionally, we divide the unit circle of $\C$ into $O(k)$ arcs each of length at most $1/(3k)$.

To each $\p \in {\cal S}_{n,k}$ we associate the following data:
\begin{itemize}
\item For each arc in our partition with midpoint $w_I$, define $q(z)$ as in Lemma \ref{lem:approx}. Then we define $\p_I$ as follows:
\begin{itemize}
\item If $\widetilde{\p}(z)$ has at least $m$ roots within distance $1/(3k)$ of $w_I$ or if $|q(w_I)|<\epsilon^3\exp(-nk)$, we let $\p_I=\textbf{Small}$.
\item Otherwise, we let $\p_I$ consist of the following data:
\begin{itemize}
\item Roundings of the roots of $\widetilde{\p}(z)$ that are within $1/(3k)$ of $w_I$ to the nearest complex numbers whose real and imaginary parts are multiples of $\delta/2$.
\item Roundings of the first $\ell$ Taylor coefficients of $\log(q)$ about $w_I$ to the nearest complex numbers whose real and imaginary parts are multiples of $\delta/2$.
\end{itemize}
\end{itemize}
\end{itemize}

We then let $D(\p)$ be the sequence $\{\p_I\}_{I\textrm{ an arc in the partition}}$. For each value $V$ that can be obtained as $D(\p)$ for some $\p\in \mathcal{S}_{n,k}$, we pick one such $\p$ called $\q_V$. We define our cover $T$ to be the set of all such $\q_V$. In order to show that this is an appropriate cover, we need to show two claims:
\begin{enumerate}
\item The number of possible values of $D(\p)$ is at most $\left(1/\eps\right)^{O(k\log(1/\eps))}.$ This implies that $|T|$ is appropriately small.
\item If $\p,\q\in{\cal S}_{n,k}$ have $D(\p)=D(\q)$, then $\dtv(\p,\q)\leq \epsilon$. This will imply that $T$ is a cover, since given any $\p\in \mathcal{S}_{n,k}$, we may take $\q=\q_{D(\p)}\in T$.
\end{enumerate}

The first claim is relatively straightforward. For each of $O(k)$ arcs, $I$, we have that $\p_I$ is either $\textbf{Small}$ or a sequence of $O(\log(1/\epsilon))$ complex numbers, each of which can take only $\poly(1/\delta)$ many possible values. Thus, the number of possible values for $\p_I$ is at most $\delta^{-O(\log(1/\epsilon))}=(1/\epsilon)^{O(\log(1/\epsilon))}$. The number of possible values for $D(\p)$ is at most this raised to the number of arcs, which is $(1/\epsilon)^{O(k\log(1/\epsilon))}.$

The second claim is slightly more involved. We note that it is sufficient to show that if $D(\p)=D(\q)$, then $|\tilde{\p}(z)-\tilde{\q}(z)|\leq \epsilon^2$ for all unit norm $z$. In particular, we show the stronger claim that for any of our arcs $I$ if $\p_I=\q_I$, then $|\widetilde{\p}(z)-\widetilde{\q}(z)|=  O(\epsilon^3)$ for all $z\in I$.

If $\p_I=\q_I=\textbf{Small}$, we claim that $|\tilde{\p}(z)|,|\tilde{\q}(z)|=O(\epsilon^3)$ for all $z\in I$. It suffices to show this merely for $\p$. On the one hand, if $\widetilde{\p}(z)$ has more than $m$ roots near $w_I$, this follows from the first part of Lemma \ref{lem:close-roots-ksiirvs}. On the other hand, if $|q(w_I)|\leq \epsilon^3\exp(-nk)$, then for any other $z\in I$ we have that
$$
q(z) = q(w_I) \exp\left(\sum_{i=1}^\infty c_i (z-w_I)^i \right),
$$
where by Lemma \ref{lem:approx}, $|c_i|\leq nk(3k)^i.$ Therefore, for $z\in I$, since $|z-w_I|\leq 1/(6k)$, we have by Lemma \ref{lem:close-roots-ksiirvs} that
$$
|\widetilde{\p}(z)|\leq |q(z)| \leq |q(w_I)| \exp(nk) \leq \epsilon^3.
$$

If $\p_I=\q_I\neq\textbf{Small}$, we note by Lemma \ref{lem:approx} that for $z\in I$ that both of $\widetilde{\p}(z)$ and $\widetilde{\q}(z)$ are within $O(mk\delta+\ell\delta+nk2^{-\ell})=O(\epsilon^3)$ of $\prod_{j=1}^{M}(z-\rho'_j)\exp\left(\sum_{j=0}^\ell c'_j(z-w_I)^j \right)$, where the $\rho'_j$ are the roundings of nearby roots and $c'_j$ the roundings of the Taylor coefficients given by the data $p_I=q_I$. Thus, again in this case, $|\widetilde{\p}(z)-\widetilde{\q}(z)|\leq O(\epsilon^3)$ for all $z\in I.$

This completes the proof of Proposition~\ref{prop:sparse-upper}.

\end{proof}

\subsection{Efficient Cover Construction} \label{sec:constr-cover}
In this section, we give an algorithm to construct a near-minimum size cover in output polynomial time:
\begin{theorem}\label{coverConstructionThm}
Let $n,k$ be positive integers and $\epsilon>0$. There exists an algorithm
that runs in time $n\left( k/\eps \right)^{O(k\log(1/\epsilon))}$ and returns a proper $\eps$-cover for $\mathcal{S}_{n,k}$, i.e.,
a cover consisting of $n\left( k/\eps \right)^{O(k\log(1/\eps))}$ $k$-SIIRVs each given as an explicit sum of $k$-IRVs.
\end{theorem}

Our algorithm builds on the existential upper bound established in the previous subsections.
We first construct an $\eps$-cover for $k$-SIIRVs in Case 2 of Theorem \ref{thm:reg}, i.e., $k$-SIIRVs whose variance is more than a sufficiently large polynomial in $k/\eps$.
By Theorem \ref{thm:reg} each such $k$-SIIRV is $\eps$-close to a random variable of the form $cZ+Y$, where $1 \le c \le k-1$ is an integer, $Z$ is a discrete Gaussian
and $Y$ is a $c$-IRV. In Section~\ref{ssec:reduce} we exploited this structural fact to
construct a non-proper cover for $k$-SIIRVs in this case. We remark that this non-proper cover may
contain ``spurious'' points, i.e., points not close to a large variance $k$-SIIRV.
Efficiently constructing a proper cover without spurious points for the high variance case requires careful arguments and is deferred to Appendix~\ref{ap:cover-high-var}.

We now focus our attention to Case 1.
By Lemma \ref{lem:simple}, we have that all such $k$-SIIRVs can be approximated by a constant
plus a sum of $\poly(k/\epsilon)$ $k$-IRVs. Since there are only $nk$ possibilities for this constant,
and all such possibilities are easily obtainable, it suffices to find an explicit $\epsilon$-cover for $\mathcal{S}_{n,k}$ when $n=\poly(k/\epsilon)$.

A simple but useful observation is that we can round each coordinate probability for each of our $k$-IRVs
to a multiple of $\epsilon/(nk)$ and introduce an error of $O(\epsilon)$ in total variation distance.
Therefore, it suffices to find a cover of $\mathcal{S}'_{n,k}$, a sum of $n=\poly(k/\epsilon)$ independent $k$-IRVs,
where each of their coordinate probabilities is a multiple of $\frac{1}{N}$ for some integer $N=\poly(k/\epsilon)$.
We will henceforth call such a $k$-IRV {\em $N$-discrete $k$-IRV}.

Our main workhorse here will once again be Lemma \ref{lem:approx}.
The cover we construct will be much the same as in Proposition~\ref{prop:sparse-upper},
but we will now explicitly produce SIIRVs that obtain every possible value of $D$.
Fortunately, the Taylor series of the log of the Fourier transform
is additive in the composite $k$-IRVs, and so there exists an appropriate dynamic program to solve this problem.

Let $\delta>0$ be given by a sufficiently small polynomial in $\epsilon/k$, and let $m$ be an integer at least a sufficiently large multiple of $\log(1/\epsilon)$. We divide the unit circle into arcs $I$ with midpoints $w_I$ as described in the proof of Proposition \ref{prop:sparse-upper}.
For any $N$-discrete $k$-IRV, $\p$, we associate the following data. For each interval $I$, let $\rho_{1,I},\ldots,\rho_{r_I,I}$ be the roots of $\widetilde{\p}$ that are within
distance $1/(3k)$ of $w_I$, and let $q(z)=\frac{\widetilde{\p}(z)}{\prod (z-\rho_{i,I})}$. For $1 \leq j \leq r_I$, let $\rho'_{j,I}$ be a rounding of $\rho_{j,I}$ with $\rho'_j,I=(a+bi)\delta$ for some $a,b \in \Z$ and $|\rho'_{j,I}-\rho_{j,I}| \leq \delta$.
For $1 \leq j \leq m$, let $c'_{j,I}$ be a rounding of $c_{j,I}$ with $c'_{j,I}=(a+bi)\delta$ for some $a,b \in \Z$ and $|c'_{i,I}-c_{i,I}| \leq \delta$, where the $c_{k,I}$ are the coefficients of first $m+1$ terms of the Taylor series $\ln q(z) = \sum_{j=0}^\infty c_j (z-w_I)^j$.
Let $\p_I$ be the data consisting of the list $(\rho'_{1,I},\ldots,\rho'_{r_I,I})$ and the vector $(c'_{0,I},c'_{1,I},\ldots,c'_{m,I})$. We let $D(\p)$ be the sequence of $\p_I$ over all intervals $I$.

Given a sequence $\p_1,\p_2,\ldots,\p_h$ of $k$-IRVs, we let $D(\p_1,\ldots,\p_k)$ be given by the following data for each $I$:
\begin{itemize}
\item The first $m$ elements of the concatenation of the lists of approximate roots of $\prod_{i=1}^h \widetilde{\p_i}(z)$ near $w_I$.
\item The list of elements $\sum_{i=1}^h c'_{j,I}(\p_i)$ for $0\leq j \leq m$, with the exception that the $j=0$ term is replaced by $-\infty$ if for any $h'<h$ we have that
the real part of  $\sum_{i=1}^{h'} c'_{0,I}(\p_i)$ is less than $-nk-m {- m \ln k}$.
\end{itemize}
Our algorithm will follow from three important claims:
\begin{claim} \label{claim:algo}
We have the following:
\begin{enumerate}
\item[(i)] $D(\p_1,\ldots,\p_h)$ can be computed in $\poly(k/\epsilon)$ time from $D(\p_1,\ldots,\p_{h-1})$ and $D(\p_h)$.
\item[(ii)] There are only $\left( k/\eps\right)^{O(k\log(1/\eps))}$ possible values for $D(\p_1,\ldots,\p_h)$ for any $h\leq n$.
\item[(iii)] If $D(\p_1,\ldots,\p_n)=D(\q_1,\ldots,\q_n)$ and $\p$,$\q$ are the distributions of $\sum_{i=1}^n X_i$
and $\sum_{i=1}^n Y_i$ for $X_i \sim \p_i$ and $Y_i \sim \q_i$ then $\dtv(\p,\q) \le \eps$.
\end{enumerate}
\end{claim}
\begin{proof}
The first statement follows from the fact that the lists of roots in $D(\p_1,\ldots,\p_h)$ are
obtained by concatenating those in $D(\p_1,\ldots,\p_{h-1})$ with those in $D(\p_h)$, and truncating if necessary.
And moreover that $\sum_{i=1}^h c'_{j,I}(\p_i)$ is obtained by adding $c'_{j,I}(\p_h)$ to $\sum_{i=1}^{h-1} c'_{j,I}(\p_i)$ (with the term remaining $-\infty$ if it was in $D(\p_1,\ldots,\p_{h-1})$).

For the second statement note that for each of the $O(I)$ intervals,
we store $O(\log(1/\epsilon))$ complex numbers whose real and imaginary parts are each multiples of $\delta$.
As each of these numbers (with the exception of a $-\infty$ term) have size at most $\poly(k/\epsilon)$
and $\delta=\poly(\eps/k)$, there are only $\poly(k/\epsilon)^{O(k\log(1/\epsilon))}$ many possible values for $D(\p_1,\ldots,\p_h)$.

The third statement is true for essentially the same reasons as in the proof of Proposition \ref{prop:sparse-upper}.
Once again, we simply need to show that for each interval $I$ it holds $|\widetilde{\p}(z) - \widetilde{\q(}z)|\leq (\epsilon/k)^c$ for all $z\in I$ and $c$ a sufficiently large constant. Note that the listed roots are simply $\delta$-approximations of the (first $m$) roots of $\widetilde{\p}$ and $\widetilde{q}$ within distance $1/(3k)$ of $w_I$, and the $\sum_{i=1}^{n} c'_{j,I}(\p_i)$ are within distance $n\delta$ of the coefficients of the Taylor expansion of the logarithm of $q(z)$ about $w_I$. If we have $m$ nearby roots, both $\widetilde{\p}$ and $\widetilde{\q}$ are small for all $z$ in this range. Otherwise, unless there is a $-\infty$ in $D(\p)=D(\q)$, they are close by Lemma \ref{lem:approx}. If we do have a $-\infty$ then
$$\Re\left( \sum_{i=1}^{{h'}} c'_{0,I}(\p_i)\right)<-nk-m {-m \ln k}$$
for {some $h' \leq h$}. Since  the later $c_{0,I}(\p_i)$ and $c_{0,I}(\q_i)$ {have $\Re c_{0,I}(\p_i) \leq m_i \ln k$ and $\Re c_{0,I}(\p_i) \leq m_i \ln k$ by Lemma \ref{lem:approx}}, this means that $|q(w_I)|<e^{-m}e^{-nk}$, and as in Proposition \ref{prop:sparse-upper}, this implies that both $\widetilde{\p}$ and $\widetilde{\q}$ are sufficiently small.
\end{proof}

We can now present the algorithm for producing our cover. The basic idea is to use a dynamic program to come up with one representative collection of $\p_1,\ldots,\p_h$ to obtain each achievable value of $D$. The algorithm is as follows:

\medskip

\fbox{\parbox{6in}{
{\bf Algorithm} {\tt Cover-SIIRV}\\
Input: $k,\epsilon>0$ and $n,N=\poly(k/\epsilon)$.\\

\vspace{-0.5cm}

\begin{enumerate}

\item Define $\delta$ and $m$ as above.

\item Let $L_0 = \{(D(\emptyset),\emptyset)\}$.

\item For $h=1$ to $n$

\item Let $L_h$ be the set of terms of the form $(D(\p_1,\ldots,\p_h),(\p_1,\ldots,\p_h))$ where $(D(\p_1,\ldots,\p_{h-1}),(\p_1,\ldots,\p_{h-1}))\in L_{h-1}$ and $\p_h$ is an $N$-discrete $k$-IRV.

\item Use a hash table to remove from $L_h$ all but one term with each possible value of $D(\p_1,\ldots,\p_h)$

\item End for

\item Return the list of distributions $\sum_{i=1}^n X_i$ with $X_i \sim \p_i$ for each $(D(\p_1,\ldots,\p_{n}),(\p_1,\ldots,\p_{n}))\in L_{n}$.

\end{enumerate}
}}

\medskip

To prove that this produces a cover, we claim by induction on $h$ that $L_h$ contains an element that achieves each possible value of $D(\p_1,\ldots,\p_h)$. This is clearly true for $h=0$. Given that it holds for $h-1$, Claim~\ref{claim:algo}(i) implies that the non-deduped version of $L_h$ also satisfies this property, and deduping clearly does not destroy it. Therefore $L_n$ contains (exactly one) element for each possible value of $D(\p_1,\ldots,\p_n)$. Therefore, by Claims~~\ref{claim:algo}(ii) and (iii),
the algorithm will return a cover of the appropriate size. For the runtime, we note that the initial size of $L_h$ before deduping is the product of the size of $L_{h-1}$ and the number of $N$-discrete $k$-IRVs, which by Claim~\ref{claim:algo}(ii) is $\poly(k/\epsilon)^{k\log(1/\epsilon)}$. Each of these elements are generated in $\poly(k/\epsilon)$ time, and the deduping process takes only polynomial time per element. Therefore, the final runtime is $\poly(k/\epsilon)^{k\log(1/\epsilon)}$. This completes the proof of Theorem~\ref{coverConstructionThm}.

\section{Cover Size Lower Bound} \label{sec:lb}
In this section we prove our lower bound on the cover size of $k$-SIIRVs.
In Section~\ref{ssec:explicit-pbds} we show the desired lower bound for the case of $2$-SIIRVs.
In Section~\ref{ssec:ksiirv} we generalize this construction for general $k$-SIIRVs.

\subsection{Cover Size Lower Bound for $2$-SIIRVs} \label{ssec:explicit-pbds}

We start by providing an explicit lower bound on the cover size of $2$-SIIRVs.
In particular, we show the following:

\begin{theorem} \label{thm:explicit-cover} For all $0< \eps \leq e^{-42}$ and $n \in \Z$ such that $7 \leq n \leq \frac{1}{6} \ln(1/\eps)$,
there is an $\eps$-packing of $\mathcal{S}_{n , 2}$ under $\dtv$ with cardinality $(1/\eps)^{\Omega(n)}$.
\end{theorem}

We begin with the following useful lemma:
\begin{lemma}\label{distSepLem}
Let $\p$ and $\q$ be $2$-SIIRVs given by parameters $p_i$ and $q_i$ for $1\leq i \leq n$, for some $n\geq 7$.
Suppose that for all $i$, $1 \leq i \leq n$, it holds
$\left|p_i-i/(n+1) \right| \leq 1/4(n+1)$ and $\left|q_i- i/(n+1) \right| \leq 1/4(n+1).$
Then,
$$
\dtv(\p,\q) \geq \max_i |p_i-q_i| \cdot e^{-3n}.
$$
\end{lemma}
\begin{proof}
Let $\eps = |p_i-q_i|e^{-3n}$.
For a distribution $\p$ supported on $[n]$, define $r_{\p}(p)$ to be the polynomial
$$r_\p(p)=\E_{X \sim \p} \left[ (p-1)^{X} \cdot p^{n-X} \right]= \sum_{i=0}^n \p(i) (p-1)^{i} p^{n-i}.$$
For a PBD $\p \in \mathcal{S}_{n,2}$ and $X \sim \p$ with $X = \sum_{i=1}^n X_i$ for $X_i \sim \mathrm{Ber}(p_i)$, we have that
\begin{eqnarray*} r_\p(p) 	& = & \E \left[ (p-1)^{X} p^{n-X} \right] 
					 =  \E \left[ (p-1)^{\sum_{i=1}^n X_i} \cdot p^{\sum_{i=1}^n (1-X_i)} \right] \\			
					& = & \E \left[ {\prod_{i=1}^n (p-1)^{X_i} p^{1-X_i} } \right] 			
			         =  \prod_{i=1}^n \E \left[ (p-1)^{X_i} p^{1-X_i}  \right] \\			
					& = & \prod_{i=1}^n \left (p_i (p-1) + (1-p_i) p \right) 
					 =  \prod_{i=1}^n (p- p_i ) \; . 
\end{eqnarray*}
Hence, the roots of the polynomial $r_{\p}$ are exactly the parameters $p_i$ of the $2$-SIIRV $\p \in \mathcal{S}_{n,2}$.
We have the following simple claim:

\begin{claim} \label{claim:dist-vs-poly}
Let $\p, \q \in {\cal S}_{n, 2}$ such that $\dtv\left(\p, \q \right) < \eps$.
Then for any $p \in [0, 1]$, we have that $$| r_\p(p) - r_\q(p) | < 2\eps.$$
\end{claim}
\begin{proof}
We have the following sequence of (in)equalities:
\begin{eqnarray*}
| r_\p(p) - r_\q(p) |  &=& \left| \sum_{i=0}^n (\p(i) - \q(i)) (p-1)^{i} p^{n-i} \right|  
                                 \le  \sum_{i=0}^n \left| (\p(i) - \q(i)) \right| \cdot \left| (p-1)^{i} p^{n-i} \right| \\
                                 &\le & \sum_{i=0}^n \left| \p(i) - \q(i) \right| =2 \dtv\left(\p, \q \right) 
                                 < 2\eps \;,
\end{eqnarray*}
where the second line is the triangle inequality and the third line uses the fact that
$|(p-1)^{i} p^{n-i}| \le 1$ for all $i \in [n]$ and $p \in [0,1].$
\end{proof}
Hence, to prove the lemma, it suffices to show that for some $p\in[0,1]$ that
$$|r_{\p}(p) - r_{\q}(p) | \geq 2\eps.$$
In particular, we show this for $p=p_i$. Noting that $r_{\p}(p_i)=0$, it suffices
to show that $|r_{\q}(p_i)| \geq 2\eps$. We now proceed to prove this fact.
If $j \neq i$ we have that,
$$
|p_{i}-q_j| \geq \frac{|i-j|}{n+1}-\left|p_i-\frac{i}{n+1} \right|-\left|q_j-\frac{j}{n+1} \right| \geq \frac{1}{2(n+1)}.
$$
Therefore, we have that
\begin{eqnarray*}
\left| r_{\q}(p_i) \right| & =& \prod_{j=1}^n \left| p_{i}-q_j \right| \nonumber
\geq |p_i-q_i| \cdot  \prod_{j \neq i} \frac{|i-j|}{2(n+1)}. \nonumber \\
\end{eqnarray*}
We note that
\begin{eqnarray}
 \prod_{j \neq i} \frac{|i-j|}{(n+1)} & =&  (i-1)!(n-i)!
\geq  \frac{n!}{\binom{n-1}{i-1}}
\geq \frac{(n/e)^{n}}{2^{n-1}} \label{eqn:dyo},
\end{eqnarray}
where we use the elementary inequalities $n! \ge (n/e)^n$ and ${\binom{n-1}{i^{\ast}-1}} \le 2^{n-1}.$ Applying this to the above, we find that
\begin{eqnarray*}
\left| r_{\q}(p_i) \right| & =& \frac{|p_i-q_i|}{e \cdot (n+1)(4e)^n}
\geq \frac{2|p_i-q_i|}{e^{3n}}
\geq 2\eps.
\end{eqnarray*}
\end{proof}

\begin{proof}[Proof of Theorem~\ref{thm:explicit-cover}]

Given $\eps>0$ and $n \in \Z$ satisfying the condition of the theorem, we define an explicit $\eps$-packing for $\mathcal{S}_{n , 2}$
as follows: Let $s = \lfloor \eps^{-1/2} \rfloor$.
For a vector  $\ba = (a_1,\ldots, a_n) \in [s]^n$, let
$$p^{\ba}_i=\frac{i}{n+1}+ \frac {a_i \sqrt{\epsilon}}{4n},  \quad i \in \{1, \ldots, n\} \;,$$
be the parameters of a $2$-SIIRV $\p_{\ba} \in {\cal S}_{n, 2}$.
We claim that the  set of $2$-SIIRVs $\big\{ \p_{\ba} \big\}_{\ba \in [s]^n}$ satisfies the conditions of the theorem, i.e.,
for all $\ba,\bb \in  [s]^n$, $\ba \neq \bb$ implies $\dtv\left(\p_{\ba}, \p_{\bb} \right) \geq \eps$.

In particular, if $\ba\neq \bb$, then there must be some $i$ so that $a_i\neq b_i$. Then, by Lemma \ref{distSepLem}, we have that
$$
\dtv(\p_{\ba},\p_{\bb}) \geq |p^{\ba}_i - p^{\bb}_i|e^{-3n} \geq \frac{\sqrt{\eps}}{4n}e^{-3n} \geq \frac{\eps^{3/4}}{4n} \geq \eps.
$$

\end{proof}

As a simple corollary we obtain the desired lower bound:

\begin{corollary} \label{cor:simple} 
For all $0< \eps < 1$ and $n = \Omega(\log(1/\eps)),$
any $\eps$-cover of $\mathcal{S}_{n,2}$ 
under $\dtv$ must be of size $n \cdot (1/\eps)^{\Omega(\log 1/\eps)}.$
 \end{corollary}
\begin{proof}
We will assume without loss of generality that $\eps$ is smaller than an appropriately small
positive constant.
First note that if there exists a $3\eps$-packing for $\mathcal{S}_{n,2}$ of cardinality $M$, then any $\eps$-cover for $\mathcal{S}_{n,2}$
must be of cardinality at least $M$. Indeed, for every $\q_i$, $i=1, \ldots, M$, in the $3\eps$-packing,
consider the (non-empty) set $N_{\eps} (\q_i)$ of points $\p$ in the $\eps$-cover with $\dtv(\q_i ,\p) \leq \eps$.
If $\p \in N_{\eps} (\q_i)$ and $j \neq i$,  we have $\dtv (\p, \q_j) \geq \dtv (\q_j, \q_i) - \dtv( \q_i, \p) \geq 2\eps$.
That is, the sets $N_{\eps} (\q_i)$ are each non-empty and mutually disjoint, which implies that the size of any $\eps$-cover is at least $M$.

By Theorem  \ref{thm:explicit-cover}, for any $0 < \eps \leq e^{-42}/3$, if we fix $n_0  =  \lfloor \frac{1}{6} \ln(1/3\eps) \rfloor$,
there is  a $3\eps$-packing for $\mathcal{S}_{n_0, 2}$ of size $(1/\eps)^{\Omega \left(\log (1/\eps)\right)}$.
From the argument of the previous paragraph, any $\eps$-cover for $\mathcal{S}_{n_0, 2}$ is of size $(1/\eps)^{\Omega \left(\log (1/\eps)\right)}$.

To prove the desired lower bound of $n \cdot (1/\eps)^{\Omega \left(\log (1/\eps)\right)}$ we construct appropriate ``shifts'' of the set $\mathcal{S}_{n_0, 2}$
as follows: Consider the set $\mathcal{S}_{n, 2}$ where $n \geq r (n_0+1)$ for some $ r \in \Z_+$. For $0\leq i < r$, let ${\cal S}^i_{n,2}$ be the subset of ${\cal S}_{n,2}$
where $i(n_0+1)$ of the parameters $p_j$ are equal to $1$, and at most $n_0$ other $p_j$'s are non-zero. Note that for $i\neq j$ any elements of ${\cal S}^i_{n,2}$ and
${\cal S}^j_{n,2}$ have disjoint supports. Therefore, any $\eps$-cover of ${\cal S}_{n,2}$ must contain disjoint $\eps$-covers for ${\cal S}^i_{n,2}$ for each $i$.
Note also that ${\cal S}^i_{n,2}$ is isomorphic to ${\cal S}_{n_0,2}$ for each $i$, and thus has minimal $\eps$-cover size at least
$(1/\eps)^{\Omega \left(\log (1/\eps)\right)}$.
Therefore, any $\eps$-cover of ${\cal S}_n$ must have size at least
$\lfloor n/ n_0 \rfloor  \cdot (1/\eps)^{\Omega(\log(1/\eps))} = n (1/\eps)^{\Omega(\log(1/\eps))}.$

\end{proof}

\subsection{Cover Size Lower Bound for $k$-SIIRVs} \label{ssec:ksiirv}

In this section, we prove our cover lower bound for $k$-SIIRVs:

\begin{theorem} \label{thm:ks-cover-lower}
For $0< \epsilon \leq e^{-12}(2k)^{-9}$ and $n \le  \lfloor \frac{1}{12} \log(1/\epsilon) \rfloor$,
there is an $\eps$-packing of $\mathcal{S}_{n, k}$ under $\dtv$ with cardinality
$(1/\eps)^{\Omega(nk)}$.
\end{theorem}
\begin{proof}
We consider $k$-SIIRVs close to the $(k-1)$ multiple of the $2$-SIIRV $\p_0$
with parameters $p_i=\frac{i}{n+1} $ we used for the explicit lower bound
in Section~\ref{ssec:explicit-pbds}.
Let $m \in \Z_+$ and $0< \delta < 1$ be parameters that will be fixed later.
Given an $\ba \in [m]^{n(k-2)}$, which
will index by $a_{ij}$, for $i \in \{1,\ldots, n\}$ and $j \in \{1,\ldots,k-2\}$,
we define a $k$-SIIRV $\p_\ba$ as follows. For each $i$, we take a $k$-IRV $Y_i$ with pdf defined as follows:
\begin{eqnarray*}
\Pr[Y_i = 0] &=& (1-p_i) \left(1 - \delta \cdot \sum_j a_{ij}\right), \\
\Pr[Y_i = j] &=&  \delta \cdot a_{ij}, \quad 1\le j \le k-2, \\
\Pr[Y_i = k-1] &=&  p_i \left(1 - \delta \sum_j a_{ij}\right) \;.
\end{eqnarray*}
For convenience, we will denote $\gamma_{\ba,i}= \left(1 - \delta \cdot \sum_j a_{ij}\right).$
We claim that the set of distributions $\p_\ba$,
$\ba \in [m]^{n(k-2)}$, is an $\eps$-packing.
To prove this statement we proceed similarly to  the proof of Theorem~\ref{thm:explicit-cover}.
For a distribution $\p$, we will consider the expectations
$$r_{\p ,ij}=\sum_{l=0}^n p_i^{n-l}(p_i-1)^l \p(l(k-1)+j)$$
for $i \in \{1, \ldots ,n\}$ and $j \in \{1, \ldots ,k-2\}$.
Similarly to Claim \ref{claim:dist-vs-poly}, we have the following:
\begin{claim} \label{claim:dist-vs-poly2}
Let $\p, \q \in {\cal S}_{n, k}$ such that $\dtv\left(\p, \q \right) < \eps$.
Then for any $i \in \{1, \ldots ,n\}$ and $j \in  \{1, \ldots ,k-2\}$, we have that $$| r_{\p,ij} - r_{\q,ij} | < 2\eps.$$
\end{claim}
\begin{proof}
We have the following sequence of (in)equalities:
\begin{eqnarray*}
| r_{\p,ij} - r_{\q,ij} |  &=& \left| \sum_{l=0}^n (\p(l(k-1)+j) - \q(l(k-1)+j))  p_i^{n-l} (p_i-1)^l \right|  \\
                                 &\le & \sum_{i=0}^n \left| (\p(l(k-1)+j) - \q(l(k-1)+j) \right| \cdot \left| p_i^{n-l} (p_i-1)^l \right| \\
                                 &\le & \sum_{i=0}^n \left| \p(l(k-1)+j) - \q(l(k-1)+j) \right|  \leq 2 \dtv\left(\p, \q \right) \\
                                 &<& 2\eps \;,
\end{eqnarray*}
where the second line is the triangle inequality and the third line uses the fact that
$|p_i^{n-l} (p_i-1)^l| \le 1$ for all $l \in [n]$ and $i \in \{1, \ldots ,n\}$.
\end{proof}
By the above claim, to complete the proof,
it suffices to show that $|r_{\p_\ba,ij}-r_{\p_\bb,ij}| \geq 2 \eps$ whenever $a_{ij} \neq b_{ij}$.
To prove this statement, we exploit the fact that these $k$-SIIRVs are close to a multiple of $\p_0$,
by ignoring terms in the expectations that are $O(\delta^2)$.

Let $Y=\sum_{i=1}^n Y_i$ with $ Y \sim \p_\ba$ for a given $\ba \in [m]^{n(k-2)}$.
We define several events depending on which coordinates $Y_i$ are equal to $0$ or $k-1$, and
consider their contribution to the expectation $r_{\p_\ba,ij}$ separately.

Firstly, let $A_{\geq 2}$ be the event that more than one $Y_i$ is not $0$ or $k-1$ .
The probability that any fixed $Y_i$ is not $0$ or $k-1$ is small, namely
$$\sum_{j=1}^{k-2} \pr[Y_i=j] = \sum_{j=1}^{k-2} \delta a_{ij} \leq (k-2)m \delta \;.$$
Hence, $$\pr[A_{\geq 2}]  \leq {n \choose 2} \left((k-2)m \delta \right)^2 \leq \frac{1}{2} \cdot (n(k-2)m \delta)^2 \;.$$
The contribution of  $A_{\geq 2}$ to $r_{\p_\ba,ij}$ is
$ r_{\p_\ba,ij,A_{\geq 2}} :=  \sum_{l=0}^n p_i^{n-l}(p_i-1)^l \pr_{Y \sim \p_\ba}\left[Y=l(k-1)+j \cap A_{\geq 2}\right] \;,$
and therefore
$$|r_{\p_\ba,ij,A_{\geq 2}}| \leq \frac{1}{2}(n(k-2)m \delta)^2\;, $$
since $|p_i^{n-l}(p_i-1)^l| \leq 1$.

Secondly, let $A_0$ be the event that all $Y_i$'s are $0$ or $k-1$.
If $A_0$ occurs then $Y$ is a multiple of $k-1$.
 Thus, for $l \in [n]$ and $j \in \{1, \ldots ,k-2\}$, we have $\pr_{Y \sim \p_\ba}\left[Y=l(k-1)+j \cap A_0 \right]=0$.
The contribution of  $A_{0}$ to $r_{\p_\ba,ij}$ is
$$ r_{\p_\ba,ij,A_{0}} :=  \sum_{l=0}^n p_i^{n-l}(p_i-1)^l \pr_{Y \sim \p_\ba}\left[Y=l(k-1)+j \cap A_{0}\right] = 0 \;.$$
Finally, for $i \in \{1, \ldots ,n\}$, let $B_i$ be the event that $Y_i$ is the only $k$-IRV
that takes a value between $1$ and $k-2$. The probability of all other $Y_h$,
with $h \neq i$, being either $0$ or $k-1$ is $\prod_{h \neq i} \gamma_{\ba,h}$.
We consider the RVs $X_{-i} = \sum_{h\neq i} X_h$, where $X_h \sim \mathrm{Ber}(p_h)$.
That is,  $X_{-i} \sim \p_{-i} \in \mathcal{S}_{n-1,2}$, i.e., it is a $2$-SIIRV with parameters $p_h$ for $h \neq i$.
Then, the conditional probability $\pr\left[\sum_{h \neq i} Y_h = l(k-1) | (B_i \cup A_0) \right]$
is equal to $\pr\left[X_{-i}=l \right]=\p_{-i}(l)$ for all $l \in [n]$.
So, for all $l \in [n]$ and $j \in \{1, \ldots ,k-2\}$ we have
\begin{eqnarray*}
\pr\left[Y=l(k-1)+j \cap B_i \right]  & = & \pr\left[\sum_{h \neq i} Y_h = l(k-1) \cap (B_i \cup A_0)\right] \pr[Y_i=j] \\
							& = & \left(\prod_{h \neq i} \gamma_{\ba,h}\right) \p_{-i}(l) \delta a_{ij}
\end{eqnarray*}
Then, the contribution of $B_i$ to  $r_{\p_\ba,gj}$ is
\begin{eqnarray*}
r_{\p_\ba,gj,B_i}  & := & \sum_{l=0}^n p_g^{n-l}(p_g-1)^l \pr_{Y \sim \p_\ba}\left[Y=l(k-1)+j \cap B_i\right] \\
					& = & \left(\prod_{h \neq i} \gamma_{\ba,h}\right) \cdot \delta a_{ij} \cdot  \sum_{l=0}^n p_g^{n-l}(p_g-1)^l \p_{-i}(l) \\
					& = & \left(\prod_{h \neq i} \gamma_{\ba,h}\right) \cdot \delta a_{ij} \cdot r_{\p_{-i}}(p_g) \\
					& = & \left(\prod_{h \neq i} \gamma_{\ba,h}\right) \cdot \delta a_{ij} \cdot  \prod_{h \neq i} (p_h - p_g) \;,
\end{eqnarray*}
where $r_{\p_{-i}}$ above is as defined in the previous section, and when $g \neq i$,
the second product includes the term $p_g - p_g=0$, so $r_{\p_\ba,gj,B_i} = 0$.
Summing these contributions to the expectation $r_{\p_\ba,ij}$ gives:
\begin{eqnarray*}
r_{\p_\ba ,ij}   & = & r_{\p_\ba,ij,A_{\geq 2}} + r_{\p_\ba,ij,A_0}+ \sum_{g=1}^n r_{\p_\ba,ij,B_g} \\
			& = & r_{\p_\ba,ij,A_{\geq 2}} +  r_{\p_\ba,ij,B_i} \\
			& = & r_{\p_\ba,ij,A_{\geq 2}} +  \prod_{h \neq i} \gamma_{\ba,h} \cdot  \delta a_{ij} \cdot  \prod_{h \neq i} (p_h - p_i)
\end{eqnarray*}
Now consider $\ba$ and $\bb$ which for some $i \in \{1,2,...,n\}$ and $j \in \{1,2,...,k-2\}$ have  $a_{ij} \neq b_{ij}$.
We have that $\prod_{h \neq i} |p_h-p_i| \geq e^{-3n}$ by Equation (\ref{eqn:dyo}), and thus,
\begin{align*}
\prod_{h \neq i} \gamma_{\ba,h} = \prod_{h \neq i} \left(1 - \delta \littlesum_j a_{hj}\right) \geq (1-(k-2)m \delta)^{n-1} \geq (1-(n-1)(k-2)m \delta),
\end{align*}
 $|a_{ij}-b_{ij}| \geq 1$, and
$|r_{\p_\ba,ij,A_{\geq 2}}| \leq \frac{1}{2}(n(k-2)m \delta)^2$.

We obtain the following sequence of inequalities:
\begin{eqnarray*}
&& | r_{\p_\ba,ij} - r_{\p_\bb,ij} | = |r_{\p_\ba,ij,B_i} - r_{\p_\bb,ij,B_i}  + r_{\p_\ba,ij,A_{\geq 2}} - r_{\p_\bb,ij,A_{\geq 2}}  | \\
&\geq& \Big| \littleprod_{h \neq i} (p_h - p_g)\big(\littleprod_{h \neq i} \gamma_{\ba,h} \delta a_{ij} - \littleprod_{h \neq i} \gamma_{\bb,h} \delta b_{ij}\big) \Big| - (n(k-2)m \delta)^2 \\
&\geq&  e^{-3n}\big|\littleprod_{h \neq i} \gamma_{\ba,h} \delta \big| \cdot |a_{ij}-b_{ij}|
- e^{-3n} \delta b_{ij}\big|\littleprod_{h \neq i} \gamma_{\ba,h}-\littleprod_{h \neq i} \gamma_{\bb,h} \big|- (n(k-2)m \delta)^2\\
&\geq& e^{-3n} (1-(n-1)(k-2)m \delta) \delta - \delta m\Big(\big|1- \littleprod_{h \neq i} \gamma_{\ba,h}\big|+\big|1- \littleprod_{h \neq i} \gamma_{\bb,h}\big| \Big) - (n(k-2)m \delta)^2 \\
&\geq& e^{-3n} \delta - 2\delta m n(k-2)m \delta - 2(n(k-2)m \delta)^2 \\
&\geq& e^{-3n} \delta - 3(n(k-2)m \delta)^2 \;.
\end{eqnarray*}
Recall that by assumption $\epsilon \leq e^{-12}(2k)^{-9}$.
We set $n=\lfloor \frac{1}{12} \log (1/ \eps) \rfloor$,
$\delta=3\eps^{3/4}$, and
$m= \lfloor \frac{\eps^{-1/4}} {2n^2(k-2)^2} \rfloor$.
Then,  $e^{-3n} \delta \geq 3 \eps$ and
$3 (n(k-2)m\delta)^2 \leq \eps$. So,  we have that
$| r_{\p_\ba,ij} - r_{\p_\bb,ij} | \geq 2 \eps$ as required.
Also, $\gamma_\ba \geq 1 - \sqrt{\eps} \geq 0$, so the $k$-IRVs are indeed well-defined.

Therefore, we have exhibited a set of $m^{n(k-2)}$ $k$-SIIRVs that have pairwise total variation distance
at least $\eps$. The proof follows by observing that
$m^{n(k-2)} = (1/\eps)^{\Omega(k \log 1/\eps)}$.
\end{proof}

\section{Sample Complexity Lower Bound} \label{sec:sample-lb}

In this section, we prove our sample complexity lower bounds.
We start with the case $k=2,$ and then generalize our construction for an arbitrary value of $k.$
As mentioned in the introduction, our sample lower bounds make crucial use of a geometric
characterization of the space of $k$-SIIRVs. In Section~\ref{ssec:geom-pbd}, we describe our geometric characterization
for $2$-SIIRVs, and in Section~\ref{ssec:sample-lb-pbd} we use it to prove our $2$-SIIRV sample lower bound.
Similarly, in Section~\ref{ssec:geom-siirv}, we describe our geometric characterization
for $k$-SIIRVs, and in Section~\ref{ssec:sample-lb-siirv} we use it to prove our $k$-SIIRV sample lower bound.

\subsection{A Useful Structural Result for $2$-SIIRVs} \label{ssec:geom-pbd}

In this subsection, we prove a novel structural result for the space of $2$-SIIRVs (Lemma~\ref{lem:explicit-eps-ball}).
This allows us to obtain a simple non-constructive lower bound on the cover size of $2$-SIIRVs under the Kolmogorov distance metric.
More importantly, this lemma is crucial for our tight sample complexity lower bound of the following subsection.

Before we state our lemma, we provide some basic intuition.
The set of all distributions supported on $[n]$ is $n$-dimensional (viewed as a metric space).
Note that each $\p \in {\cal S}_{n,2}$ is defined by $n$ parameters.
It turns out that ${\cal S}_{n,2}$ is also $n$-dimensional in a precise sense.
This intuition is formalized in the following lemma:

\begin{lemma} \label{lem:explicit-eps-ball}
(i) Given any $\p \in {\cal S}_{n,2}$ with distinct parameters in $(0, 1)$, there is a radius  $\delta = \delta(\p)$
such that any distribution $\q$ with support $[n]$ that satisfies $\dk(\p, \q) \leq \delta$ can also be expressed as a $2$-SIIRV,
i.e., $\q \in {\cal S}_{n,2}.$

(ii) Let $\p_0 \in {\cal S}_{n,2}$ be the $2$-SIIRV with parameters $p_i = \frac{i}{n+1}$, $1 \le i \le n$.
Then any distribution $\q$ with support $[n]$ that satisfies $\dk(\p_0,\q) \leq 2^{-9n}$
is itself a $2$-SIIRV with parameters $q_i$ such that $|q_i-p_i| \leq \frac{1}{4(n+1)}$. \end{lemma}

\begin{proof}
We consider the space of cumulative distribution functions (CDF's) of all distributions of support $[n]$.
Let $T_n$ be the set of sequences $0\leq x_1\leq x_2\leq \ldots \leq x_n \leq 1$. Consider the map $\mathcal{P}_n : T_n \to T_n$ defined as follows:
For  $\mathbf{p}=(p_1,\dots,p_n) \in T_n$ (i.e., with ordered parameters $0\leq p_1\leq \ldots \leq p_n \leq 1$), let $\p$ be the corresponding $2$-SIIRV in $\mathcal{S}_{n, 2}$.
For $i \in \{1, \ldots ,n \}$, let $(\mathcal{P}_n(\mathbf{p}))_i=\p(< i)$. Namely, $\mathcal{P}_n$ maps a sequence of probabilities
to the sequence of probabilities defining the CDF of the corresponding $2$-SIIRV.

The basic idea of the proof is that the mapping $\mathcal{P}_n$ is invertible in a neighborhood of a point $\mathbf{p}$ with distinct coordinates.
This allows us to uniquely obtain the distinct parameters of a $2$-SIIRV $\p \in {\cal S}_{n,2}$ from its CDF.
We will make essential use of the inverse function theorem for $\mathcal{P}_n$, which we now recall:

\begin{theorem}[Inverse function theorem~\cite{rudin-principles}] \label{thm:inverse-function}
Let $F:S \rightarrow \mathbb{R}^n$, $S \subseteq \mathbb{R}^n$, be a continuously differentiable function
and $\mathbf{x}$ be a point in the interior of $S$ such that the Jacobian matrix of $F$, $\mathrm{Jac}(F)(\mathbf{x})$,  is non-singular.
Then there exists an inverse function, $F^{-1}$, of $F$ in a neighborhood of $F(\mathbf{x})$.
Furthermore the inverse function $F^{-1}$ is continuously differentiable and
its Jacobian matrix satisfies $\mathrm{Jac}(F^{-1})(F(\mathbf{x}))=\left(\mathrm{Jac}(F)(\mathbf{x})\right)^{-1}$. \end{theorem}

We will apply the inverse function theorem for $\mathcal{P}_n$ at the point $\mathbf{p}$ defining the distinct parameters of the $2$-SIIRV $\p$
in the statement of the theorem. It is easy to see that  $\mathcal{P}_n$ is continuously differentiable.
The main part of the argument involves proving that
the Jacobian matrix of $\mathcal{P}_n$ at $\mathbf{p}$,  $\mathrm{Jac}(\mathcal{P}_n)(\mathbf{p})$,  is non-singular.

Recall that  $\mathrm{Jac}(\mathcal{P}_n)(\mathbf{p})$
is the $n \times n$ matrix whose $(i, j)$ entry is the partial derivatives of $(\mathcal{P}_n)_i$ in direction $j$,  i.e.,
$(\mathrm{Jac}(\mathcal{P}_n)(\mathbf{p}))_{ij}=\frac{\partial (\mathcal{P}_n(\mathbf{p}))_i}{\partial p_j}$.
We start by showing the following lemma:

\begin{lemma} \label{lem:jacobian-equation}
For a $2$-SIIRV $\p \in {\cal S}_{n,2}$ with parameters $\mathbf{p}$, we have
\begin{equation} \label{eq:jacobian-equation}
M(\mathbf{p}) \cdot  \mathrm{Jac}(\mathcal{P}_n)(\mathbf{p}) = - \mathrm{diag}\left(\prod_{j \ne i} (p_i - p_j) \right) \end{equation}
where $M(\mathbf{p})$ is the $n \times n$ matrix with entries $(M (\mathbf{p}))_{ij}=(1-p_i)^{j-1} p_i^{n-j}$, $1 \le i, j \le n$.
Here, for $x \in \R^n$, we denote by $\mathrm{diag}(x)$ the diagonal matrix with entries $(\mathrm{diag}(x))_{ii}=x_i$.
\end{lemma}
\begin{proof}
To calculate the partial derivative $\frac{\partial (\mathcal{P}_n(\mathbf{p}))_i}{\partial p_j}$, we isolate the effect of the parameter $p_j$ from the other variables.
In particular, for $X \sim \p$, i.e., $X= \sum_{i=1}^n X_i$, with $X_i \sim \mathrm{Ber}(p_i)$, we can write
$X = X_{-j} + X_j$, where $X_{-j} = \sum_{i \neq j} X_i$. Note that $X_j \sim \p_{-j} \in {\cal S}_{n-1,2}$, i.e., it is
the $(n-1)$ parameter $2$-SIIRV with parameters $p_i$ for $i\neq j$.
Now, for $1 \le i \le n$, we can write
$$
(\mathcal{P}_n(\mathbf{p}))_i = \p(< i) = \p_{-j}(<(i-1)) + (1-p_j)\p_{-j}(i-1).
$$
The derivative of this quantity with respect to $p_j$ equals
$
\frac{\partial (\mathcal{P}_n(\mathbf{p}))_i}{\partial p_j} = -\p_{-j}(i-1) .
$
Therefore, the $j$-th column of $\mathrm{Jac}(\mathcal{P}_n)(\mathbf{p})$ equals $-1$ times the pdf of the distribution $\p_{-j}$.
This allows us to consider multiplying on the right by $\mathrm{Jac}(\mathcal{P}_n)(\mathbf{p})$ as taking the expectations of certain distributions.
In particular,  for $y \in \R^n$ and any $1 \le j \le n$, we have that
$$(y^T  \mathrm{Jac}(\mathcal{P}_n)(\mathbf{p}))_j = - \sum_{i=1}^n y_i \p_{-j}(i-1) = - \E \left[ y_{X_{-j}+1} \right] \;.$$
Therefore, for $1 \le i, j \le n$, we can write
\begin{eqnarray*}
(M(\mathbf{p}) \cdot \mathrm{Jac}(\mathcal{P}_n)(\mathbf{p}))_{ij}
& = & - \sum_{k=1}^n  (p_i-1)^{k-1} p_i^{n-k} \p_{-j}(k-1)  
 =  - \E \left[ (p_i-1)^{X_{-j}} p_i^{n-X_{-j}-1} \right]\\
& = & - \E \left[ \prod_{k \neq j} (p_i-1)^{X_k} p_i^{1-X_k} \right] 
 =  - \prod_{k \neq j}  \E \left[ (p_i-1)^{X_k} p_i^{1-X_k}\right] \\
 &=&  - \prod_{k \neq j} \left[ (p_i-1)p_k + p_i(1-p_k) \right] 
 =  - \prod_{k \neq j} (p_i - p_k) \;.
\end{eqnarray*}
Note that for $i \neq j$, the above product contains the term $(p_i-p_i)$ and so is equal to $0$.
When $i=j$, we have $(M(\mathbf{p}) \cdot \mathrm{Jac}(\mathcal{P}_n)(\mathbf{p}))_{ii} = - \prod_{k \neq i} (p_i - p_k)$
completing the proof of the lemma.
\end{proof}

We are now ready to prove part (i) of Lemma~\ref{lem:explicit-eps-ball}.
To this end, consider a $2$-SIIRV $\p$ with distinct parameters $\mathbf{p}$, i.e.,  $p_i \neq p_j$ for $i \neq j$,  such that $p_i \in (0,1)$ for all $i$.
Note that $\mathbf{p}$ lies in the interior of $T_n$.
Moreover, for all $i$, we have $\prod_{j\neq i} (p_i - p_j) \neq 0$ and therefore the matrix $\mathrm{diag}(\prod_{j\neq i} (p_i - p_j))$ appearing in (\ref{eq:jacobian-equation}) is non-singular. It follows from Lemma~\ref{lem:jacobian-equation} that both matrices on the LHS of (\ref{eq:jacobian-equation}) are non-singular.
In particular, $\mathrm{Jac}(\mathcal{P}_n)(\mathbf{p})$ is non-singular, hence we can apply the inverse function theorem. As a corollary, there exists an inverse mapping $\mathcal{P}_n^{-1}$ in some neighborhood of $\mathcal{P}_n (\mathbf{p})$. Specifically, there is some $\delta > 0$ such that $\mathcal{P}_n^{-1}$ is defined
at every $\mathbf{x} \in T_n$ with $\|\mathbf{x}- \mathcal{P}_n (\mathbf{p})\|_\infty \leq \delta$.

Let $\q$ be a distribution over $[n]$ satisfying $\dk(\p,\q) \leq \delta$. Equivalently, if $\mathbf{y} = \left( \q(< i) \right)_{i=1}^n \in T_n$
is the CDF of $\q$, then $\|\mathcal{P}_n (\mathbf{p}) - \mathbf{y} \|_\infty \leq \delta$. Thus  $\mathcal{P}_n^{-1}$ is defined at $\mathbf{y}$
and $\mathbf{q} = \mathcal{P}_n^{-1} (\mathbf{y})  \in T_n$ are the parameters of a $2$-SIIRV with distribution $\q.$
Thus, $\q$ is a $2$-SIIRV with parameters  $\mathbf{q}$, which completes the proof of (i). Note that the proof also implies that $\q$ in this neighborhood can
be taken to be $\mathcal{P}_n(\mathbf{q}')$ for $\mathbf{q}'$ in some small neighborhood of $\mathbf{p}$.

To prove part (ii) of Lemma~\ref{lem:explicit-eps-ball},
we use a geometric argument. Recall that the parameters of $\p_0$ are  $\mathbf{p}_0 =\left(\frac{1}{n+1},\ldots,\frac{n}{n+1}\right)$.
Let $S \subseteq T_n$ be the set of vectors
$\mathbf{p}$ with $\|\mathbf{p}-\mathbf{p}_0\|_{\infty} \leq \frac{1}{4(n+1)}$. By Lemma \ref{distSepLem} we have that any $\q$ in $\mathcal{P}_n(\partial S)$ satisfies $\dtv(\p_0,\q) \geq \frac{e^{-3n}}{4(n+1)},$ and therefore $\dk(\p_0,\q) \geq \frac{e^{-3n}}{8(n+1)^2} \geq 2^{-9n}.$


Let $B$ be the set of distributions $\q$ on $[n]$ so that $\dk(\p_0,\q)\leq 2^{-9n}$. We claim that $\mathcal{P}_n(S)\cap B=B.$ 
To begin, note that $S$ is compact, and therefore this intersection is closed. 
On the other hand, since $\mathcal{P}_n(\partial S)$ is disjoint from $B,$ 
this intersection is $\mathcal{P}_n(\mathrm{int}(S))\cap B.$ 
On the other hand, since $\mathcal{P}_n$ has non-singular Jacobian on $\mathrm{int}(S),$ 
the open mapping theorem implies that $\mathcal{P}_n(\mathrm{int}(S))\cap B$ is an open subset of $B.$ 
Therefore, $\mathcal{P}_n(S)\cap B$ is both a closed and open subset of $B$, and therefore, since $B$ is connected, it must be all of $B.$
This completes the proof of part (ii).
\end{proof}

As a simple application of our structural lemma, we obtain a non-constructive lower bound on the cover size under the Kolmogorov distance metric:

\begin{theorem} \label{thm:lb-non-con}
For any $\eps > 0$ and $n = \Omega(\log(1/\eps))$ any $\eps$-cover of $\mathcal{S}_{n,2}$
under $\dk$ must have size at least $n \cdot (1/\eps)^{\Omega(\log(1/\eps))}.$
\end{theorem}
\begin{proof}
Note that by an argument identical to that of Corollary~\ref{cor:simple} it suffices to prove a packing lower bound of
$(1/\eps)^{\Omega(\log(1/\eps))}$ for $n = \Theta(\log(1/\eps))$.

To that end, fix $n = n_0 = \lfloor \frac{1}{18} \log_2(1/\eps) \rfloor$. Then, we have $2^{-9n} \geq \sqrt{\eps}$.
By Lemma \ref{lem:explicit-eps-ball}(ii),  there is a $2$-SIIRV $\p_0 \in \mathcal{S}_{n , 2}$, such that any distribution $\q$ with support $[n]$
and $\dk(\p_0,\q) \leq \sqrt{\eps}$ is
in $\mathcal{S}_{n , 2}$. We will give an $\eps$-packing lower bound for this subset of $2$-SIIRVs.

Let us denote by $\mathbf{z} \in T_n$ the vector defining the CDF of $\p_0$, i.e., $\mathbf{z} = (\p_0(<i))_{i=1}^n.$
Let $S \subseteq \R^n$ be the set of points $\mathbf{x} \in \R^n$ with
$\|\mathbf{x}-\mathbf{z} \|_{\infty} \leq \sqrt{\eps}$. Note that $S$ is an $n$-cube with side length $2\sqrt{\eps}$.

We claim that every $\mathbf{x} \in S$ is the CDF of a $2$-SIIRV $\q \in \mathcal{S}_{n , 2}$ . By Lemma \ref{lem:explicit-eps-ball}, this follows
immediately if $\mathbf{x} \in T_n$, i.e., if $\mathbf{x}$ is the CDF of a distribution. So, it suffices to show that $S \subseteq T_n$.
Suppose for the sake of contradiction that there is a point $\mathbf{y} \in S \setminus T_n$. Then, there is a point
$\mathbf{x} \in S$ such that $\mathbf{x}$ lies on the boundary of $T_n$. 
For such a point $\mathbf{x}$, one of the inequalities $0 \leq x_1 \leq x_2 \leq \ldots \leq x_n \leq 1$ is tight. Thus, $\mathbf{x}$ is
the CDF of a distribution $\q$ which has $\q(i)=0$ for some $i$.
Since $\mathbf{x} \in S \cap T_n$, $\q$ is a $2$-SIIRV with parameters given by Lemma \ref{lem:explicit-eps-ball}.
In particular $\q$ does not have any parameters equal to $0$ or $1$. Thus, we have $\q(i) > 0$ for all $i \in [n]$, a contradiction.

Therefore, any $\eps$-cover of $\mathcal{S}_{n,2}$ in Kolmogorov distance
induces an $\eps$-cover of the same size in $L_\infty$ distance of the CDFs of distributions in $\mathcal{S}_{n,2}$.
If $s$ is the size of such a cover,  then we have $s$ $n$-cubes of side length $\eps$ whose union contains $S$.
Recall that  $S$ is an $n$-cube of side length $\sqrt{\eps}$.
The volume of each of these $s$ $n$-cubes is $(2\eps)^n$
and the volume of $S$ is $(2\sqrt{\eps})^n$.
The volume of the union of $s$ $n$-cubes is at most $s \cdot (2\eps)^n$ and hence
$s \cdot (2\eps)^n \geq (2\sqrt{\eps})^n$ or $s = (1/\eps)^{\Omega(n)}$, which competes the proof.
\end{proof}

\subsection{Sample complexity lower bound for $2$-SIIRVs} \label{ssec:sample-lb-pbd} 
In this subsection, we prove our tight sample lower bound for learning $2$-SIIRVs. 
Our proof uses a combination of information-theoretic arguments
and the structural lemma of the previous subsection. In particular, we show:

\begin{theorem} [Sample Lower Bound for $2$-SIIRVs] \label{thm:sample-lower-bound}
Let $\mathcal{A}$ be any algorithm which, given as input $n$, $\eps$, and sample access to an unknown $\p \in \mathcal{S}_{n, 2}$
outputs a hypothesis distribution $\mathbf{H}$ such that $\E[\dtv(\mathbf{H}, \p)] \le \eps$. Then, $\mathcal{A}$ must use
$\Omega((1/\eps^2) \cdot \sqrt{\log(1/\eps)})$ samples.
\end{theorem}

Our main information-theoretic tool to prove our lower bound is Assouad's Lemma~\cite{Assouad:83}.
We recall the statement of the lemma (see, e.g.,~\cite{DG85}), tailored to discrete distributions below:

\begin{theorem} \label{thm:assouad} [Theorem 5, Chapter 4, \cite{DG85}]
Let $r \geq 1$ be an integer. For each $\mathbf{b}  \in \{-1,1\}^r$,
let $\p_{\mathbf{b}}$ be a probability distribution over a finite set $A$.
For $ 1 \le \ell \le r$ and $\mathbf{b} \in \{-1,1\}^r$, we denote by $\mathbf{b}^{(\ell, +)}$ (resp. $\mathbf{b}^{(\ell, -)}$) the vector
with $\mathbf{b}^{(\ell, +)}_i = \mathbf{b}_i$ (resp. $\mathbf{b}^{(\ell, -)}_i = \mathbf{b}_i$) for
$i \neq \ell$ and $\mathbf{b}^{(\ell, +)}_\ell = 1$ (resp. $\mathbf{b}^{(\ell, -)}_\ell = -1$).
Suppose there exists a partition $A_0, A_1, \ldots, A_r$ of $A$ such that for all $\mathbf{b}  \in \{-1,1\}^r$
and all $1 \le \ell \le r$, the following inequalities are valid:
\begin{enumerate}
\item[(a)] $\sum_{x \in A_{\ell}} |\p_{\mathbf{b}^{(\ell, +)}}(x) - \p_{\mathbf{b}^{(\ell, -)}}(x)|  \geq \alpha$, and

\item[(b)] $\sum_{x \in A} \sqrt{\p_{\mathbf{b}^{(\ell, +)}}(x) \p_{\mathbf{b}^{(\ell, -)}}(x)}  \geq 1-\gamma > 0.$
\end{enumerate}
Then, for any any algorithm $\mathcal{A}$ that draws $s$ samples from an unknown $\p \in \p_{\mathbf{b}}$
and outputs a hypothesis distribution $\mathbf{H}$, there is some $\mathbf{b} \in \{-1,1\}^r$
such that if the target distribution $\p$ is $\p_{\mathbf{b}}$,
$$\E \left[ \dtv(\p ,\mathbf{H}) \right] \geq (r \alpha/4)(1 - \sqrt{2 s \gamma}).$$
\end{theorem}

Recall that $2$-SIIRVs are discrete log-concave distributions. We will use the following basic properties of log-concave distributions:
\begin{lemma} \label{lem:lc-stddev}
There exists a universal constant $c>0$ such that the following holds:
For any log-concave distribution $\p$ supported on the integers and standard deviation $\sigma$, there exist at least $\Omega (\sigma)$
consecutive integers with probability mass under $\p$ at least $ c \cdot \frac{1}{1+\sigma}$.
\end{lemma}
\begin{proof}
Note that if $\sigma\leq 1$, taking the mode trivially satisfies this property.

Without loss of generality we can assume that $0$ is the mode of $\p$.
We know that $\sum_{x\in\Z}x^2 \p(x) = \Theta(\sigma^2).$ Let $\sigma_+^2 = \sum_{x>0}x^2 \p(x).$
Let $t_+$ be the largest integer so that $\p(t_++1)/\p(t_+) \leq e^{1/t_+}$. We note that
\begin{align*}
\sum_{x>0}x^2 \p(x) & \leq \sum_{x=0}^\infty x^2 \p(t_+) e^{-(x-t_+)/t_+} = \Theta(t_+^3 \p(0)),
\end{align*}
and
$$
\sum_{x>0}x^2 \p(x) \geq \p(t_+) \sum_{x=0}^{t_+} x^2 = \Theta(t_+^3 \p(0)).
$$
Also note that
$$
\sum_{x>0} \p(x) \leq \sum_{x=0}^\infty \p(t_+) e^{-(x-t_+)/t_+} = \Theta(t_+ \p(0)).
$$
Similarly, defining $\sigma_-$ and $t_-$, we find that $\sigma^2 = \Theta( \sigma_+^2+\sigma_-^2) =\Theta(\p(0)(t_+^3 + t_-^3)).$ 
Thus, $\max(t_+,t_-)^3\p(0)=\Theta(\sigma^2)$ and $\max(t_+,t_-)\p(0) = \Omega(1)$. 
Without loss of generality this maximum is $t_+$. 
Note that for all $0\leq x\leq t_+$ that $\p(x)= \Theta( \p(t_+))$. This implies that $t_+\p(0)=O(1)$, and thus, by the above is $\Theta(1)$. 
Therefore, it follows by the variance bounds that $t_+^2 = \Omega(\sigma^2)$, so $t_+=\Theta(\sigma).$
Hence,  $x=0,1,\ldots,t_+$ are $\Omega(\sigma)$ terms on which the value of $\p$ is $\Omega(1/t_+)=\Omega(1/\sigma).$ 
This completes the proof.
\end{proof}

We are now ready to prove Theorem~\ref{thm:sample-lower-bound}.

\begin{proof}[Proof of Theorem \ref{thm:sample-lower-bound}]
Ideally, we would like to use the set of $2$-SIIRVs whose parameters are explicitly described in Theorem \ref{thm:explicit-cover}
in our application of Assouad's lemma. Unfortunately, however, this particular set is not in a form
that allows a direct application of the theorem.
The difficulty lies in the fact that it is not clear how to isolate the changes
between distributions in disjoint intervals using explicit parameters.

We therefore proceed with an indirect approach making essential use of Lemma \ref{lem:explicit-eps-ball}(ii).
We start from the $2$-SIIRV $\p_0$ in the statement of the lemma and we perturb its pdf appropriately to construct
our ``hypercube'' distributions $\p_{\mathbf{b}}.$ The lemma guarantees that, if the perturbation is small enough,
all these distributions are indeed $2$-SIIRVs.

Observe that the variance of $\p_0$ is $\Omega(n)$ since $\Omega(n)$ parameters $p_i$ lie in $[1/4, 3/4].$
By Lemma \ref{lem:lc-stddev}, there exist $r=\Omega(\sqrt{n})$ consecutive integers,
an integer $m$, $0 \leq m \leq n$, and a real value $t$ with $t \geq c \cdot r$,
such that for all $i$, with $m \leq i \leq m+2r$, we have
$$\p(i) \geq \frac{2}{t} \;.$$
For $n$ sufficiently large, we can assume that $2^{-9n} \leq c$ and therefore
$\frac{1}{t} \geq \frac{2^{-9n}}{r}.$

We are now ready to define our ``hypercube'' of $2$-SIIRVs.
For $\bb \in \{ -1, 1 \}^r$, consider the distribution $\p_\bb$ with
$$\p_\bb(i)=\begin{cases} \p_0(i) & \text{if } i < m \text{, } i > m+2r, \text{ or } \bb_{\lfloor \frac{1}{2}(i-m) \rfloor}=-1 \\
					\p_0(i) - \frac{2^{-9n}}{r} & \text{if } \bb_{\lfloor \frac{1}{2}(i-m) \rfloor}=1 \text{ and } i \text{ is even} \\
					\p_0(i) + \frac{2^{-9n}}{r} & \text{if } \bb_{\lfloor \frac{1}{2}(i-m) \rfloor}=1 \text{ and } i \text{ is odd} \end{cases}
$$
Note that all these distributions are $2$-SIIRVs as follows from Lemma \ref{lem:explicit-eps-ball}(ii)
since $$\dk(\p_\bb,\p_0) \leq \dtv(\p_\bb,\p_0) = 2^{-9n} \;.$$
For $0\le i \le r-1$, the sets $A_{i+1} = \{m+2i, m+2i+1\}$ define the partition of the domain.
We can now apply Assouad's lemma to this instance.

For $\bb \in \{-1,1 \}^r$ we can write
$$\sum_{x \in A_{\ell}} \left|\p_{\mathbf{b}^{(\ell, +)}}(x) - \p_{\mathbf{b}^{(\ell, -)}}(x) \right|  =  \frac{2 \cdot 2^{-9n}}{r} \;.$$
Similarly,
\begin{eqnarray*}
\sum_{i=0}^n \left( \sqrt{\p_{\mathbf{b}^{(\ell, +)}}(i)}-\sqrt{\p_{\mathbf{b}^{(\ell, -)}}(i)} \right)^2
& = & \sum_{i=m+2\ell, m+2\ell+1} \left( \frac{\p_{\mathbf{b}^{(\ell, +)}}(i)-\p_{\mathbf{b}^{(\ell, -)}}(i)}{\sqrt{\p_{\mathbf{b}^{(\ell, +)}}(i))}+\sqrt{\p_{\mathbf{b}^{(\ell, -)}}(i)}} \right)^2 \\
& = & \sum_{i=m+2\ell, m+2\ell+1} \left( \frac{2^{-9n} /r}{\sqrt{\p_{\mathbf{b}^{(\ell, +)}}(i))}+\sqrt{\p_{\mathbf{b}^{(\ell, -)}}(i)} } \right)^2 \\
& \geq & \sum_{i=m+2\ell, m+2\ell+1} \left( \frac{2^{-9n}/r}{2\sqrt{1/t}} \right)^2  \\
& = & \frac{2^{-18n} \cdot c}{2r}  \;,
\end{eqnarray*}
where the first inequality uses the fact that
$$\p_\bb(i) \geq \p_0(i) - \frac{2^{-9n}}{r} \geq \frac{2}{t}-\frac{1}{t} \geq \frac{1}{t},$$ for $m \leq i \leq m+2k.$				

Therefore, the parameters in Assouad's Lemma are
$$\alpha:=\frac{2 \cdot 2^{-9n}}{r}, \quad \gamma = \frac{2^{-18n} \cdot c}{2r}, \quad \textrm{and} \quad  s=\frac{1}{8\gamma}$$
from which we obtain that that there is a $\p_\bb$ with
$$\E \left[ \dtv(\mathbf{H}, \p_\bb) \right] \geq (r\alpha/4) \cdot (1-\sqrt{2s\gamma} )= \frac{2^{-9n}}{4}. $$
Hence, for $\epsilon=2^{-9n-2}$, if the number of samples satisfies
$$s \leq \frac{1}{8\gamma}= \frac{r \cdot 2^{18n}}{4c}=O( 2^{18n} \sqrt{n})= O\left((1/\eps^2) \sqrt{\log(1/\eps)}\right),$$
then $\E \left[ \dtv(\mathbf{H}, \p_\ba) \right] \geq \eps,$
completing the proof of the theorem.
\end{proof}

\subsection{A Useful Structural Result for $k$-SIIRVs} \label{ssec:geom-siirv} 

In this subsection, we prove the analogous structural result to Lemma~\ref{lem:explicit-eps-ball} for $k$-SIIRVs.

\begin{proposition} \label{prop:explicit-eps-ball}
Let $k\geq 2$ be a positive integer and $\eps \leq 1/\poly(k)$ be sufficiently small.
Let $n$ be a sufficiently small multiple of $\log(1/\eps).$
Define $\p$ to be the $k$-SIIRV given by $X \sim \p$ such that
$
X=\sum_{i=1}^n X_i \;,
$
where $X_i(j)=p_{i,j},$ and for $1\leq i \leq n,$ $1\leq j \leq k-2$
we have that $$p_{i,j} = 1/(3(k-2)n),  p_{i,0} = 1/3+(i-1)/(3n), p_{i,k-1}(k-1) = 1/3+(n-i)/(3n).$$
Then, if $\q$ is any distribution supported on $[n(k-1)]$ with $\dtv(\p,\q)\leq \epsilon$, then $\q$ is a $k$-SIIRV.
\end{proposition}
\begin{proof}
The basic idea of the proof will be topological. We note that the dimensionality of the parameter space of $n$-variate $k$-SIIRVs
is the same as the dimensionality of the space of random variables of appropriate support size. Our result will follow from the following lemma:

\begin{lemma} \label{lem:siirv-pgf}
Let $q_{i,j}$ $(1\leq i\leq n,0\leq j\leq k-1)$ be a sequence of positive real numbers
with $\sum_{j=0}^{k-1} q_{i,j}=1$ for each $i$.
Let $Y$ be the $k$-SIIRV defined by the $q_{i,j}$. Suppose that
$
\max_{i,j} (|p_{i,j}-q_{i,j}|) = \epsilon^{2/3}.
$
Then, $\dtv(X,Y) \geq \eps.$
\end{lemma}
\begin{proof}
Let $I$ and $J$ be one of the pairs of integers such that we achieve $|p_{I,J}-q_{I,J}|=\epsilon^{2/3}.$
$X$ has probability generating function 
$
\widetilde{X}(z) = \E[z^X] = \prod_{i=1}^n \widetilde{X_i}(z).
$
We start with the following claim:
\begin{claim}
Assuming $n$ is sufficiently large, the roots of $\widetilde{X_i}(z)=0$ satisfy
$$\left|z_\ell - e^{\pi i (1+2\ell)/(k-1)} \left(\frac{n+I}{2n-I} \right)^{1/(k-1)}\right| \leq O(1/(k-1)n) \;,$$
for $0\leq \ell \leq k-2.$ Also, $|z_\ell^j| \leq 3$ for all $1 \leq j \leq k-1.$
\end{claim}
\begin{proof}
Specifically we claim that when $n \geq 200$, there is a root within distance $33/(k-1)n.$

Consider the polynomial $f_I(x)= (1/3+(n-I)/(3n))x^{k-1} + (1/3 + I/(3n)).$
Then, $f_i(x)=0$ has roots $x=a_{\ell},$ where
$$a_{\ell} = e^{\pi i (1+2\ell)/(k-1)} \left(\frac{n+I}{2n-I} \right)^{1/(k-1)},$$
for $0\leq \ell \leq k-2.$
Note that $\widetilde{X_I}(x) = f_I(x)+ \sum_{j=1}^{k-2} x^j/(3(k-2)n) - 1/3n.$
Also, for any $1 \leq k \leq k-1,$ we have $\frac{1}{2} \leq |a_{\ell}^j| \leq 2.$

We will show that for any $y \in \C$ with  $|y-a_{\ell}|= 33/(k-1)n,$
it holds
$|\widetilde{X_I}(y)| \geq |\widetilde{X_I}(a_{\ell})|,$ and therefore
there is a root $z_{\ell}$  of $\widetilde{X_I}(x)$ with $|z_{\ell} - a_{\ell}| \leq  33/(k-1)n.$
We have
$$|\widetilde{X_I}(a_{\ell})| = |f_I(a_{\ell}) + \sum_{j=1}^{k-2} a_{\ell}/(3(k-2)n) - a_{\ell}/3n| \leq 2|a_{\ell}|/3n \leq 4/3n .$$
Now we consider $f_I(x)$ expressed as a polynomial in $w=x-a_{\ell}.$
We claim that this is dominated by the $w$ term when $|w| = 33/(k-1)n.$
We show that, under certain conditions, the binomial series is dominated by its first two terms:
\begin{claim} \label{clm:binomial-first}
If $m|x/b| \leq 1/3,$ then
$|(b+x)^m-b^m-(m-1)xb^{m-1}| \leq (m-1)|xb^{m-1}|/2.$
\end{claim}
\begin{proof}
By the binomial theorem $(b+x)^m = \sum_{j=0}^m {m \choose j}  x^j b^{m-j}.$
Note that the ratio of the absolute values of the $x^{j+1}$ and $x^j$ terms is
$$ \left|{m \choose j+1}  x^{j+1} b^{m-j-1}\right|/\left|{m \choose j}  x^j b^{m-j}\right|= (m-j)/(j+1) \cdot |x/b| \leq 1/3 \;.$$
Thus,
$$|(b+x)^m-b^m-(m-1)xb^{m-1}| = |\sum_{j=2}^m {m \choose j}  x^j b^{m-j}| \leq (m-1)|xb^{m-1}| \sum_{j=1}^{m-1} 3^{-j} \leq (m-1)|xb^{m-1}|/2.$$
\end{proof}
When $|w|=33/(k-1)n$, we have
$(k-1)|(w/a_I)| \leq 66/n \leq 1/3,$
and therefore $$f_I(w+a_I) =  (1/3+(n-I)/(3n))(w+a_I)^{k-1} + (1/3 + I/(3n))$$ satisfies
$$|f_I(w+a_I) - (1/3+(n-I)/(3n)) (k-1)wa_I^{k-1}| \leq (1/3+(n-I)/(3n))(k-1) |wa_I^{k-1}|/2.$$
Since $|(1/3+(n-I)/(3n))(k-1) |wa_I^{k-1}|/2  \geq 33/12n,$
and so $|f_I(w+a_I)| \geq 33/12n.$

Now we have that $|f_i(y)| \geq 33/12n$
and $f_I(a_{\ell})=0.$
We also have $$|(\widetilde{X_I}(y) -f_i(y))|= |\sum_{j=1}^{k-2} y^j/(3(k-2n) - 1/3n| \leq \sum_{j=1}^{k-2} (|a_{\ell}|+33/(k-1)n)^j/(3(k-2)n) + 1/3n.$$
By Claim \ref{clm:binomial-first} on $(|a_1|+33/(k-1)n)^j$, we have that
$$(|a_1|+33/(k-1)n)^j \leq |a_1|^j +3j|a_1|33/(k-1)2n \leq 2+99j/(k-1)n \leq 3.$$
So, $$1/3n + \sum_{j=1}^{k-2}(|a_{\ell}|+1/n)^j/(3(k-2)n)  \leq 1/n + 1/3n=4/3n.$$
We have
$$|\widetilde{X_I}(y)| \geq |f_i(y)| - |(\widetilde{X_I}(y) -f_i(y))| \geq (33-16)/12n > 4/3n \geq |\widetilde{X_I}(a_{\ell})|.$$
Since this holds for all $y\in \C$ with $|y-a_{\ell}|=33/(k-1)n,$
it follows that there is a $z_{\ell} \in \C$ with $|z_{\ell} - a_{\ell}| \leq 33/(k-1)n.$

Finally, note that since $|z_\ell- a_\ell| \leq 33/(k-1)n \leq 1/6(k-1)$,
for any $1 \leq j \leq k-1,$
we have $(j-1)(z_\ell-a_\ell)/a_\ell \leq 1/3$
and so by Claim~\ref{clm:binomial-first},
$$|z_\ell^j -a_\ell^j - (j-1)(z_\ell-a_\ell)a_\ell^{j-1}| \leq  |(j-1)(z_\ell-a_\ell)a_\ell^{j-1}|/2 \leq 1/6.$$ 
Thus, $|z_\ell^j| \leq |a_\ell^j|+1/2 \leq 3.$
\end{proof}

Our lemma will follow easily from the following claim:
\begin{claim}
For some $\ell,$ we have that
$
|\widetilde{X}(z_\ell)-\widetilde{Y}(z_\ell)| \geq \eps^{5/6}.
$
\end{claim}
\begin{proof}
Note that for each $i$, since $|z_\ell|^j \leq 3$ for all $1 \leq j \leq k-1$,
$$
|\widetilde{X_i}(z_\ell)-\widetilde{Y_i}(z_\ell)| \leq 3\eps^{2/3} \leq \epsilon^{1/2}/n \;.
$$
Furthermore, note that for $i\neq I$ that $|\widetilde{X_i}(z_\ell)| = \Theta(|i-I|/n).$ This implies that
$$
\prod_{i\neq I} \widetilde{Y_i}(z_\ell) = 2^{O(n)}.
$$
However, we have that
$$
\widetilde{X_I}(z_\ell) = 0
$$
for all $\ell.$

It suffices to show that $|\widetilde{Y_I}(z_\ell)|\geq \eps^{3/4}$ for some $\ell.$
Let $z_{k-1}=1$. By standard polynomial interpolation, we have that
$$
\widetilde{Y_I}(z) = \sum_{\ell=0}^{k-1} \widetilde{Y_I}(z_i)\left(\prod_{j\neq i} \frac{z-z_j}{z_i-z_j} \right).
$$
Similarly,
$$
\widetilde{X_I}(z) = \sum_{\ell=0}^{k-1} \widetilde{X_I}(z_\ell)\left(\prod_{j\neq \ell} \frac{z-z_j}{z_\ell-z_j} \right).
$$
In order to make use of this, we need to bound the size of the coefficients of the polynomial $\left(\prod_{j\neq \ell} \frac{z-z_j}{z_\ell-z_j} \right).$
\begin{claim}
For any $\ell,$ we have that all coefficients of
$
\left(\prod_{j\neq \ell} \frac{z-z_j}{z_\ell-z_j} \right)
$
are $O(1)$.
\end{claim}
\begin{proof}
Let
$$Q(z)=\tilde {X_i}(z) = \sum_{j=0}^{k-1} p_{i,j}z^j = \left(\frac{1+2(n-i)}{5} \right)\prod_{j=1}^{k-1} (z-z_j).$$
Firstly, for $\ell=k-1$ the polynomial in question is $Q(z)/Q(1)=Q(z)$, which clearly has coefficients of size $O(1)$. For $\ell < k-1$, the polynomial in question is
$$\frac{Q(z)(z-1)}{(z-z_\ell)(Q'(z_\ell))(z_\ell -1)}.$$
It should be noted that $(Q'(z_\ell))(z_\ell -1)=\Omega(1)$ and that multiplying a polynomial by $z-1$ 
at most doubles the size of its maximum coefficient. Therefore, it suffices to consider the polynomial $Q(z)/(z-z_\ell).$ 
In order to analyze this, we write $1/(z-z_{\ell})$ as a power series $P(z):=\sum_{m=0}^\infty -z^m/z_\ell^{m+1}.$ 
We note that the polynomial in question is the product of $Q(z)$ times this power series. 
We note that we need only consider the first $k$ terms of this product since terms of degree more than $k$ cancel. 
Noting that the first $k$ coefficients of $P(z)$ are all $O(1)$ and that the coefficients of $Q(z)$ have absolute values 
summing to $1$, implies that the first $k$ coefficients in their product are all $O(1)$. This completes the proof.
\end{proof}

Therefore, the largest coefficient of $\widetilde{X_I}(z)-\widetilde{Y_I}(z)$ is at most
$$
O(1) \sum_{{\ell}} \left| \widetilde{X_I}(z_\ell)-\widetilde{Y_I}(z_\ell)\right|.
$$
Recall that this largest coefficient is $\eps^{2/3}$ by assumption.
Therefore, for some $\ell$ we must have that
$$
\left| \widetilde{X_I}(z_\ell)-\widetilde{Y_I}(z_\ell)\right| \geq \Omega(\eps^{2/3}/k)  \geq \eps^{3/4}.
$$
On the other hand, we have that
$$
\widetilde{X_I}(z_{k-1})=\widetilde{Y_I}(z_{k-1})=1 \;,
$$
and so for some other $\ell$
we must have that $|\widetilde{Y_I}(z_\ell)|\geq \eps^{3/4}.$ Noting that
$$
\widetilde{X}(z_\ell) = 0 \;,
$$
and
$$
\widetilde{Y}(z_\ell) \geq 2^{O(n)} \eps^{3/4} \geq \eps^{5/6}.
$$
This proves the claim.
\end{proof}
The lemma now follows from the fact that
\begin{align*}
\left| \widetilde{X}(z_\ell)-\widetilde{Y}(z_\ell)\right| & = \sum_{m=0}^{n(k-1)} z^m |X(m)-Y(m)|\\
& \leq 2^{O(n)} \sum_{m=0}^{n(k-1)}  |X(m)-Y(m)|\\
& = 2^{O(n)} \dtv(X,Y).
\end{align*}
\end{proof}
Note that Lemma~\ref{lem:siirv-pgf} actually applies for any $Y$ and $Z$
with all parameters of $Z$ within $\eps^{2/3}$ of those of $X$,
and the parameters of $Y$ of distance $\delta$ from $Z,$
that $\dtv(Y,Z) \geq \eps^{1/3}\delta.$
This implies that the derivative of the map $F:\R^{n(k-1)}\rightarrow \R^{n(k-1)}$
from parameters of $(n,k)$-SIIRVs close to $X$ to probability distributions on $[n (k-1)]$
is everywhere injective (if $DF(Z)$ had some null-vector $v,$
then $F(Z+\delta v)$ would be $F(Z)+o(\delta)$, which is a contradiction).
Therefore, by the Inverse Function Theorem, $F$ is an open map.

Let $B_1$ be the set of parameters of $k$-SIIRVs within $\epsilon^{2/3}$ in $L^\infty$ of those of $X.$
Let $B_2$ be the set of distributions on $[n(k-1)]$ within $\epsilon$ of $X$.
Let $V=F(B_1)\cap B_2$. On the one hand since $B_1$ is compact, this must be a closed subset of $B_2.$
On the other hand, Lemma~\ref{lem:siirv-pgf} implies that $V=F(\textrm{Int}(B_1))\cap B_2,$
which is an open subset of $B_2,$ since $F$ is an open map.
Therefore, $V$ is both an open and closed subset of $B_2.$
Since $B_2$ is connected, this implies that $V=B_2.$
Thus, every element of $B_2$ is in the image of $F,$
and is thus a $k$-SIIRV, proving Proposition~\ref{prop:explicit-eps-ball}.
\end{proof}

\subsection{Sample complexity lower bound for $k$-SIIRVs} \label{ssec:sample-lb-siirv} 

In this subsection, we prove our general sample lower bound against $k$-SIIRVs:

\begin{theorem} [Sample Lower Bound for $k$-SIIRVs] \label{thm:sample-lower-bound-k-siirv}
Let $\mathcal{A}$ be any algorithm which, given as input $n$, $k \geq 2$, $\eps \leq 1/\poly(k),$
and sample access to an unknown $\p \in \mathcal{S}_{n, k}$ outputs a hypothesis distribution $\mathbf{H}$
such that $\E[\dtv(\mathbf{H}, \p)] \le \eps$. Then, $\mathcal{A}$ must use
$\Omega((k/\eps^2) \cdot \sqrt{\log(1/\eps)})$ samples.
\end{theorem}

In addition to the structural result of the previous subsection, 
we also need to prove an analogue of Lemma~\ref{lem:lc-stddev}, which does not immediately apply,
as $k$-SIIRVs need not be logconcave. In fact, we remark that Lemma~\ref{lem:lc-stddev} does not apply to the $k$-SIIRVs
used in the lower bound construction of Section~\ref{ssec:ksiirv}. So, we need to use a slightly different construction.
\begin{lemma}\label{clm:consecutive-example}
For the $k$-SIIRV $\p$ defined in Proposition \ref{prop:explicit-eps-ball},
there exist $\Omega ((k-1) \sqrt{n})$ consecutive integers
with probability mass under $\p$ at least $\Omega(\frac{1}{(k-1) \sqrt{n}}).$
\end{lemma}
\begin{proof}
We wish to reduce this claim to Lemma \ref{lem:lc-stddev},
which gives that there are universal constants $c>0$ such that for any PBD $\q$
with standard deviation $\sigma,$
there are at least $\Omega(\sigma)$ consecutive integers with probability mass at least  $ c \cdot \frac{1}{1+\sigma}.$

Recall that $\p$ is the $k$-SIIRV given by $X \sim \p$ such that
$
X=\sum_{i=1}^n X_i \;,
$
where $X_i(j)=p_{i,j}$ and for $1\leq i \leq n,$ $1\leq j \leq k-2,$
we have that  $p_{i,j} = 1/(3(k-2)n),$ $p_{i,0} = 1/3+(i-1)/(3n),$ $p_{i,k-1}(k-1) = 1/3+(n-i)/(3n).$
So, we have that $\Pr[X_i = 0 \vee X_i=k-1] = 1-1/3n$ for all $i.$

Let $A_0$ be the event that all $X_i$ are equal to $0$ or $k-1.$
Then, $\Pr[A_0]= (1-1/3n)^n = \Omega(1).$
Let $Y=X/(k-1)$ and $Y_i=X_i/(k-1).$
Conditioned on the event $A_0$, each $Y_i$ is a Bernoulli random variable
and $Y$ is a PBD $\q.$
Note that $\Var[Y \mid A_0] \geq n(1/3 \cdot 2/3)=2n/9=\Omega(n).$
So, by Lemma \ref{lem:lc-stddev}, we have that there are integers $a,b,$
with $a - b =\Omega(\sqrt{n})$ such that $\q(h) \geq \frac{3c}{\sqrt{2}(1+\sqrt{n})},$
for each integer $a \leq i \leq b.$
Since the probability of $A_0$ is $\Omega(1),$
it follows that any integer $h \in [(k-1)a,(k-1)b]$
with $h \equiv 0 \pmod{k-1}$ has $\Pr[X=h] \geq \Omega(\frac{1}{\sqrt{n}}) \geq \Omega(\frac{1}{(k-1) \sqrt{n}}).$

For a given $1 \leq i \leq n$, let $B_i$ be the event that only $X_i$
takes a value between $1$ and $k-2.$
Then, the conditional distribution of  $Y_{-i}=\sum_{j \neq i} Y_i$
under either $A_0$ or $B_i$ is a PBD $\q_{-i},$ which is the same in both cases.
Now, $Y=Y_{-i}+Y_{i}$ and conditional on $A_0,$
$Y_i$ is a Bernoulli for any integer $h,$
so either $\pr[Y_{-i}=h \mid A_0] \geq \pr[Y=h|A_0]/2,$
or $\pr[Y_{-i}=h-1 \mid A_0] \geq \pr[Y=h \mid A_0]/2.$
In particular, $\q(a) \geq \Omega(1/\sqrt{n})$
and $\q(b) \geq \Omega(1/\sqrt{n}),$
so it follows that either $\q_{-i}(a) \geq \Omega(1/\sqrt{n})$ or $\q_{-i}(a-1) \geq \Omega(1/\sqrt{n})$ and either $\q_{-i}(b) \geq \Omega(1/\sqrt{n})$ or $\q_{-i}(b-1) \geq \Omega(1/\sqrt{n}).$
However, as a PBD, $\q_{-j}$ is unimodal,
and it follows that for every integer $a \leq h \leq b-1$, $\q_{-i}(h) \geq \Omega(1/\sqrt{n}).$
Now, consider an integer $ (k-1)a < h < (k-1)b$
with $h \not\equiv 0 \pmod{k-1}$.
We can write $h=q(k-1)+r$ for integers $a \leq q \leq b-1$
and $1 \leq r \leq k-1$. Note that $\pr[X_i=r|B_i]=1/(k-2),$
since we are conditioning on it not taking the values $0$ or $k-1.$
Then $\pr[X=h \mid B_i]=\pr[Y_{-i}=q \mid B_i] \pr[X_i=r \mid B_i] =\Omega(1/\sqrt{n}) \cdot 1/(k-2) = \Omega(1/((k-1)\sqrt{n}).$

For each $1 \leq i \leq n$, $\pr[B_i] = (1-1/3n)^{n-1} \cdot 1/3n = \Omega(1/n)$.
So, consider any integer $ (k-1)a \leq h \leq (k-1)b.$
If $h \not\equiv 0 \pmod{k-1},$
$\pr[X=h] \geq \sum_{i=1}^n \Pr[X=h \wedge B_i] = \sum_{i=1}^n \Pr[X=h | B_i] \pr[A_i] = \sum_{i=1}^n \Omega(1/((k-1)\sqrt{n}) \cdot \Omega(1/n) = \Omega(1/((k-1)\sqrt{n})$. When $h \equiv 0 \pmod{k-1}$, we showed earlier using $A_0$ that $\Omega(\frac{1}{(k-1) \sqrt{n})})$. This holds for $(k-1)(a-b) \geq (k-1)(\Omega(\sqrt{n})-1)=\Omega((k-1)\sqrt{n})$ consecutive integers.
\end{proof}

The proof of Theorem \ref{thm:sample-lower-bound-k-siirv} using Assouad's Lemma 
is now almost identical to that of Theorem \ref{thm:sample-lower-bound}.

\begin{proof}[Proof of Theorem \ref{thm:sample-lower-bound-k-siirv}]
Let $\p$ be the $k$-SIIRV defined in \ref{prop:explicit-eps-ball}.
Let $C$ be a constant large enough that Proposition \ref{prop:explicit-eps-ball} implies that all distributions $\q$ with  $\dtv(\p,\q) \leq 2^{-Cn}$ are $k$-SIIRVs.

By Lemma~\ref{clm:consecutive-example}, there exists some $c>0$ and $r=\Omega((k-1)\sqrt{n})$ consecutive integers,
an integer $m$, $0 \leq m \leq n$, and a real value $t$ with $t \geq c \cdot r$,
such that for all $i$, with $m \leq i \leq m+2r$, we have
$$\p(i) \geq \frac{2}{t} \;.$$
For $n$ sufficiently large, we can assume that $2^{-Cn} \leq c$ and therefore
$\frac{1}{t} \geq \frac{2^{-Cn}}{r}.$

We are now ready to define our ``hypercube'' of $k$-SIIRVs.
For $\bb \in \{ -1, 1 \}^r$, consider the distribution $\p_\bb$ with
$$\p_\bb(i)=\begin{cases} \p_0(i) & \text{if } i < m \text{, } i > m+2r, \text{ or } \bb_{\lfloor \frac{1}{2}(i-m) \rfloor}=-1 \\
					\p_0(i) - \frac{2^{-Cn}}{r} & \text{if } \bb_{\lfloor \frac{1}{2}(i-m) \rfloor}=1 \text{ and } i \text{ is even} \\
					\p_0(i) + \frac{2^{-Cn}}{r} & \text{if } \bb_{\lfloor \frac{1}{2}(i-m) \rfloor}=1 \text{ and } i \text{ is odd} \end{cases}
$$
Note that Proposition \ref{prop:explicit-eps-ball} yields that all these distributions are $k$-SIIRVs
since $$\dk(\p_\bb,\p_0) \leq \dtv(\p_\bb,\p_0) = 2^{-Cn} \;.$$
For $0\le i \le r-1$, the sets $A_{i+1} = \{m+2i, m+2i+1\}$ define the partition of the domain.
We can now apply Assouad's lemma to this instance.

For $\bb \in \{-1,1 \}^r$ we can write
$$\sum_{x \in A_{\ell}} \left|\p_{\mathbf{b}^{(\ell, +)}}(x) - \p_{\mathbf{b}^{(\ell, -)}}(x) \right|  =  \frac{2 \cdot 2^{-Cn}}{r} \;.$$
Similarly,
\begin{eqnarray*}
\sum_{i=0}^n \left( \sqrt{\p_{\mathbf{b}^{(\ell, +)}}(i)}-\sqrt{\p_{\mathbf{b}^{(\ell, -)}}(i)} \right)^2
& = & \sum_{i=m+2\ell, m+2\ell+1} \left( \frac{\p_{\mathbf{b}^{(\ell, +)}}(i)-\p_{\mathbf{b}^{(\ell, -)}}(i)}{\sqrt{\p_{\mathbf{b}^{(\ell, +)}}(i))}+\sqrt{\p_{\mathbf{b}^{(\ell, -)}}(i)}} \right)^2 \\
& = & \sum_{i=m+2\ell, m+2\ell+1} \left( \frac{2^{-Cn} /r}{\sqrt{\p_{\mathbf{b}^{(\ell, +)}}(i))}+\sqrt{\p_{\mathbf{b}^{(\ell, -)}}(i)} } \right)^2 \\
& \geq & \sum_{i=m+2\ell, m+2\ell+1} \left( \frac{2^{-Cn}/r}{2\sqrt{1/t}} \right)^2  \\
& = & \frac{2^{-2Cn} \cdot c}{2r}  \;,
\end{eqnarray*}
where the first inequality uses the fact that
$$\p_\bb(i) \geq \p_0(i) - \frac{2^{-Cn}}{r} \geq \frac{2}{t}-\frac{1}{t} \geq \frac{1}{t},$$ for $m \leq i \leq m+2k.$				

Therefore, the parameters in Assouad's Lemma are
$$\alpha:=\frac{2 \cdot 2^{-Cn}}{r}, \quad \gamma = \frac{2^{-2Cn} \cdot c}{2r}, \quad \textrm{and} \quad  s=\frac{1}{8\gamma}$$
from which we obtain that that there is a $\p_\bb$ with
$$\E \left[ \dtv(\mathbf{H}, \p_\bb) \right] \geq (r\alpha/4) \cdot (1-\sqrt{2s\gamma} )= \frac{2^{-Cn}}{4}. $$
Hence, for $\epsilon=2^{-Cn-2}$, if the number of samples satisfies
$$s \leq \frac{1}{8\gamma}= \frac{r \cdot 2^{2Cn}}{4c}=O( 2^{2Cn}(k-1) \sqrt{n})= O\left((k/\eps^2) \sqrt{\log(1/\eps)}\right),$$
then $\E \left[ \dtv(\mathbf{H}, \p_\ba) \right] \geq \eps,$
completing the proof of the theorem.
\end{proof}

\bibliographystyle{alpha}

\bibliography{allrefs}

\appendix

\section*{Appendix}

\section{Basic Facts from Probability} \label{sec:probtools}

\begin{definition} \label{def:discretizedNormal}
Let $\mu \in \R$, $\sigma \in \R^{\geq 0}$.  We
let $\discnorm{\mu}{\sigma^2}$ denote the \emph{discretized normal}
distribution.  The definition of $Z \sim \discnorm{\mu}{\sigma^2}$
is that we first draw a normal $G \sim \normal(\mu, \sigma^2)$
and then we set $Z = \round{G}$; i.e., $G$ rounded to the nearest
integer.
\end{definition}

\noindent We begin by recalling some basic facts concerning total variation distance,
starting with the ``data processing inequality for total variation distance'':

\begin{proposition} [Data Processing Inequality for Total Variation Distance]
\label{prop:dpdtv}
Let $X,$ $X'$ be two random variables over a domain
$\Omega$. Fix any (possibly randomized) function $F$ on $\Omega$
(which may be viewed as a distribution over deterministic functions on
$\Omega$) and let $F(X)$ be the random variable such that a draw from
$F(X)$ is obtained by drawing independently $x$ from $X$
and $f$ from $F$ and then outputting $f(x)$ (likewise for $F(X')$).
Then we have
$\dtv(F(X), F(X'))  \leq \dtv(X, X').$
\end{proposition}

Next we recall the subadditivity of total variation distance
for independent random variables:

\begin{proposition} \label{prop:dtv-subadditive}
Let $A, A', B, B'$ be integer random variables such that $(A, A')$ is independent of $(B, B')$.  Then
$\dtv(A + B, A' + B') \leq \dtv(A, A') + \dtv(B, B')$.
\end{proposition}

We will use the following standard result which bounds
the variation distance between two normal distributions in terms of
their means and variances:
\begin{proposition} \label{prop:normal-dist}
Let $\mu_1,\mu_2 \in \R$ and $0 < \sigma_1 \le \sigma_2$.
Then $\dtv(\normal(\mu_1, \sigma_1^2), \normal(\mu_2, \sigma_2^2))
\leq {1 \over {2}} \left({|\mu_1 - \mu_2| \over \sigma_1} + {\sigma_2^2 -\sigma_1^2 \over \sigma_1^2}\right).$
\end{proposition}

\section{Lower Bounds on Matching Moments} \label{app:moments-lb}
We start by giving an explicit example of two PBDs over $k+1$
variables that agree exactly on the first $k$ moments and have total variation distance
$2^{-\Omega(k)}.$

\begin{proposition} \label{prop:moments}
Let $\p, \q \in \mathcal{S}_{k+1, 2}$ be PBD's with parameters
$p_i= (1+\cos (\frac{2\pi i}{k+1}))/2$ and $q_i=(1+\cos(\frac{2\pi i+\pi}{k+1}))/2$ respectively, where $1\leq i \leq k+1$.
Then $\p$ and $\q$ agree on their first $k$ moments and have $\dtv(\p,\q)\geq 4^{-k}$.
\end{proposition}
\begin{proof}
Let $X = \sum_{i=1}^{k+1} X_i$, where  $X_i$ are independent Bernoulli variables, and suppose that $X \sim \p$.
We note that, for $m\leq k$, the random variable $X^m$ can be expressed as a degree $m$ polynomial in the $X_i$'s.
Therefore, the $m$-th moment of $\p$ is a degree $m$ symmetric polynomial of the $p_i$'s.
Similarly, the $m$-th moment of $\q$ must be the same symmetric polynomial of the $q_i$. Therefore, to show that the first $k$ moments of $\p$ and $\q$ agree, it suffices to show that the first $k$ elementary symmetric polynomials in the $p_i$ have the same values as the corresponding polynomials of the $q_i$'s.

Note that the $p_i$ are the roots of $T_{k+1}(2x-1)-1$ and that the $q_i$ are the roots of $T_{k+1}(2x-1)+1$, where $T_{k+1}$ is the $(k+1)$-st Chebychev polynomial.
Therefore, for $m \leq k$, the $m$-th elementary symmetric polynomial in the $p_i$ is $[x^{k+1-m}] (-1)^m2^{-2k-1}T_{k+1}(2x+1)$
and the same holds for the $q_i$. Thus, the first $k$ moments of $\p$ and $\q$ agree.
To bound the total variation distance from below we observe that
$$\prod_{i=1}^{k+1} p_i = \p(k+1) = [x^0] (-1)^{k+1}2^{-2k-1}(T_{k+1}(2x+1)-1),$$
and $$\prod_{i=1}^{k+1} q_i = \q(k+1) = [x^0] (-1)^{k+1}2^{-2k-1}(T_{k+1}(2x+1)+1).$$
Therefore, the probability that $\p=k+1$ and the probability that $\q=k+1$ differ by $4^{-k}$.
This implies the appropriate bound in their variational distance and completes the proof.
\end{proof}

\noindent We also show that matching moments does not suffice for the case of $k$-SIIRVs, even for $k=3$:

\begin{proposition}  \label{prop:moments-3}
For $n$ an even integer,
there exist $\p,\q \in \mathcal{S}_{n/2,3}$ with disjoint supports such that their first $n-1$ moments agree.
\end{proposition}
\begin{proof}
We first show that there exist such $\p$ and $\q$ with $\p$ supported on even numbers and $\q$ supported on odd numbers,
so that
$$
\p(2j) = 2^{-n+1}\binom{n}{2j},
$$
and
$$
\q(2j+1) = 2^{-n+1}\binom{n}{2j+1}.
$$
We begin by showing that $\p\in \mathcal{S}_{n/2,3}$. Since $\sum_j 2^{-n+1}\binom{n}{2j} = 1$,
we will show that the polynomial $\widetilde{\p}(z) = \sum_j 2^{-n+1}\binom{n}{2j}z^{2j}$ factors as a product of $n/2$ quadratic polynomials with non-negative coefficients.
To prove this, we note that it suffices to show that all roots of $\widetilde{\p}$ are pure imaginary; then,
the natural factorization into quadratics using complex conjugate pairs will complete the argument.
For this, we observe that $\widetilde{\p}(z)=2^{-n}((1+z)^n+(1-z)^n)$. Therefore, $z$ is a root of $\widetilde{\p}$
only when $|1+z|=|1-z|$, or when $z$ is equidistant from $1$ and $-1$,
which happens only when the real part of $z$ is $0$, i.e., when $z$ is pure imaginary.

Similarly, we show that $\q\in \mathcal{S}_{n/2,3}$. Once again $\sum_j 2^{-n+1}\binom{n}{2j+1} = 1$, and so we merely need to show that
$\widetilde{\q}(z) = \sum_j 2^{-n+1}\binom{n}{2j+1}z^{2j+1}$ factors into quadratics with non-negative coefficients.
Since $\widetilde{\q}(z)=2^{-n}((1+z)^n-(1-z)^n)$, it also has only purely imaginary roots.

It remains to show that $\p$ and $\q$ have identical first $n-1$ moments.
For this, it suffices to show that $\widetilde{\p}(z)^{(k)}(1)=\widetilde{\q}(z)^{(k)}(1)$ for all $0\leq k <n$. Indeed,
we have that
$$
\widetilde{\p}(z)^{(k)}(1)-\widetilde{\q}(z)^{(k)}(1) = 2^{1-n}\frac{\partial^k}{\partial z^k}(1-z)^n|_{z=1} =
\frac{2^{1-n}(1-z)^{n-k}n!}{(n-k)!}|_{z=1} = 0.
$$
This completes the proof.
\end{proof}

\section{Omitted Proofs from Section~\ref{sec:learn}} \label{app:learn}

\subsection{Bootstrapping Our Sampler} \label{app:sampler}

The running time of the sampler described in Section~\ref{ssec:sampler}
has an $O(\log n)$ dependence. In this subsection, we show that the dependence on $n$
can be easily removed, by dealing separately with the case that the variance is $\Omega(\poly(k/\eps)).$
In particular, we have the following algorithm, 
which is similar to the {\tt Learn-Heavy} routine from \cite{DDOST13focs}.

\begin{lemma} \label{lem:learn-heavy}
There is an algorithm with the following performance guarantee: 
For any $\eps > 0$ and $X\in \Scal_{n,k}$ with  $\Var[X]=\Omega(\poly(k/\eps))$, 
the algorithm draws $O(k/\eps^2)$ samples from $X,$ 
runs in $\widetilde{O}(k^2/\eps^2)$ time, and with high constant probability outputs a distribution $cZ+Y,$ 
where $ 1 \leq c \leq k,$ $Z$ is a discrete Gaussian, and $Y$ is a $c$-IRV, 
with $\dtv(X, cZ+Y) \leq \eps.$ 
\end{lemma}
\begin{proof}

By Theorem \ref{thm:reg}, there is a $1 \leq c' \leq k$ such that the discrete Gaussian $Z'$ with parameters $\E[X]/c'$ and $\Var[X]/c'^2$ and the $c'$-IIRV $Y':=X \pmod{c'}$ satisfy $\dtv(X,c'Z'+Y') \leq \eps$. 

We start by guessing $c.$ 
For each guess for $c,$ 
we learn the appropriate $Y$ and $Z.$ 
Finally, we run a tournament over the possible values of $c.$ 
Fix $1 \leq c \leq k.$ 
To learn $Y$, we first draw $\Theta(c/\eps^2)$ samples and let $X'$ be the resulting empirical distribution. 
Then, we take $Y= X' \pmod{c}.$ 
To learn $Z,$ we take $\Theta(1/\eps^2)$ samples from $X$ and calculate the empirical mean and variance, 
$\widetilde{\mu}$ and $\widetilde{\sigma}^2.$ 
Then, we let $Z$ be the distribution obtained by sampling from $\mathcal{N}(\widetilde{\mu}/c, \widetilde{\sigma}^2/c^2)$ 
and rounding the sample to the nearest integer. 

Suppose that $c=c'.$ 
By standard facts, we have $\dtv(Y,Y')=\dtv(X' \pmod{c},X \pmod{c}) \leq \eps/4$ with high probability. 
Also, with high probability, we have  
$(1-\eps/4)\widetilde{\sigma}^2 \leq \Var[X] \leq (1+\eps/4)\widetilde{\sigma}^2$ and $|\E[X]-\widetilde{\mu}| \leq \widetilde{\sigma}\eps/4$.
By a combination of Propositions~\ref{prop:dpdtv} and~\ref{prop:normal-dist}, we have that 
$\dtv(Z,Z') \leq \frac{1}{2}\left( \frac{|\E[Z]-\E[Z']|}{\sqrt{\var[Z]}} + \frac{|\var[Z]-\var[Z']|}{\var[Z]}\right) \leq \eps/4.$
Thus, we have $\dtv(Y+cZ,Y'+cZ') \leq \dtv(Y,Y') + \dtv(Z,Z') \leq \eps/2,$ 
and therefore $\dtv(X,Y+cZ) \leq \dtv(X,Y'+cZ')+ \dtv(Y+cZ,Y'+cZ') \leq \eps.$

In summary, we have $k$ different hypothesis distributions $Y_c+cZ_c,$ for each $1 \leq c \leq k,$ 
one of which is promised to satisfy $\dtv(X,Y_c+cZ_c) \leq \eps$.
We can now run a standard tournament procedure~\cite{DL:01, DDS15-journal} 
that produces a hypothesis with $\dtv(X,Y_c+cZ_c) \leq O(\eps)$ with high probability.
This requires $O(\log k/\eps^2)$ samples and can be easily done in $\widetilde{O}(k^2 /\eps^2)$ time.
\end{proof}

We thus obtain the following corollary:

\begin{corollary} For all $n, k \in \Z_+$ and $\eps>0,$
there is an algorithm with the following
performance guarantee: Let $X \in \Scal_{n,k}$ be an unknown $k$-SIIRV.
The algorithm uses $O(k\log^{2}(k/\eps)/\eps^2)$
samples from $\p,$  runs in time $\widetilde{O}(k^3/\eps^2),$
and with probability at least $9/10$ outputs an $\eps$-sampler for $X.$
This $\eps$-sampler produces a single sample in time $\widetilde{O}(k).$
\end{corollary}
\begin{proof}
First we take $O(1)$ samples and estimate the variance of $X.$ 
If the variance is $\Omega(\poly(k/\eps))$, we use the algorithm given by Lemma \ref{lem:learn-heavy} to output a distribution
 $cZ+Y,$ where $ 1 \leq c \leq k$, $Z$ is a discrete Gaussian and $Y$ is a $c$-IRV, with $\dtv(X, cZ+Y) \leq \eps.$
 Note that $cZ+Y$ can be sampled in time $O(k).$
 
If the variance is $O(\poly(k/\eps))$, we use {\tt Learn-SIIRV}. 
This produces a distribution $\h$ given by its DFT modulo $M=O(\poly(k/\eps))$ 
at $O(k \log (k/\eps))$ points. 
By Theorem \ref{thm:sampler}, we can compute an $\eps$-sampler 
which produces a single sample in time 
$$O(\log(M) \log(M/\epsilon) \cdot |S|)=O(\log^2(k/\eps) \cdot k \log(k/\eps)).$$
\end{proof}

\subsection{A Bound on the $1/2$-norm of $k$-SIIRVs} \label{app:norm}

 \begin{lemma} \label{lem:half-norm}
 The $1/2$-norm of a $k$-SIIRV $\p$ with variance $\sigma^2$ is $O(\sigma+k).$ 
 \end{lemma}
\begin{proof}
Recall that $\|\p\|_{1/2} = (\sum_i \sqrt{\p(i)})^2.$ 
Let $\mu$ be the mean of $X \sim \p.$ 
By Cauchy-Schwartz, for any $S \subseteq [kn]$, we have $\sum_{i \in S} \sqrt{\p(i)} \leq \sqrt{ \p(S) \cdot |S|}$.
By Bernstein's inequality, for any $\eps>0$, it holds $\Pr[|X-\mu| >  (k+\sigma) \log(1/\eps)] \leq \eps$. Therefore, we can write
\begin{align*}
\sum_i \sqrt{\p(i)} = &  \sum_{|\mu-i| \leq \sigma + k} \sqrt{\p(i)} + \sum_{m=0}^\infty \sum_{(\sigma + k) 2^m < |\mu-i| \leq 2^{m+1} (\sigma + k)} \sqrt{\p(i)} \\
\leq & \sqrt{\sigma + k} + \sum_{m=0}^\infty 2 \sqrt{\sigma+k} \cdot 2^{m/2} \sqrt{\Pr[|X-\mu| > (\sigma + k) 2^m]} \\
\leq & \sqrt{\sigma + k} + \sum_{m=0}^\infty 2 \sqrt{\sigma+k} \cdot 2^{m/2 - 2^{m/2}} 
=  O(\sqrt{\sigma + k}) \;.
\end{align*}
\end{proof}

\section{Omitted Proofs from Section~\ref{sec:cover-ub}} \label{app:upper}

\subsection{Proof of Lemma~\ref{lem:simple}.} \label{app:simple}
For convenience, we restate Lemma~\ref{lem:simple}:

\medskip

\noindent {\bf Lemma~\ref{lem:simple}.}
{\em Let $\p \in {\cal S}_{n,k}$ be a $k$-SIIRV with $\Var[X]=V$. For any $0< \delta < 1/4$,
there exists $\q\in {\cal S}_{n,k}$ with $\dtv(\p,\q) = O(\delta V)$
such  that all but $O(k+V/\delta)$ of the $k$-IRV's defining $\q$ are constant.
}

\smallskip

\begin{proof}
For a $k$-IRV $A$ let $m(A)$ be an index $i$
so that $\pr[A=i]$ is maximized. Let $d(A) = \pr[A \neq m(A)]$ be the probability $A$ assigns to values in $[k] \setminus \{i\}$.
Suppose that $d(A) \le 1/2.$ Then we have that
$$
 d(A)/2 \leq  (1/2) \cdot \pr(A\neq A')  \leq  (1/2) \cdot \E[|A-A'|^2]  =   \var[A] \leq \E[|A-m(A)|^2] \leq k^2 \cdot d(A),
$$
where $A'$ is an independent copy of $A$. The leftmost inequality follows from our assumption that $d(A) \le 1/2.$
The proof of the lemma will make repeated applications of the following claim:
\begin{claim} \label{claim:iter}
Let $A,B$ be independent $k$-IRV's with $m(A)=m(B)$ and $d(A)+d(B)\leq 1/2$.
Then there exist independent $k$-IRV's $C$ and $D$, where $D$ is a constant, $d(C)= d(A)+d(B)$, and
$\dtv(A+B,C+D)=O(d(A)d(B))$.
\end{claim}
\begin{proof}
Let $m(A)=m(B)=i$. Let $d(A)=\delta_1,d(B)=\delta_2$.
Let $A'$ be the random variable $A$ conditioned on $A$ not equaling $i$,
and $B'$ be the random variable $B$ conditioned on it not equaling $i$.
Note that $A$ is a mixture of $i$ and $A'$ and $B$ a mixture of $i$ and $B'$.
Furthermore $A+B$ equals $2i$ with probability $(1-\delta_1)(1-\delta_2)$, $i+A'$ with probability $\delta_1(1-\delta_2)$,
$i+B'$ with probability $(1-\delta_1)\delta_2$ and $A'+B'$ with probability $\delta_1\delta_2$.

Let $D$ be the random variable that is deterministically $i$
and $C$ be the random variable that equals $i$ with probability $1-\delta_1-\delta_2$,
$A'$ with probability $\delta_1$, and $B'$ with probability $\delta_2$.
Then $C+D$ equals $2i$, $i+A'$, $i+B'$ and $A'+B'$ with probabilities $1-\delta_1-\delta_2$, $\delta_1$, $\delta_2$, and $0$.
These probabilities are within an additive $\delta_1\delta_2$ of the corresponding probabilities for $A+B$ and therefore 
$\dtv(A+B,C+D)=O(\delta_1\delta_2)$. Note that $C=i$ with probability $1-\delta_1-\delta_2$, so $d(C)=\delta_1+\delta_2$, which completes the proof.
\end{proof}

For a random variable $X \sim \p$, we have that $X = \sum_{i=1}^n A_i$ where the $A_i$'s are independent $k$-IRV's.
We iteratively modify $\p$ as follows: If two of the non-constant component $k$-IRV's of $\p$
are $A$ and $B$, with $m(A)=m(B)$ and $d(A),d(B)<\delta$,  then we replace the pair $A$ and $B$ with the pair $C$ and $D$ as described by the above claim.
Notice that every step reduces the number of non-constant component variables, and therefore this process terminates, giving
a $k$-SIIRV $\q$ with for $Y \sim \q$, $Y= \sum_{i=1}^n B_i$.

By construction, for each $1\leq i \leq k$, $\q$ has at most one non-constant component variable with $m(B_j)=i$ and $d(B_j)< \delta$.
Claim~\ref{claim:iter} implies
the sum of the $d$'s of the component variables does not increase in any iteration,
and therefore
$$\sum_{j=1}^n d(B_j) \leq \sum_{j=1}^n d(A_j) \leq 2 \sum_{j=1}^n \var[A_j] = 2\var[X] = 2V \;,$$
where the second inequality uses the aforementioned lower bound on the variance of a $k$-IRV.
Hence, the number of non-constant component variables in $\q$ is at most $k+2V\delta^{-1}$.

It remains to show that $\dtv(\p,\q)  = O(\delta V).$
Let $A,B$ and $C,D$ be the $k$-IRV's of Claim~\ref{claim:iter}.  Then $\dtv(A+B,C+D)=O(d(A)d(B)) = O([d(C)^2+d(D)^2] - [d(A)^2+d(B)^2]).$
That is, the total variation distance error introduced by replacing $A, B$ by $C, D$
is at most a constant times the amount that the sum of the squares of the $d$'s of the component variables increases by.
Repeated application of this observation combined with the sub-additivity of total variation distance gives
$\dtv(\p,\q) = O\left(\sum_{j=1}^n d(B_j)^2 - \sum_{j=1}^n d(A_j)^2 \right).$ On the other hand, note that all of the $B_j$'s
 that are not also $A_j$ satisfy $d(B_j)\leq 2 \delta$. Therefore, we have that
$ \dtv(\p,\q) \leq O\left(\sum_{j:d(B_j)\leq 2\delta} d(B_j)^2\right) = O\left( \delta \sum_j d(B_j) \right) = O(\delta V) \;,$
which completes the proof.
\end{proof}

\subsection{Proof of Lemma~\ref{lem:close-roots-ksiirvs}. } \label{ap:close-roots}
For convenience, we restate Lemma~\ref{lem:close-roots-ksiirvs}:

\medskip

\noindent {\bf Lemma~\ref{lem:close-roots-ksiirvs}.}
{\em Fix $x \in \mathbb{C}$ with $|x|=1$.
Suppose that $\rho_1,\ldots,\rho_m$ are roots of $\widetilde{\p}(x)$ (listed with appropriate multiplicity)
which have $|\rho_i-x| \leq \frac{1}{2k}$. Then, we have the following:
\begin{itemize}
\item[(i)] $|\widetilde{\p}(x)| \le 2^{-m} \;.$
\item[(ii)] For the polynomial $q(x) = \widetilde{\p}(x) / \prod_{i=1}^m (x - \rho_i)$,  we have that $|q(x)| \leq k^m$.
\end{itemize}
}

\medskip

To prove our lemma, we will make essential use of the following simple lemma:

\begin{lemma}\label{lem:close-roots}
For any polynomial $p(x) \in\C[x]$ of degree $d$
where the sum of the absolute values of the coefficients of $p$ is at most $1$, we have the following:
Fix $z\in\C$ with $|z|=1$. Suppose that $p$ has roots $\rho_1,\ldots,\rho_m$ with $|\rho_i - z| \leq \frac{1}{2d}$, for $i \in \{1,\ldots,m\}$.
Then, the following hold:
\begin{itemize}
\item[(i)] $\left| p(z) \right| \leq 2^{-m}$,
\item[(ii)] for the polynomial  $q(x) = p(x)/ \prod_{i=1}^m (x-\rho_i)$ we have that $|q(z)| \leq d^m.$
\end{itemize}
\end{lemma}
\begin{proof}
The lemma is proved by repeated applications of the following claim:
\begin{claim}\label{claim:poly-size}
Let $p(x)\in\C[x]$ be a degree-$d$ polynomial such that the sum of the absolute values of the coefficients of $p$ is at most $1$.
Let $\rho$ be a root of $p(x)$ and $q(x)$ be the polynomial $\frac{p(x)}{x-\rho}$. Then, the sum of the absolute values
of the coefficients of $q$ is at most $d$.
\end{claim}
\begin{proof}
We write the coefficients of $p(x)$ and $q(x)$ as $p(x) = \sum_{i=0}^d p_i x^i$ and $q(x)=\sum_{i=0}^{d-1} q_i x^i$.
Since $p(x) = (x-\rho)q(x)$, for $1 \leq i \leq d-1$, we have
\begin{equation} \label{eq:obvious} p_i = q_{i-1} - \rho q_i \;,\end{equation}
and similarly $p_d = q_{d-1}$, $p_0 = -\rho q_0$.

We consider two cases based on the magnitude of $\rho$.
First, suppose that $|\rho| \leq 1$. Since $q_{d-1} = p_d$ and,  by (\ref{eq:obvious}), $q_{i-1} = p_i + \rho q_i$, for $1 \leq i \leq d-1$,
an easy induction gives that $q_i = \sum_{j=i+1}^d p_j \rho^{j-i-1}$ for $0 \leq i \leq d-1$. Summing and taking absolute values gives:
\begin{eqnarray*}
\sum_{i=0}^{d-1}  |q_i| 	&\leq& \sum_{i=0}^{d-1}  \sum_{j=i+1}^d |p_j| |\rho|^{j-i-1}
				            =  \sum_{i=1}^ d (|p_i| \sum_{j=0}^{i-1} |\rho|^j) \\
				            &\leq& \sum_{i=1}^d |p_i| i
				             \leq d \sum_{i=1}^d |p_i|
					    \leq  d \;.
\end{eqnarray*}
Second, suppose $|\rho| > 1$.
Then, $\frac{1}{|\rho|} < 1$.
We have $q_0 = -\frac{1}{\rho}p_0$ and by (\ref{eq:obvious}),
for $1 \leq i \leq d-1$, $q_i = \frac{1}{\rho}(q_{i-1}-p_i)$.
By an easy induction, for $0 \leq i \leq d$, $q_i = -\sum_{j=0}^{i} p_j \frac{1}{\rho^{i-j}}$.  Summing and taking absolute values gives:
\begin{eqnarray*}
\sum_{i=0}^{d-1}  |q_i| 	& \leq & \sum_{i=0}^{d-1}  \sum_{j=0}^{i} |p_j| \frac{1}{|\rho|^{i-j}}
							=  \sum_{i=0}^{d-1} (|p_i| \sum_{j=i}^{d-1} \frac{1}{|\rho|^{d-1-i}}) \\
							& \leq & \sum_{i=0}^{d-1} |p_i| (d-1-i)
							\leq  d \sum_{i=0}^{d-1} |p_i|
						          \leq  d \;.
\end{eqnarray*}
\end{proof}
By repeated applications of the claim it follows that the polynomial
$q(x)$
has the sum of the absolute values of its coefficients at most $d^m$.
Since $|z|=1$, it follows that $|q(z)| \leq d^m$ which gives (ii). To show (i) we note that
$$|p(z)|  =  |q(z)| \cdot {\prod_{i=1}^m |z-\rho_i}| \leq |q(z)| \cdot (1/2d)^m \leq 2^{-m} \; .$$
This completes the proof of Lemma~\ref{lem:close-roots}.
\end{proof}

\begin{proof}[Proof of Lemma~\ref{lem:close-roots-ksiirvs}.]
Note that $\widetilde{\p}(x)$ is the degree $n(k-1)$ polynomial defined by
$\widetilde{\p}(x) = \sum_{i=0}^{n(k-1)} \p(i) x^i.$  Note that the sum of the absolute values of $\widetilde{\p}$'s coefficients is $1$.
However, to apply Lemma~\ref{lem:close-roots} directly to $\widetilde{\p}$ we would need the roots to be at distance at most $\frac{1}{2n(k-1)}.$

Note that $\widetilde{\p}(x)$ factors as $\prod_{i=1}^n p_i(x)$, where $p_i(x) =\E[x^{X_i}]$ is a degree $k-1$
polynomial that is determined by the $i$-th $k$-IRV. It is clear that the coefficients of $p_i(x)$ are non-negative and sum to $1$,
hence we may apply Lemma~\ref{lem:close-roots} to $p_i(x)$.  Suppose that $p_i(x)$ has $m_i$ roots with $|\rho_i-x| \leq \frac{1}{2k}$.
Lemma~\ref{lem:close-roots}(i) implies that $|p_i(x)| \leq 2^{-m_i}$. Since $\widetilde{\p}(x) = \prod_{i=1}^n p_i(x)$, this
yields part (i) of Lemma~\ref{lem:close-roots-ksiirvs}.

Lemma~\ref{lem:close-roots}(ii) implies that the polynomial $q_i(x) = p_i(x) / \prod_{j \in S_i}  (x - {\rho_j})$,
for $S_i \subseteq \{1, \ldots, m\}$ with $|S_i| = m_i$, satisfies $|q_i(x) | \le k^{m_i}$.
Note that $q(x)=\prod_{i=1}^n q_i(x)$.
Therefore, $|q(x)| \le \prod_i k^{m_i}  = k^m$, giving part (ii) of Lemma~\ref{lem:close-roots-ksiirvs}.
\end{proof}

\subsection{Proper Cover Construction for the High Variance Case.} \label{ap:cover-high-var}

Exhausting over the $k-1$ possible values of $c$,
we can assume that $c$ is known to the algorithm.
Before proceeding further, we will need further structural information about the $k$-SIIRVs in this case.
We start with the following simple lemma giving an upper bound on the total variation distance
between two high variance $k$-SIIRVs:

\begin{lemma} \label{lem:dist-non-sparse}
For $\eps > 0$, let $X$, $X'$ be $k$-SIIRVs
with $\Var[X], \Var[X'] \geq \poly(k/\eps)$ for a sufficiently large $\poly(k/\eps)$
that have $\dtv(X,Y+cZ) \leq \eps$ and $\dtv(X',Y'+cZ') \leq \eps$ for $c$-IRVs $Y$,$Y'$
and discrete Gaussians $Z$,$Z'$, with $\E[X]=c\E[Z]$, $\var[X]=c^2\var[Z]$, $\E[X']=c\E[Z']$ and $\var[X']=c^2\var[Z']$.
Then we have that
$$
\dtv(X,X') \leq 4\eps + \dtv\left(X \pmod{c}, X' \pmod{c}\right)+\frac{1}{2}\frac{|\E[X]-\E[X']|}{\sqrt{\var[X]}} + \frac{1}{2} \frac{|\var[X]-\var[X']|}{\var[X]}
$$
where $X \pmod{c}$ is the $c$-IRV with $\Pr[X \pmod{c}=i]=\Pr[X \equiv i \pmod{c}]$ for $i \in [c]$.
\end{lemma}
\begin{proof}
Using Proposition \ref{prop:dpdtv},
since $\dtv(X,Y+cZ) \leq \eps$ with $Y \equiv Y+cZ \pmod{c}$,
we have $\dtv\left(X \pmod{c},Y\right) \leq \eps$. Similarly, $\dtv\left(X' \pmod{c},Y'\right) \leq \eps$.
By a combination of Propositions~\ref{prop:dpdtv} and~\ref{prop:normal-dist}, we have that $\dtv(Z,Z') \leq \frac{1}{2}\left( \frac{|\E[Z]-\E[Z']|}{\sqrt{\var[Z]}} + \frac{|\var[Z]-\var[Z']|}{\var[Z]}\right)$. Since $\E[X]=c\E[Z]$, $\var[X]=c^2\var[Z]$, $\E[X']=c\E[Z']$ and $\var[X']=c^2\var[Z']$ it follows that
$$\frac{|\E[Z]-\E[Z']|}{\sqrt{\var[Z]}} + \frac{|\var[Z]-\var[Z']|}{\var[Z]} = \frac{|\E[X]-\E[X']|}{\sqrt{\var[X]}} + \frac{|\var[X]-\var[X']|}{\var[X]}.$$ Therefore,
\begin{eqnarray*}
\dtv(Y+cZ,Y'+cZ') &\le& \dtv(Y,Y')+\dtv(Z,Z') \\
&\le& 2\eps + \dtv\left(X\pmod{c},X'\pmod{c}\right)+\\
&+&\frac{1}{2}\left( \frac{|\E[X]-\E[X']|}{\sqrt{\var[X]}} + \frac{|\var[X]-\var[X']|}{\var[X]}\right).
\end{eqnarray*}
By another application of the triangle inequality, we have that $\dtv(X,X') \leq \dtv(X,Y+cZ)+ \dtv(Y+cZ,Y'+cZ')+ \dtv(Y'+cZ',X') \leq 2\eps + \dtv(Y+cZ,Y'+cZ)$, which completes the proof.
\end{proof}

To use the above lemma, we need a way to characterize the constant $c$ in the statement of Theorem \ref{thm:reg}, namely
to show that the theorem applies to both $X$ and $X'$ for {\em the same value} of $c$. For a $k$-IRV $A$,
let $m(A)$ be an index $i$ so that $\pr[A=i]$ is maximized.
The following result is implicit in the proof of Theorem~ \ref{thm:reg} in \cite{DDOST13focs} (in particular, in Theorem 4.3 of that paper):
\begin{lemma}[\cite{DDOST13focs}] \label{lem:explicit-c}
Given a $k$-SIIRV $X = \sum_{i=1}^n X_i$ with $\var[X] \geq \poly(k/\eps)$,
let $\mathcal{H}$ be the set of integers $b$ such that $\sum_{i=1}^n \Pr[X_i-m(X_i)=c] \geq \Theta(k^7/\eps^2)$
and $c=\mathrm{gcd}(\mathcal{H})$.
Then there is a $c$-IRV $Y$ and a discrete Gaussian $Z$ with $\dtv(X,Y+cZ) \leq \eps$.
\end{lemma}

Let $X \in \mathcal{S}_{n, k}$ be a $k$-SIIRV with $\Var[X] \geq \poly(k/\eps)$ as in Case 2 of Theorem~ \ref{thm:reg}.
Our main claim is that, up to $\epsilon$ error in total variation distance,
we can assume that $X$ has a special structure. In particular,
we can take all but one of the component IRVs of $X$ to be constant modulo $c$, with the last one being a $c$-IRV.
More formally, we
claim that there is a $k$-SIIRV  $X'$ with $\dtv(X,X') \leq \eps$, such that $X'=\sum_{i=1}^n X_i'$ with
\begin{itemize}
\item {For $1 \leq i \leq H$, where $H=\Theta(k^7/\eps^2)$,} $X_i'$ is either $0$ or $c$ each with equal probability.
\item {For $1 \leq i \leq n-1$,} $X_i'$ is constant modulo $c$.
\item $X_n' $ is a $c$-IRV.
\end{itemize}
where $c$ is as in Lemma \ref{lem:explicit-c}.

We can construct such an $X'$ from $X$ as follows.
For $1 \leq i \leq H$, we replace $X_i$ with the $X_i'$ above that is $0$ or $c$ with equal probability.
For $H+1 \leq i \leq n-1$, we replace each $X_i$ by $X_i$ conditioned on the event that $X_i \pmod{c}=m(X_i) \pmod{c}$.
Finally we take  $X'_n$ to be $(X-\sum_{i=1}^{n-1} X_i')  \pmod{c}$ noting that $\sum_{i=1}^{n-1} X_i'  \pmod{c}$ is a constant.

We now show that the above procedure only changes the expectation and variance
by $|\E[X] - \E[X']| \leq \poly(k/\eps)$ and $|\var[X]-\var[X']|\leq \poly(k/\eps)$.
Note that
for two arbitrary $k$-IRVs, $A$ and $B$, we have that $|\E[A] - \E[B]| \leq k$ and
$|\var[A]-\var[B]|\leq k^2$. Thus,
$$|\E[X_n + \littlesum_{i=1}^H X_i] - \E[X'_n + \littlesum_{i=1}^H X'_i]| \leq (H+1)k \leq \poly(k/\eps)$$
and $$|\var[X_n + \littlesum_{i=1}^H X_i]-\var[X'_n + \littlesum_{i=1}^H X'_i]|\leq (H+1)k^2 \leq \poly(k/\eps).$$
For the remaining variables $H+1 \leq i \leq n-1$, we have $\dtv(X_i,X'_i) \leq \Pr[X_i - m(X_i) \not\equiv 0 \pmod{c}]$ and so  $|\E[X_i] -\E[X'_i]| \leq k \Pr[X_i - m(X_i) \not\equiv 0 \pmod{c}]$ and  $|\var[X_i] -\var[X'_i]| \leq k^2 \Pr[X_i - m(X_i) \not\equiv 0 \pmod{c}]$.
For each integer $0 \leq b \leq k-1$ that does not divide $c$, by Lemma \ref{lem:explicit-c},
we must have that $b \notin \mathcal{H}$ and hence
$\sum_{i=1}^n \pr[X_i-m(X_i)=b] = O(k^7/\eps^2)$. Thus,
$\sum_{i=1}^n  \Pr[X_i - m(X_i) \not\equiv 0 \pmod{c}] = O(k^8/\eps^2)$.

If $\var[X]$ is a sufficiently large $\poly(k/\eps)$,
then $\var[X']$ is large enough that we can apply
Theorem \ref{thm:reg} and Lemma \ref{lem:explicit-c} to $X'$.
Note that $\sum_{i=1}^n \Pr[|X'_i - m(X'_i)|=c] \geq \sum_{i=1}^H \Pr[|X'_i - m(X'_i)|=c]=H/2$.
We thus have that either $c \in \mathcal{H}$ or $-c \in \mathcal{H}$.
Since for $b$ that does not divide $c$, we have $\sum_{i=1}^n \Pr[X'_i - m(X'_i)=b] = \Pr[X'_n - m(X'_n)=b] \leq 1$ and thus $b \notin \mathcal(H)$, we have that $\mathrm{gcd}(\mathcal{H})=c$.
Thus, for $X$ with sufficiently large $\poly(k/\eps)$ variance, we have that $\dtv(X,Y+cZ) \leq \eps/10$ and $\dtv(X',Y'+cZ') \leq \eps/10$ for the same $1 \leq c \leq k-1$ and $c$-IRVs $Y,Y'$ and discrete Gaussians $Z,Z'$. In conclusion, we can apply Lemma \ref{lem:dist-non-sparse} to $X$ and $X'$.
We have that $X' \pmod{c}=X'_n= X \pmod{c}$. We have shown that $\E[X]-\E[X'] \leq \poly(k/\eps)$ and $\var[X]-\var[X'] \leq \poly(1/\eps)$.
If $\var[X]$ is a sufficiently large $\poly(k/\eps)$ then we can make the contributions of each of these to $\dtv(X,X')$ in Lemma \ref{lem:dist-non-sparse} smaller than $\eps/10$. Then we have $\dtv(X,X') \leq \eps$.

Since every $k$-SIIRV $X$ in Case 2  is $\eps$-close to an $X'$ of the aforementioned form,
to compute a proper cover for this case, we can consider only $k$-SIIRVs of the form stated above.
By a similar argument as above, our cover only needs to ensure that the triple of $X\pmod{c},\E[X],\var[X]$
is sufficiently close to any such triple achievable by an element of $\mathcal{S}_{n,k}$ of this form.
Obtaining a cover of $X\pmod{c}$ is easy, as we only need to deal with the single term {$X_n$}
that is non-constant modulo $c$, and produce a cover for $c$-IRVs.
Indeed, it is straightforward to produce such a cover of size $O(k/\epsilon)^k$.

As explained in Section~\ref{ssec:reduce}, we have an explicit cover for the discrete Gaussian random variables
that can appear in this setting.
However, we are left with the difficulty of producing an explicit $k$-SIIRV approximating one of these $c$ times a discrete Gaussian whenever such an approximation is possible. Fortunately, we note that we only need to be able to approximately match the mean and the variance. Note that as above, the $H=\poly(k/\epsilon)$ components that we are requiring to be either $0$ or $c$, and the one that is a $c$-IRV {can be assumed to have negligible effect on the final mean and variance if we had a sufficiently large $\poly(k/\eps)$ threshold for the variance}.

Let $C$ be the largest multiple of $c$ that is at most $k-1$. Let $\mathcal{S}_{n,k,c}$ be the set of $k$-SIIRVs on $n$ components all of which are constant modulo $c$. For a given $\sigma>\poly(k/\epsilon)$ and $\mu$ we need to determine whether or not there is an element of $\mathcal{S}_{n,k,c}$ whose mean and variance match $\mu$ and $\sigma$ to within $\epsilon \sigma$, and if so to produce one. To do this, we first need a couple of observations about which $\mu,\sigma$ are attainable.

\begin{observation}
For $\p\in \mathcal{S}_{n,k,c}$, $\var_{X \sim \p}[X]< n C^2/4$.
\end{observation}
\begin{proof}
This is because any $k$-IRV that is constant modulo $c$ has a distance of at most $C$ between its minimum and maximum values, and thus has variance at most $C^2/4.$
\end{proof}

\begin{observation} \label{obs:second}
For $\p\in \mathcal{S}_{n,k,c}$ and $X \sim \p$, if $\E[X]\leq nC/2$, then $\var[X] \leq C\E[X]-\E[X]^2/n$.
\end{observation}
\begin{proof}
We note that in the range in question the quantity $C\E[X]-\E[X]^2/n$ is increasing in $\E[X]$, and therefore,
we may show that for any given achievable variance the minimum possible expectation satisfies this inequality.
Note that for the minimum achievable expectation, we may assume that each of the component IRVs is deterministically
$0$ modulo $c$, since otherwise we could subtract a constant from it, which would decrease the expectation and leave the variance unchanged.
The observation now follows given that for any $k$-IRV, $Y$ that {has $\Pr [Y \pmod{c}=0] = 1 $} it holds 
$\var[Y] = \E[Y^2]-\E[Y]^2 \leq C\E[Y]-\E[Y]^2.$
\end{proof}

\begin{observation}
For $\p\in \mathcal{S}_{n,k,c}$ and $X \sim \p$, if $\E[X]\geq n(k-1)-nC/2$, then $\var[X] \leq C(n(k-1)-\E[X])-(n(k-1)-\E[X])^2/n$.
\end{observation}
\begin{proof}
This follows from the previous observation by considering the random variable $n(k-1)-X$.
\end{proof}

We now claim that any pair of expectation and variance {$\mu$ and $\sigma^2$} not disallowed by the above observations
may be approximated by an explicitly computable element of $\mathcal{S}_{n,k,c}$.
Note that, by symmetry, we may assume that $\mu \leq n(k-1)/2.$ If $\mu \geq 2\sigma^2/C$, we may make $\lfloor 4\sigma^2/C^2 \rfloor \leq n$ of our IRVs either $x_i$ or $x_i+C$ with equal probability for some integers $0 \leq x_i \leq k-1$ and all other $X_i$ with $H+1 \leq i \leq n-1$ constant. By adjusting the $x_i$'s and the constants, we can make the expectation of $X$ satisfy $|\E[X]-\mu| \leq 1$ so long as $\mu \geq 2\sigma^2/C$, and the variance
$\var[X]= C^2 \lfloor 4\sigma^2/C^2 \rfloor$ satisfies $|\var[X]-\sigma^2| \leq 1$.

Otherwise, if $\mu \leq 2\sigma^2/C$, let $\sigma^2=C\mu\cdot q$ with $1>q>1/2$. We then use a sum of $k$-IRVs that are $0$ with probability $q$ and $C$ with probability $1-q$, and some $k$-IRVs that are deterministically 0. If we have $a$ many IRVs of the first type, then we get a mean and variance of
$\E[X]=a(1-q)C$ and $\var[X]= aq(1-q)C$. Letting $a$ be approximately $\var[X]/(q(1-q)C)$ completes the argument.
We simply need to verify that in this case $a \leq n$ i.e., that $\sigma^2/(q(1-q)C)\leq n$. Indeed, note that
$$
\var[X]/(q(1-q)C) =\frac{\var[X]}{(\var[X]/(C\E[X]))(1-(\var[X]/(C\E[X])))C} = \frac{C\E[X]^2}{C\E[X]-\var[X]} \leq n
$$
by Observation~\ref{obs:second}.
This shows that given a discrete {Gaussian}, $Z$ so that $cZ$ approximates some element of $\mathcal{S}_{n,k,c}$,
we can efficiently find such an element. In Section~\ref{ssec:reduce}  we gave an appropriately small cover of the set of such Gaussians, which consists of a grid of means and variances of size $O(n)$. It is easy to construct such a grid and by the above, we can construct an $X$ with $|\E[X]-c\mu| \leq \poly(k/\eps)$ and $|\Var[X] - c^2\sigma^2| \leq \poly(k/\eps)$ for each $\mu,\sigma^2$ in the grid that is not disallowed by our observations.
Thus, we can efficiently find a cover of the elements of $\mathcal{S}_{n,k}$ satisfying Case 2 of Theorem \ref{thm:reg}.

\end{document}